%% file: main.tex
\documentclass[12pt]{article}

\usepackage[english]{babel}

\usepackage{amsmath,amsthm,amssymb, graphicx, multicol, array, bm, bbm}
\usepackage{graphicx}
\usepackage[colorlinks=true, allcolors=blue]{hyperref}
\usepackage{caption}
\usepackage{subcaption}
\usepackage{natbib}
\usepackage{algorithm}
\usepackage[utf8]{inputenc}
\usepackage[margin=1in]{geometry}
\usepackage{booktabs}
\usepackage{array}
\usepackage{amsmath, amsfonts, amssymb}
\usepackage{color}
\usepackage{hyperref}
\usepackage{bbm}
\usepackage[noend]{algpseudocode}
\usepackage{subcaption}
\usepackage{graphbox}
\usepackage{accents}
\usepackage{algorithm}
\usepackage{pdflscape}
\usepackage{lscape}
\usepackage{multirow}
\usepackage{cleveref}
\usepackage{authblk}
\usepackage{pifont}
\usepackage{setspace}

\definecolor{darkgreen}{rgb}{0.0, 0.5, 0.0}
\newcommand{\greenup}{\textcolor{darkgreen}{\boldmath$\uparrow$}}
\newcommand{\reddown}{\textcolor{red}{\boldmath$\downarrow$}}

\newcommand{\cmark}{\ding{51}}%
\newcommand{\xmark}{\ding{55}}%

\input{math_command}

\title{Prediction-based Inference in Electronic Health Record (EHR)-linked Biobanks with Clinically Informative Outcomes}

\begin{document}

% \author[]{Draft}
\author[1]{Xingran Chen}
\author[2]{Cheng-Han Yang}
\author[1]{Zhenke Wu}
% \author[1]{Xu Shi}
\author[2]{Bhramar Mukherjee}

\affil[1]{Department of Biostatistics, University of Michigan}
% \affil[12]{Department of Statistics, University of Washington}
\affil[2]{Department of Biostatistics, Yale University}
\affil[ ]{\texttt{\{chenxran,zhenkewu\}@umich.edu}}
% \affil[ ]{\texttt{\{tylermc\}@uw.edu}}
\affil[ ]{\texttt{\{bhramar.mukherjee\}@yale.edu}}
\maketitle

% Version 1.0 
\newpage
\doublespacing

\begin{abstract}

Electronic health record (EHR)-linked biobank data hold tremendous promise for large-scale discoveries via genome-wide association study (GWAS) on diverse phenotypic traits and biomarkers routinely captured in the EHR. However, heterogeneous missingness in biomarkers compromises the validity and efficiency of statistical analyses. Prediction-based (PB) inference methods meet this challenge by using external machine learning (ML) predictions to impute missing biomarker outcomes, thereby improving statistical power and estimation accuracy in association analyses. Yet, their suitability remains unclear when outcomes are subject to clinically informative observation processes, that is, when laboratory tests are ordered based on both measured and unmeasured patient- and health system-level characteristics. In this paper, we review the statistical underpinnings of popular PB methods and then evaluate nine methods, including four PB methods and five traditional missing-data approaches, under an encompassing set of outcome observation processes for continuous and binary outcomes. PB methods can substantially improve statistical power and estimation efficiency when the missing-data mechanism is correctly specified. Under misspecification, however, these gains require both conditional independence between the covariates of interest and the missingness mechanism and independence between imputation error and the missingness mechanism. Using \textit{All of Us} (AoU) data, we perform GWAS of six laboratory biomarkers and demonstrate that PB methods can replicate known genetic associations while improving efficiency relative to (weighted) complete-case analysis (CCA). \textcolor{black}{Their performance in replicating existing GWAS results in AoU also depends} on imputation quality and the underlying missingness mechanism.
\end{abstract}

% \textbf{Keywords:} \textit{All of Us}, Clinically informative observation process, Electronic health records, Genome-wide association studies, Missing data, Missing not at random.

% , to accurately predict the missing values within the dataset and carry out statistical analysis.

% \tableofcontents

\newpage

\section{Introduction}
\label{sec:intro}

Electronic health record (EHR)-linked biobank data can accelerate biomedical research by enabling large-scale, cost-effective studies of diseases, genetics, and treatment outcomes~\citep{abul2019personalized}.
Such databases link biological samples with detailed, longitudinal EHR and allow researchers to study disease progression, drug responses, and genetic markers \textcolor{black}{in large national biobanks such as US-based \textit{All of Us} (AoU)~\citep{all2019all} and the UK Biobank~\citep{sudlow2015uk}.} However, a central analytic challenge in fully leveraging these databases for scientific insights is that outcomes of interest, particularly laboratory biomarkers, are often only partially observed. Unlike clinical trials, which follow a pre-specified protocol that prescribes the measurement timing, instrument, quality check and control, laboratory tests are typically ordered based on clinical indication, \textcolor{black}{and consequently} different biomarkers exhibit distinct frequencies and patterns of missingness in the EHR data. 
For example, \textcolor{black}{in the laboratory biomarker genome-wide association studies (GWAS) using the AoU dataset}, \textcolor{black}{46,226 (96.3\%) of the 47,970 participants} had at least one recorded white blood cell (WBC) count, a biomarker routinely measured during annual examinations. In contrast, for C-reactive protein (CRP), which is typically ordered only when inflammation is suspected, \textcolor{black}{only 25,455 (26.2\%) of the 97,210 participants had at least one measurement} \textcolor{black}{(See Section~\ref{sec:data_preparation} for details on the GWAS cohort construction process, which leads to differences in sample sizes across biomarkers).}
Without rigorously considering and addressing this outcome missingness issue, obtaining valid and efficient inference for these phenotypes in large-scale biomedical studies, such as GWAS~\citep{tam2019benefits,uffelmann2021genome}, is challenging~\citep{beesley2020emerging, zhou2022global}.

With the rapid development of predictive models based on machine learning (ML)~\citep{ho1995random,cortes1995support}, deep learning (DL)~\citep{krizhevsky2012imagenet,he2016deep,vaswani2017attention,ren2025comprehensive}, \textcolor{black}{and, more recently, large language models~\citep{nazi2024large,li2024scoping}}, highly accurate predictions are increasingly available for imputing missing entries in rectangular analytic datasets and \textcolor{black}{even for complex} unstructured multimodal data. This dovetails with the rich \textcolor{black}{and diverse expanse of} health records that may provide abundant features for predicting \textcolor{black}{unobserved values for multiple} laboratory biomarkers~\citep{goldstein2016opportunities,sharma2024advances}. \textcolor{black}{However, these ML/DL-based predictions do not fit the standard multiple imputation (MI) procedure~\citep{rubin1987multiple,rubin1996multiple,murray2018multiple,little2019statistical}, as analysts may only have access to a single predicted value for each missing cell rather than multiple draws from a predictive distribution, making MI not directly applicable. See detailed discussion in Section~\ref{sec:prediction-based-inference}.} As such, a growing body of work has developed prediction-based (PB) inference methods~\citep{chen2000unified, angelopoulos2023prediction,angelopoulos2023ppi++,mccaw2023leveraging,mccaw2024synthetic,gronsbell2024another,miao2025assumption,miao2024valid, gan2024prediction, kluger2025prediction, chen2025unified, zhao2025imputation, xu2025blockwise} to leverage ML prediction without introducing bias from imputed data distribution \textcolor{black}{and accounting for uncertainty in an appropriate way}. In the context of EHR-linked biobanks \textcolor{black}{where the outcome lab biomarker is partially observed, the PB inference considers splitting the available analytic data into two subsets: a} labeled subset consists of patients for whom the laboratory test was ordered and recorded, \textcolor{black}{and an} unlabeled subset consists of patients for whom only ML-predicted laboratory results are available. Existing PB methods differ in how they incorporate predictions: some augment the estimating equation~\citep{angelopoulos2023prediction,angelopoulos2023ppi++,miao2025assumption,gan2024prediction, zhao2025imputation, xu2025blockwise}, while others augment the estimator directly~\citep{chen2000unified, gronsbell2024another, kluger2025prediction, chen2025unified}. As an exception, SynSurr~\citep{mccaw2024synthetic} formulates a joint model for the observed and the imputed outcomes. Recent extensions further incorporate inverse probability weighting to accommodate the missingness-at-random mechanism~\citep{kluger2025prediction, chen2025unified}. We provide a detailed structured review and categorization of these methods in Section~\ref{sec:method}.

While PB methods avoid assumptions about the prediction models, they  rely on assumptions about the missingness mechanism, namely, all existing PB methods require missing completely at random (MCAR) or, at most, missing at random (MAR). In practice, however, EHR data are collected under clinically informative observation processes~\citep{du2024new, yang2026joint}, so that the resulting missingness may depend on the unobserved outcome itself, a mechanism known as missing not at random (MNAR). For example, a clinician is more likely to order a Hemoglobin A1c (HbA1c) test for a patient suspected of having diabetes or pre-diabetes, \textcolor{black}{in which case patients with elevated HbA1c values are disproportionately likely to be tested.}
% in which case the HbA1c level is \textcolor{black}{more likely to be high or low}. 
In such a setting, the probability of observing the outcome depends on its value, violating the MCAR and MAR assumptions on which PB methods rely. It remains unclear whether PB methods provide \textcolor{black}{greater statistical power and better efficiency compared to complete-case analysis (CCA) while maintaining valid type I error control} when these assumptions are violated. To date, comprehensive empirical evaluations of PB methods under realistic EHR missingness mechanisms are lacking, as are theoretical conditions characterizing when each approach remains valid.

To address this knowledge gap, this paper makes the following contributions in the context of cross-sectional \textcolor{black}{association} analyses with partially observed outcomes \textcolor{black}{and completely observed covariate data. This \textcolor{black}{setting} is natural when the predictors are genetic markers, because biobanks have genetic data on a large fraction of participants.} First, \textcolor{black}{we provide a structured review of recent PB inference \textcolor{black}{methods, selected from the extant literature}. We categorize \textcolor{black}{these} PB methods into three \textcolor{black}{classes}} based on how they incorporate predictions, and summarize their \textcolor{black}{assumptions regarding underlying missingness mechanisms} (Section~\ref{sec:method}; Table~\ref{tab:summary_literature}). Second, we conduct a comparative simulation study evaluating
\textcolor{black}{nine methods, including \textcolor{black}{four of the reviewed PB methods} and five traditional methods based on CCA and multiple imputation (MI),} under 10 outcome observation process models spanning MCAR, MAR, and MNAR, for both continuous and binary outcomes (Section~\ref{sec:data_simulation_settings}). The simulation results (Sections~\ref{sec:simulation-results} and~\ref{sec:binary-outcome-results}) reveal that PB methods \textcolor{black}{achieve higher statistical power} \textcolor{black}{and more precise estimation} when their \textcolor{black}{assumed missingness mechanisms} hold, but can exhibit inflated type I error under MNAR even in settings where CCA remains valid. Third, we provide theoretical results that explain these patterns \textcolor{black}{revealed in the simulations}; these results are synthesized into a decision tree (Figure~\ref{fig:decision_tree}) that provides practitioners \textcolor{black}{with a structured summary} on when valid inference is achievable using each approach. Finally, we apply CCA, weighted CCA, and two PB methods to conduct GWAS of six laboratory biomarkers using AoU data, comparing results with external GWAS reference summary statistics (Section~\ref{sec:gwas}). The analysis demonstrates that PB methods replicate known genetic associations with improved efficiency compared to CCA, and, importantly, extend inference to a more representative study population by incorporating individuals with missing outcomes. \textcolor{black}{In addition, it is observed that the performance} of PB methods in GWAS depend on the imputation quality and the underlying \textcolor{black}{missingness} mechanisms.

The remainder of the paper is organized as follows. In Section~\ref{sec:method}, we introduce the problem setup, review existing approaches for missing data analysis, and provide a structured review of PB inference methods developed so far. The simulation design is introduced in Section~\ref{sec:data_simulation_settings}. In Section~\ref{sec:simulation-results}, we present the simulation results and provide theoretical explanations for the findings. In Section~\ref{sec:gwas}, we conduct GWAS of six laboratory biomarkers using AoU data to perform further empirical comparison. In Section~\ref{sec:discussion}, we discuss the results from both the simulations and the analysis of AoU data, and conclude the paper with the limitations and potential directions for future research.

\section{Method}\label{sec:method}

\subsection{Problem setup}

\textbf{Inferential target.} \textcolor{black}{We consider cross-sectional regression settings in EHR-linked biobank data. Specifically, we use the first recorded value of each laboratory biomarker for each individual to summarize longitudinal laboratory measurements in the EHR-linked biobank data \textcolor{black}{, an approach typically used for large-scale association analysis~\citep{goldstein2020labwas}}. This cross-sectional formulation underlies large-scale analyses such as \textcolor{black}{GWAS and} LabWAS~\citep{goldstein2020labwas}}, \textcolor{black}{in which the data may be summarized using median or mean, while the challenges remain similar.}. \textcolor{black}{We consider a regression model relating a biomarker of interest $Y$ to a $p$-dimensional vector of covariates of primary interest $\vX \in \R^p$ and $q$ additional confounders $\vZ \in \R^q$. In GWAS, $\vX$ is typically a scalar variable ($p = 1$) representing one of millions of genetic variants, and the regression model is repeatedly fitted for these variants one at a time. The confounders $\vZ$ generally include age, sex, and genetic principal components (PCs), which are commonly adjusted for in GWAS. The confounders $\vZ$ are included in the regression to remove spurious associations between $\vX$ and $Y$. The primary scientific interest of GWAS lies in the adjusted association between $Y$ and $\vX$. With this notation, for each genotype being tested we consider the underlying full dataset with sample size $N$: $\cD_{\text{Full}} = \{(Y_i, \vX_i, \vZ_i)\}^{N}_{i=1}$.}

When $Y_i$ is continuous, we adopt a linear regression model parameterized by $\vbeta = (\beta_0, \vbeta^\top_{\vX}, \vbeta^\top_{\vZ})^\top$, with systematic component: $\E[Y_i| \vX_i, \vZ_i] = \beta_0 + \vX_i\vbeta_{\vX} + \vZ_i\vbeta_{\vZ}$, where $\beta_0$ is the intercept term, $\vbeta_{\vX}$ is the parameters of interest, and $\vbeta_{\vZ}$ is the main effect of $\vZ_i$. Denote $\vG_i = (1, \vX_i, \vZ_i)$, when the full data $\cD_{\text{Full}}$ are available, one may estimate $\vbeta$ by solving $\widehat{\vbeta}_{\text{Full}}$ from
\begin{equation}
\label{eq:linear_ee}
\sum^N_{i=1} \vG_i(Y_i - \vG_i\vbeta) = 0.
\end{equation}
If $Y_i$ is binary, we have the systematic component: $\text{logit}(\E[Y_i| \vX_i, \vZ_i]) = \beta_0 + \vX_i\vbeta_{\vX} + \vZ_i\vbeta_{\vZ}$, where $\text{logit}(x) = \ln{(x / (1-x))}$. In this setting, the full-data estimator $\widehat{\vbeta}_{\text{Full}}$ solves
\begin{equation}
\label{eq:logistic_ee}
\sum^N_{i=1} \vG_i\{Y_i - \text{expit}(\vG_i\vbeta)\} = 0,
\end{equation}
where $\text{expit}(x) = e^x / (1 + e^x)$. Under standard regularity conditions, these estimators are consistent and asymptotically normal. We refer the reader to~\citet{van2000asymptotic} for further details.

\noindent \textbf{Missing data mechanisms.} \textcolor{black}{For large-scale biomedical studies such as GWAS \textcolor{black}{and} LabWAS~\citep{goldstein2020labwas}, it is more important to consider missingness in $Y$ as genotype data \textcolor{black}{($\vX$) and other covariates ($\vZ$) such as age, gender, and genetic PCs} are typically more complete. Therefore, in this work we consider an outcome-only missingness setting, where} not all individuals have their laboratory biomarker $Y_i$ recorded in the EHR data. We introduce the missingness indicator $R_i \in \{0,1\}$, where $R_i = 1$ if $Y_i$ is observed and $R_i = 0$ otherwise. We assume that $R_i$ follows an outcome observation model $\pi(R_i; \vomega, \vX_i, \vZ_i, Y_i)$ parameterized by $\vomega$ and potentially dependent on $(\vX_i, \vZ_i, Y_i)$. Following~\citet{little2019statistical} and~\citet{rubin1976inference}, the missingness mechanism can be classified as follows. Under MCAR, missingness is independent of all variables, i.e., $\pi(R_i; \vomega, \vX_i, \vZ_i, Y_i) = \pi(R_i; \vomega)$. For example, a patient’s WBC counts may be missing because a blood sample was accidentally damaged in the laboratory; since this event is unrelated to patient characteristics or outcomes, the mechanism is MCAR. Under MAR, missingness depends only on observed variables, i.e., $\pi(R_i; \vomega, \vX_i, \vZ_i, Y_i) = \pi(R_i; \vomega, \vX_i, \vZ_i)$. For instance, whether a laboratory test is ordered may be determined by information from the last visit, as well as age or sex available in the EHR. In this case, the probability of missingness depends on observed data. Finally, the MNAR assumption indicates that the missingness indicator $R_i$ depends on the unobserved outcome $Y_i$ itself. \textcolor{black}{In EHR-linked biobank data, lab biomarkers are collected under the clinically informative observation process~\citep{du2024new, yang2026joint}.} For example, a clinician is more likely to order an HbA1c test when the patient is suspected of having diabetes, so that the probability of observation depends directly on the (potentially unobserved) outcome value. As a result, the probability of missingness depends directly on the missing outcome value.

\subsection{\textcolor{black}{Traditional} missing-data approaches}

In this subsection, \textcolor{black}{we provide a recap of  traditional} approaches for statistical analysis in the presence of missing data.

\noindent \textbf{(Weighted) complete-case analysis.} Given a dataset with partially observed outcome biomarkers, the simplest and most widely used strategy is CCA, which discards all observations with missing biomarkers and proceeds using only the subset of fully observed cases. In the regression settings, CCA yields valid statistical inference as long as the missing data mechanism is MCAR or MAR. Even when the mechanism is MNAR, \citet{kundu2024framework} showed that CCA can achieve valid type I error control in logistic regression when the covariates of interest \textcolor{black}{$\vX$ and the confounders $\vZ$ are independent of the missingness mechanism.} The primary drawback of CCA is its incapability to incorporate information from incomplete cases: when 90\% of subjects have missing $Y$, CCA uses only the remaining 10\% \textcolor{black}{of fully observed subjects, discarding substantial information from partially observed individuals.} Furthermore, under MAR or MNAR, the inferential target shifts to the complete-case subpopulation, which may not represent the broader biobank population. To address this, one can estimate the probability of observing the outcome and weight each complete case by the inverse of this probability, a strategy known as inverse probability weighting (IPW)~\citep{li2013weighting,sun2018inverse,sun2018inverse_epid,ross2023accounting}. For linear regression, such a weighted CCA estimator is obtained by solving $\widehat{\vbeta}$ from $\sum^N_{i=1} w_i\vG_i(Y_i - \vG_i\vbeta) = 0$, where $w_i$ is the only additional term compared with Equation~(\ref{eq:linear_ee}), and $w_i = 1 / \widehat{\pi}(R_i = 1| \vX_i, \vZ_i)$ is the inverse of the estimated probability that the 
$i$-th outcome is observed.

\noindent \textbf{Multiple imputation.} An alternative to CCA is to impute missing entries and carry out inference using the imputed dataset. MI, proposed by~\citet{rubin1987multiple}, accounts for imputation uncertainty by performing inference across multiply imputed datasets. Here we briefly present the MI procedure in our context; a
more thorough overview is available in~\citet{murray2018multiple}. First, the predictive distribution of the outcome biomarkers $Y$ given $\vX$ and $\vZ$ is modeled. Then, the missing values of $Y$ are imputed from this estimated distribution $K$ times to create $K$ completed datasets. The regression model is fitted to each completed dataset, yielding $(\widehat{\vbeta}^{(1)}, \cdots, \widehat{\vbeta}^{(K)})$, and corresponding variance estimates $(\widehat{U}^{(1)}, \cdots, \widehat{U}^{(K)})$. The MI estimator $\widehat{\vbeta}_{\text{MI}}$ and its estimated sampling variance $\widehat{U}_{\text{MI}}$ are obtained via Rubin's rule:
$$\widehat{\vbeta}_{\text{MI}} = \frac{1}{K} \sum^K_{k=1} \widehat{\vbeta}^{(k)}, \qquad \widehat{U}_{\text{MI}} = \bar{\widehat{U}} + \left ( 1 + \frac{1}{K} \right ) B,$$
where 
$$\bar{\widehat{U}} =  \frac{1}{K}\sum^K_{k=1} \widehat{U}^{(k)}$$
and
$$B = \frac{1}{K - 1}\sum^K_{k=1}\left ( \widehat{\vbeta}^{(k)} - \widehat{\vbeta}_{\text{MI}} \right )\left ( \widehat{\vbeta}^{(k)} - \widehat{\vbeta}_{\text{MI}} \right )^\top$$
are the within-imputation and the between-imputation variance, respectively. Under large-sample theory, the MI estimator follows a $t$-distribution, based on which hypothesis testing and confidence intervals can be derived~\citep{rubin1987multiple}. MI imposes assumptions on the imputation model; improper imputations may lead to anti-conservative inferential results~\citep{rubin1996multiple,murray2018multiple}. \textcolor{black}{A thorough discussion of CCA, IPW, and MI for missing data can be found in~\citet{little2024comparison}.}

\subsection{Prediction-based inference approach}
\label{sec:prediction-based-inference}

With the rapid development of predictive models based on ML and DL methods, highly accurate predictions for missing biomarkers have become increasingly available for the EHR-linked biobank data. When ML imputations are available, the setting differs subtly from multiple imputation~\citep{chen2025unified}: in MI, the imputation model is constructed solely from the observed analytic data. In the context of ML-based imputation, it is commonly assumed that analysts can access predictions from large-scale models trained on external data sources (e.g., large language models). In the regression setting, this implies that, in addition to the partially observed biomarkers, analysts have access to predicted biomarker values $\{\widehat{Y}_i\}_{i=1}^N$ for each individual, regardless of whether the true biomarker is observed. Direct application of the MI procedure is infeasible in this setting. First, the ML/DL models may not satisfy the assumptions of the imputation models in MI, as they are too complicated to manipulate and may lie outside the analysts' control. Second, in many applications, only a single prediction is available per individual, making it impossible to apply Rubin's rule.

Consequently, a growing body of literature has proposed PB methods that incorporate ML-based imputation without imposing assumptions on the imputation model itself. \textcolor{black}{Table~\ref{tab:summary_literature} summarizes recent work in this area. Specifically, the studies are summarized in terms of publication year, accommodated missingness patterns, and assumed missing data mechanisms. Most methods~\citep{angelopoulos2023prediction, angelopoulos2023ppi++, gan2024prediction, chen2000unified, gronsbell2024another, mccaw2024synthetic} focus on the missing-outcome scenario, in which only the outcome variable is subject to missingness. Several other works address either covariate missingness~\citep{miao2025assumption, kluger2025prediction} or more general missingness patterns~\citep{chen2025unified, zhao2025imputation, xu2025blockwise}. In terms of the assumed missing data mechanism, most studies consider MCAR, whereas~\citet{kluger2025prediction, chen2025unified} assume MAR, under which the probability of missingness depends on observed data. To elucidate the core idea of these methods, we first illustrate their use through a mean estimation example, and then describe how different PB methods are formalized, leading to a natural categorization.}

Consider the problem of estimating the population mean HbA1c level. Suppose the data consist of two subsets: 1) the labeled subset $\cD_{\text{lab}} = \{(Y^{\text{lab}}_i, \widehat{Y}^{\text{lab}}_i)\}^{m}_{i=1}$, for which both the true HbA1c levels $Y^{\text{lab}}_i$ and the predicted HbA1c levels $\widehat{Y}^{\text{lab}}_i$ are available. In practice, such predictions can be made using glucose levels. 2) the unlabeled subset $\cD_{\text{unlab}} = \{(\widehat{Y}^{\text{unlab}}_i)\}^{M}_{i=1}$, for which only the predicted HbA1c levels $\widehat{Y}^{\text{unlab}}_i$ are available. The total sample size is therefore $N = M + m$. Under the assumption of MCAR, the mean estimation from the PB methods can be written as follows:
\begin{equation}
\widehat{\mu}_{\text{PB}} = \frac{1}{M} \sum^{M}_{i=1}  \widehat{Y}^{\text{unlab}}_i - \left (\frac{1}{m} \sum^{m}_{i=1}  \widehat{Y}^{\text{lab}}_i - \frac{1}{m} \sum^{m}_{i=1}  Y^{\text{lab}}_i \right ).
\end{equation}

There are two complementary perspectives for interpreting this estimator. First, following~\citet{angelopoulos2023prediction}, the first term, $ \sum^{M}_{i=1}  \widehat{Y}^{\text{unlab}}_i / M$, can be viewed as a potentially biased estimate of the mean HbA1c level based on predicted HbA1c values in the unlabeled subset. The second term, $ \sum^{m}_{i=1}  (\widehat{Y}^{\text{lab}}_i -  Y^{\text{lab}}_i) / m$, serves as a rectifier that corrects this bias by estimating it using the true and predicted HbA1c levels in the labeled subset.

Alternatively, by rearranging the terms, we have:
\begin{equation}
\widehat{\mu}_{\text{PB}} = \frac{1}{m} \sum^{m}_{i=1}  Y^{\text{lab}}_i - \left ( \frac{1}{m} \sum^{m}_{i=1}\widehat{Y}^{\text{lab}}_i -  \frac{1}{M} \sum^{M}_{i=1}\widehat{Y}^{\text{unlab}}_i \right ).
\end{equation}
In this form, $\widehat{\mu}_{\text{PB}}$ can be expressed as the difference between a complete-case (CC) mean estimator $\widehat{\mu}_{\text{CC}} = \sum_{i=1}^{m} Y^{\text{lab}}_i / m$, which is unbiased and consistent, and an estimate of zero $\sum^{m}_{i=1}\widehat{Y}^{\text{lab}}_i / m - \sum^{M}_{i=1}\widehat{Y}^{\text{unlab}}_i / M$, which is defined as the difference between the average predicted HbA1c levels in the labeled and unlabeled subsets.

The goal of PB inference is to incorporate additional information in the predicted outcomes to reduce the estimator’s variance. Since the labeled subset and the unlabeled subset are independent, the covariance between the two subsets is zero. Thus, the variance of $\widehat{\mu}_{\text{PB}}$ is:
\begin{equation}
\Var(\widehat{\mu}_{\text{PB}}) = \frac{1}{M}\Var(\widehat{Y}^{\text{unlab}}_i) + \frac{1}{m}\Var\bigl(Y^{\text{lab}}_i - \widehat{Y}^{\text{lab}}_i\bigr).
\end{equation}
Mathematically, $\widehat{\mu}_{\text{PB}}$ tends to have a small variance when the $M \gg m$, and when $\Var\bigl(Y^{\text{lab}}_i - \widehat{Y}^{\text{lab}}_i\bigr)$ is small. Intuitively, this suggests that PB methods \textcolor{black}{can yield substantial efficiency gains in estimation and greater statistical power in hypothesis testing} \textcolor{black}{compared to CCA} when the proportion of individuals with recorded HbA1c levels is small, and when highly accurate predictions of HbA1c are available.

Building upon this estimator, recent work~\citep{angelopoulos2023ppi++} introduces a tuning parameter $\theta$ and considers the family:
\begin{equation}
\widehat{\mu}_{\text{PB}}(\theta) = \frac{1}{M} \sum^{M}_{i=1}  \widehat{Y}^{\text{unlab}}_i - \theta \left (\frac{1}{m} \sum^{m}_{i=1}  \widehat{Y}^{\text{lab}}_i - \frac{1}{m} \sum^{m}_{i=1}  Y^{\text{lab}}_i \right ).
\end{equation}
It has been proved in~\citet{angelopoulos2023ppi++} that a data-driven choice $\widehat{\theta} = M / (M + m)$ guarantees that $\widehat{\mu}_{\text{PB}}(M/(M + m))$ \textcolor{black}{to be at least as efficient as} $\widehat{\mu}_{\text{CC}}$.

More broadly, the PB inference literature typically focuses on the framework of convex M-estimation or, equivalently, Z-estimation. Mathematically, when the full data are observed, a Z-estimator $\widehat{\vbeta}_{\text{Full}}$ is defined as a solution to the estimating equation $\sum^N_{i=1} \vpsi(\vG_i, Y_i; \vbeta) = 0$, where $\vpsi$ is a $d$-dimensional estimating function. In particular, Equations~(\ref{eq:linear_ee}) and~(\ref{eq:logistic_ee}) correspond to the estimating equations for linear and logistic regression, respectively. When the outcome is partially observed, methods differ in how they incorporate the imputed outcomes $\widehat{Y}_i$. Excluding SynSurr~\citep{mccaw2024synthetic}, most existing methods fall into the following two categories.

\noindent \textbf{Augmenting information at the estimating-equation level.} These methods formulate their estimators by an estimating equation of the form
\begin{align}
&\frac{\sum^N_{i=1} R_i\vpsi(\vG_i, Y_i; \vbeta)}{\sum R_i}  
\nonumber - \vTheta \left (  
\frac{\sum^N_{i=1} R_i\vpsi(\vG_i, \widehat{Y}_i; \vbeta)}{\sum R_i}  
-  \frac{\sum^N_{i=1} (1 - R_i)\vpsi(\vG_i, \widehat{Y}_i; \vbeta)}{N - \sum R_i} 
\right )  
= \vzero \\ \label{eq:estimating_equation_level}
\end{align}

The first term corresponds to the CCA estimating equation, whereas the second term provides an unbiased estimate of zero. A tuning parameter $\vTheta \in \R^{d \times d}$ is introduced so that efficiency gains over CCA can be achieved when it is appropriately chosen. Existing methods in this category may differ from Equation~(\ref{eq:estimating_equation_level}) by treating the missingness indicator $R_i$ as fixed (i.e., not accounting for randomness in the missingness mechanism). They may also specify the tuning parameter $\vTheta$ differently. In particular, \citet{angelopoulos2023prediction, angelopoulos2023ppi++, miao2025assumption, gan2024prediction} can be viewed as special cases of this formulation. PPI++~\citep{angelopoulos2023ppi++} specifies the tuning parameter as $\vTheta = \theta \mathbf{I}$, where $\theta$ is a scalar tuning parameter and $\mathbf{I}$ is an identity matrix; when $\theta = 1$, this reduces to PPI~\citep{angelopoulos2023prediction}. PSPA~\citep{miao2025assumption} is another special case, with $\vTheta$ being restricted as a diagonal matrix.

\noindent \textbf{Augmenting information at the estimator level.} These methods define their estimators in the form of:
\begin{equation}
\label{eq:estimator_level}
\widehat{\vbeta}_{\text{PB}} = \widehat{\vbeta}_{\text{CC}} - \vTheta (\widehat{\vgamma}_{1} - \widehat{\vgamma}_2),
\end{equation}
where $\widehat{\vbeta}_{\text{CC}}$ is obtained by solving $\sum^N_{i=1} w_i R_i\vpsi(\vG_i, Y_i; \vbeta) = 0$, while $\widehat{\vgamma}_{1}$ and $\widehat{\vgamma}_2$ are obtained by solving $\sum^N_{i=1} w_i R_i\vpsi(\vG_i, \widehat{Y}_i; \vgamma) = 0$ and $\sum^N_{i=1} w_i (1 - R_i)\vpsi(\vG_i, \widehat{Y}_i; \vgamma) = 0$, respectively. $w_i$ denotes inverse probability weights used to account for the MAR assumption~\citep{chen2025unified,kluger2025prediction}. The estimator proposed in \citet{gronsbell2024another, chen2000unified} for linear regression is a special case of Equation~(\ref{eq:estimator_level}) with $w_i = 1$, as the missingness mechanism is assumed to be MCAR. In addition, the Predict-Then-Debias  (PTD) estimator~\citep{kluger2025prediction} and the Pattern-Stratified PPI (PS-PPI) estimator~\citep{chen2025unified} are also special cases, with $w_i$ equal to either the inverse sampling probability or the inverse probability of being observed. Because these methods share the same structure, it has been proved by~\citet{chen2025unified} that their frameworks are mathematically equivalent.

\noindent \textbf{Joint likelihood.} SynSurr~\citep{mccaw2024synthetic} does not fall into either category above and is exclusively designed for linear regression. Rather, it assumes the following \textcolor{black}{joint distribution between the true values $Y_i$ and the synthetic surrogates $\widehat{Y}_i$}:
\begin{equation}
\label{eq:synsurr_framework}
\begin{bmatrix}
Y \\ \widehat{Y}
\end{bmatrix} \Big| \vG
=
\begin{bmatrix}
\vG & 0 \\
0 & \vG
\end{bmatrix}
\begin{bmatrix}
\vbeta \\ \valpha
\end{bmatrix}
+
\begin{bmatrix}
\varepsilon_T \\ \varepsilon_S
\end{bmatrix},
\end{equation}
where the residuals $(\varepsilon_T, \varepsilon_S)^\top$ follow a bivariate normal distribution $\cN(\vzero, \Sigma_\varepsilon)$. According to~\cite{mccaw2024synthetic}, the computational procedure involves two regressions: 1) regressing $Y$ on $\vG$ to obtain $\widehat{\valpha}$; 2) regressing $\widehat{Y}$ on $(Y, \vG)$ to obtain $(\widehat{\delta}, \widehat{\vgamma})$, where $\widehat{\delta}$ and $\widehat{\vgamma}$ are coefficients for $Y$ and $\vG$, respectively. The resulting SynSurr estimator is $\widehat{\vbeta} = \widehat{\vgamma} + \widehat{\delta}\widehat{\valpha}$. A Wald test can be derived from the assumed joint association model in Equation~(\ref{eq:synsurr_framework}).

\section{Simulation Settings}
\label{sec:data_simulation_settings}

\subsection{Methods included in the simulation}\label{sec:methods_selected}

We evaluate PB methods alongside CC and MI methods under different outcome observation models. The methods are grouped as follows.

\noindent \textbf{Complete-case approaches.} \textbf{CCA} is included to fit the regression model using only individuals with observed outcomes. CCA weighted by the inverse of the estimated probability of being observed, denoted as \textbf{WCCA (est)}, is also considered, where inverse probability weights are obtained from a propensity score model fitted using both $\vX$ and $\vZ$. This represents the best practically achievable propensity score model. When the true outcome observation model depends on $Y_i$, such propensity score models are misspecified.

\noindent \textbf{Multiple imputation approaches.} Predictive mean matching (PMM)~\citep{rubin1986statistical,little1988missing} and random forest imputation~\citep{ho1995random} are used as predictive models. The missing outcomes are multiply imputed using these models, and inference is conducted based on Rubin’s rules. We use the implementation of MI provided by the \texttt{MICE 3.16.0}\footnote{\href{https://github.com/amices/mice}{https://github.com/amices/mice}} package in \texttt{R}~\citep{van2011mice}. MI with PMM is denoted as \textbf{MI}, while MI with random forest is denoted as \textbf{MI-RF}.

\noindent \textbf{PB approaches.} A naive imputation approach is first included, in which missing outcomes are directly imputed using machine learning predictions and regression is performed on the \textcolor{black}{expanded} completed dataset without bias correction; this approach is denoted as \textbf{Naive}. PB methods that augment information at different levels are included. Specifically, we include \textbf{PPI} and \textbf{PPI\texttt{++}}~\citep{angelopoulos2023prediction, angelopoulos2023ppi++}, which augment information at the estimating equation level, \textbf{PS-PPI}~\citep{chen2025unified}, which augments information at the estimator level. Finally, \textbf{SynSurr}~\citep{mccaw2024synthetic} is included as a joint association modeling method.

In addition to the methods from the above three approaches, we include an infeasible \textcolor{black}{but gold standard} full-data analysis, which assumes all outcomes are observed, denoted as \textbf{Full}, for reference.

%\subsection{Data simulation settings} 
\subsection{Data generating process}
% \label{sec:data_simulation_settings}

%\paragraph{Data generating process} 
We consider three simulation settings. The first is a linear regression setting in which each $Y_i$ is generated from
\begin{equation}
    \label{eq:linear_regression_cts}
    Y_{i} = \beta_0 + \beta_1X_{i1} + \beta_2X_{i2} + \beta_3 Z_{i1} + \beta_4 Z_{i2} + \epsilon_i,
\end{equation}
where
$$X_{i1} = \cos(Z_{i1}) + \tau_i, \; X_{i2} = \sin(Z_{i2}) + \nu_i,$$
and \textcolor{black}{$Z_{i1}$ and $Z_{i2}$ are independent and identically distributed as univariate normal distributions with mean zero and marginal variance $\sigma_z$, respectively.} The error terms are independently distributed as $\epsilon_i \sim \cN(0, \sigma^2)$, $\tau_i \sim \cN(0, \sigma_\tau^2)$, and $\nu_i \sim {\tt Exponential}(\lambda)$. \textcolor{black}{The scientific interest is to obtain consistent estimates and valid type I error controls of the coefficients $\beta_1$ and $\beta_2$ corresponding to the covariates of interest $X_1$ and $X_2$, respectively}; $Z_{1}$ and $Z_{2}$ serve as confounders.

Second, we consider another linear regression setting in which the confounders are dummy variables derived from a single categorical variable. Note that we do not include multiple categorical variables in this setting.
\begin{equation}
    \label{eq:linear_regression_dummy}
    Y_{i} = \beta_0 + \beta_1X_{i1} + \beta_2X_{i2} + \beta_3 Z_{i3} + \beta_4 Z_{i4} + \epsilon_i.
\end{equation}
Specifically, $Z_{i3} = I(Z_{i1} < 0) I(Z_{i2} < 0)$ and $Z_{i4} = I(Z_{i1} > 0)I(Z_{i2} > 0)$ are two dummy variables derived from a categorical variable indicating the quadrant in which $(Z_{i1}, Z_{i2})$ falls. We consider two distinct covariate settings in the linear regression simulations because the included methods exhibit different performance across these settings (see Section~\ref{sec:simulation-results}).

Finally, we consider a logistic regression setting, where $Y_i$ is binary and its conditional expectation satisfies
\begin{equation}
    \label{eq:logistic_regression}
    \text{logit}\{\E[Y_{i}|\vX_i, \vZ_i]\} = \beta_0 + \beta_1X_{i1} + \beta_2X_{i2} + \beta_3 Z_{i1} + \beta_4 Z_{i2},
\end{equation}
where the covariates are generated in the same way as above.

To construct synthetic surrogates for the PB methods, in the linear regression setting, we design the predicted outcome as
$$\widehat{Y}_i = Y_i + 2\sin(X_{i1}^2 + X_{i2}^3 + Z_{i1}^2 + Z_{i2}^2).$$
This construction induces covariate-dependent noise in the true outcome. \textcolor{black}{We defer a detailed discussion of how such covariate-dependent imputation errors influence type I error control and the consistency of PB estimators to Section~\ref{sec:theoretical_explanation}.} Synthetic surrogate design for logistic regression is detailed in Section~\ref{sec:synthetic_surrogates_logistic_regression} of the Supplementary Materials. The impact of imputation quality is examined in Section~\ref{sec:simulation_quality_of_imputation}.

\subsection{Observation mechanism for $Y$}

%\paragraph{Observation Mechanism for outcome $Y$} As mentioned, we assume that 

We assume only the outcome variable $Y_i$ is subject to missingness. To simulate the outcome observation process under different missing data mechanisms, we construct a series of outcome observation models for $R_i$. All models use a logistic specification; we present representative examples here and provide the full configurations in Section~\ref{sec:supplementary-setting-details} of the Supplementary Materials.

Under MCAR, the probability of observing $Y_i$ is constant:
\begin{equation*}
\Pr(R_i = 1) = \omega. \tag{MCAR}
\end{equation*}
Under MAR, missingness depends on observed covariates and/or confounders. For example:
\begin{align*}
&\text{logit}\{\Pr(R_i = 1)\} = \omega_0 + \omega_1 Z_{i2}, \tag{MAR1} \\
&\text{logit}\{\Pr(R_i = 1)\} = \omega_0 + \omega_1 X_{i1} + \omega_2 Z_{i2}. \tag{MAR2}
\end{align*}
In MAR1, the missingness depends solely on $Z_2$, whereas in MAR2, it depends on both $Z_2$ and the covariate of interest $X_1$.
%\begin{align*}
%& R_i \sim Z_{i2}, \tag{MAR1} \\
%& R_i \sim X_{i1} + Z_{i2}. \tag{MAR2}
% &\text{logit}\{\Pr(R_i = 1|\vZ_i)\} = \omega_0 + \omega_1 Z_{i2}, \tag{MAR1} \\
% &\text{logit}\{\Pr(R_i = 1|\vX_i, \vZ_i)\} = \omega_0 + \omega_1 X_{i1} + \omega_2 Z_{i2}. \tag{MAR2}
%\end{align*}
Under MNAR, we allow $R_i$ to depend on the unobserved outcome $Y$, with progressively richer dependence structures:
\begin{align*}
&\text{logit}\{\Pr(R_i = 1)\} = \omega_0 + \omega_1 Z_{i2} + \omega_2 g(Y_i), \tag{MNAR1} \\
&\text{logit}\{\Pr(R_i = 1)\} = \omega_0 + \omega_1 X_{i1} + \omega_2 Z_{i2} + \omega_3 g(Y_i), \tag{MNAR2} \\
&\text{logit}\{\Pr(R_i = 1)\} = \omega_0 + \omega_1 X_{i2} + \omega_2 Z_{i2} + \omega_3 g(Y_i), \tag{MNAR3} \\
&\text{logit}\{\Pr(R_i = 1)\} = \omega_0 + \omega_1 X_{i1} + \omega_2 X_{i2} + \omega_3 Z_{i2} + \omega_4 g(Y_i). \tag{MNAR4}
\end{align*}
Specifically, in MNAR1 the missingness depends on $Z_2$ and $Y$. In MNAR2 and MNAR3, $X_1$ and $X_2$ are added to the model, respectively, and both are included in MNAR4. \textcolor{black}{We introduce $g(Y_i)$ to accommodate two implementations used in our simulations. In the linear regression settings, we take $g$ to be an indicator function, for example, $g(Y_i) = I(Y_i < y_0)$ for a prespecified $y_0$. In the logistic regression setting, we take $g$ to be the identity function, i.e., $g(Y_i) = Y_i$. We refer the reader to Section~\ref{sec:supplementary-setting-details} for additional details.} 
%Finally, to simulate the MNAR mechanism, we allow $R_i$ to depend on the unobserved outcome $Y_i$. We first specify a model where missingness depends on both $Z_{2}$ and $Y$:
%\begin{align*}
%& R_i \sim Z_{i2} + Y_i, \tag{MNAR1}
% &\text{logit}\{\Pr(R_i = 1|\vZ_i, Y_i)\} = \omega_0 + \omega_1 Z_{i2} + \omega_2 Y_i. \tag{MNAR1}
%\end{align*}
%We then allow missingness to also depend on covariates of interest, i.e., $X_{1}$ and/or $X_{2}$:
%\begin{align*}
%& R_i \sim X_{i1} + Z_{i2} + Y_i, \tag{MNAR2} \\
%& R_i \sim X_{i2} + Z_{i2} + Y_i. \tag{MNAR3} \\
%& R_i \sim X_{i1} + X_{i2} + Z_{i2} + Y_i. \tag{MNAR4}
% &\text{logit}\{\Pr(R_i = 1|\vX_i, \vZ_i, Y_i)\} = \omega_0 + \omega_1X_{i1} + \omega_2Z_{i2} + \omega_3Y_i, \tag{MNAR2} \\
% &\text{logit}\{\Pr(R_i = 1|\vX_i, \vZ_i, Y_i)\} = \omega + 2.0X_{i2} + 0.8Z_{i2} + 0.5Y_i, \tag{MNAR3} \\
% &\text{logit}\{\Pr(R_i = 1|\vX_i, \vZ_i, Y_i)\} = \omega + 1.5X_{i1} + 1.5X_{i2} + 1.2Z_{i2} + 0.5Y_i. \tag{MNAR4}
%\end{align*}
Building upon MNAR4, we additionally introduce interaction effects between $Y$ and the covariates/confounders:
\begin{align*}
\text{logit}\{\Pr(R_i = 1)\} 
&= \omega_0 + \omega_1 X_{i1} + \omega_2 X_{i2} + \omega_3 Z_{i2} + \omega_4 g(Y_i) + \omega_5 X_{i1}g(Y_i), \tag{MNAR5} \\[6pt]
\text{logit}\{\Pr(R_i = 1)\} 
&= \omega_0 + \omega_1 X_{i1} + \omega_2 X_{i2} + \omega_3 Z_{i2} + \omega_4 g(Y_i) + \omega_5 X_{i2}g(Y_i), \tag{MNAR6} \\[6pt]
\text{logit}\{\Pr(R_i = 1)\} 
&= \omega_0 + \omega_1 X_{i1} + \omega_2 X_{i2} + \omega_3 Z_{i2} + \omega_4 g(Y_i) + \omega_5 Z_{i2}g(Y_i). \tag{MNAR7}
\end{align*}
Here, interaction terms between $X_1$ and $Y$, and between $X_2$ and $Y$, are included in MNAR5 and MNAR6, respectively. In MNAR7, an interaction term between $Z_2$ and $Y$ is included.

For each subject, we sample $R_i$ from one of these observation models and mask the corresponding $Y_i$ when $R_i = 0$. Directed acyclic graphs (DAGs) illustrating the full data-generating process under each outcome observation model are presented in Figure~\ref{fig:selection-model-dag}.

\begin{figure}[t]
\centering
\includegraphics[width=1\linewidth]{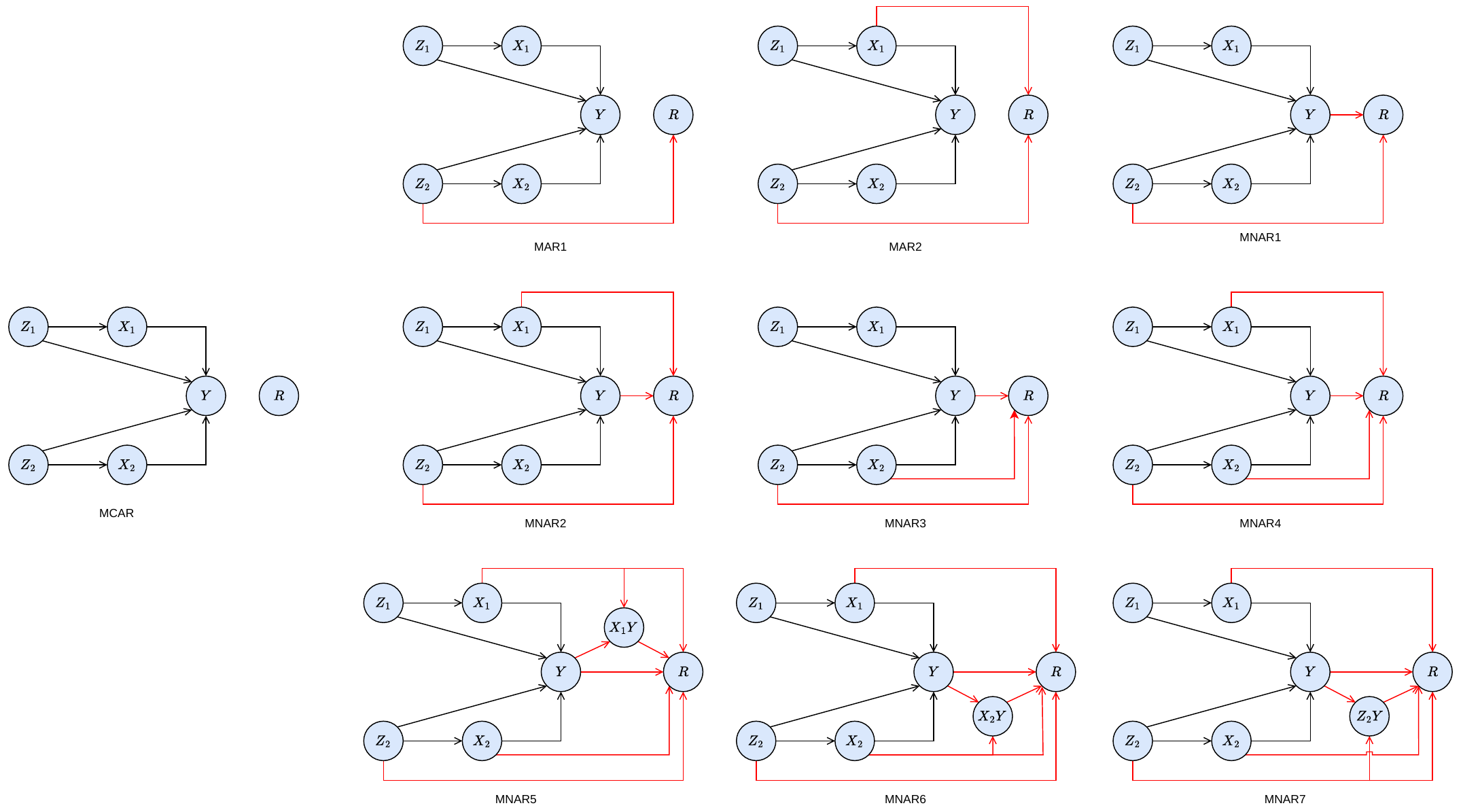}
\caption{Directed acyclic graphs (DAGs) of the full data-generating process under different outcome observation models. Red arrows indicate the missingness mechanism (i.e., which variables the missingness indicator depends on). Note that under the linear regression setting with a categorical confounder, $Z_1$ and $Z_2$ influence $Y$ via intermediate variables $Z_3$ and $Z_4$.}
\label{fig:selection-model-dag}
\end{figure}

\subsection{Evaluation metrics} We evaluate the performance of all included methods for estimating $\beta_1$ and $\beta_2$ using the following metrics: 1) bias, 2) mean squared error (MSE), 3) type I error, and 4) power. We take $\beta_1$ as an example to illustrate how these evaluation metrics are defined. Suppose that $L$ simulation replicates are generated, and in the $j$-th replicate let $\widehat{\beta}^{(j)}_1$ denote the point estimate of $\beta_1$, and let $p^{(j)}$ denote the $p$-value for testing $H_0: \beta_1 = 0$. We compute the type I error, i.e., the probability of rejecting $H_0$ when the true parameter satisfies $\beta_1 = 0$:
$$ \texttt{Type I Error} = \frac{1}{L} \sum^L_{j=1} I\left ( p^{(j)} < 0.05 \right ).$$
We follow~\citet{zhang1994adjusted} and use adjusted power. Specifically, we use the empirical 5th percentile of the $p$-values under the null setting (i.e., when the true value $\beta_1 = 0$), denoted by $p_{0.05}$, as a calibrated significance threshold. The adjusted power is then
$$ \texttt{Power} = \frac{1}{L} \sum^L_{j=1} I\left ( p^{(j)} < p_{0.05} \right ).$$
To evaluate point-estimation performance, we compute bias and MSE as
$$ \texttt{Bias} = \left \lvert \frac{1}{L} \sum^L_{j=1} \left ( \widehat{\beta}^{(j)}_1 - \beta_1 \right ) \right \rvert, \; \texttt{MSE} = \frac{1}{L} \sum^L_{j=1} \left ( \widehat{\beta}^{(j)}_1 - \beta_1 \right )^2.$$
Similar evaluation metrics are defined for $\beta_2$.

\section{Simulation Results}
\label{sec:simulation-results}

Table~\ref{tab:linear_regression_summary_continuous} and~\ref{tab:linear_regression_summary_dummy} present the linear regression simulation results of the included methods under different missingness mechanisms. \textcolor{black}{Results} on type I error and power under varying imputation quality are presented in Figures~\ref{fig:linear_regression_simulation_varying_quality_power} and~\ref{fig:linear_regression_simulation_varying_quality_type_i_error} for the continuous-confounder setting, and in Figures~\ref{fig:linear_regression_simulation_binary_z_varying_quality_power} and~\ref{fig:linear_regression_simulation_binary_z_varying_quality_type_i_error} for the categorical-confounder setting. Logistic regression simulation results are presented in Section~\ref{sec:binary-outcome-results} of the Supplementary Materials.

\subsection{Missing completely at random}

\noindent \textbf{Hypothesis testing.} Under MCAR, in both the continuous- and categorical-confounder settings \textcolor{black}{(Tables~\ref{tab:linear_regression_summary_continuous} and~\ref{tab:linear_regression_summary_dummy})}, most methods (e.g., CCA, WCCA (est), PPI, PPI++, SynSurr, and PS-PPI) achieve type I error rates close to the nominal 0.05 level when testing both $\beta_1$ and $\beta_2$. MI and MI-RF exhibit conservative type I error under the categorical-confounder setting, although MI-RF shows moderately inflated type I error of 0.20 and 0.19 for $\beta_1$ and $\beta_2$, respectively, when the confounders are continuous \textcolor{black}{(Table~\ref{tab:linear_regression_summary_continuous})}. In contrast, the Naive method shows severely inflated type I error (1.0 for both coefficients, regardless of \textcolor{black}{confounder types}). Among methods with valid type I error control, \textcolor{black}{according to Table~\ref{tab:linear_regression_summary_continuous}}, the PB methods (except PPI) show comparable performance, with power of about 0.97 for $\beta_1$ and around 0.86 for $\beta_2$ when the confounders are continuous; similar findings hold when the confounder is a single categorical variable \textcolor{black}{(Table~\ref{tab:linear_regression_summary_dummy})}. Overall, the PB approaches exhibit higher power than CCA and WCCA (est), while MI is uniformly less powerful than these approaches. If one ignores type I error inflation, the Naive estimator attains power close to 1.0 for both $\beta_1$ and $\beta_2$ across the two confounder settings.

\noindent \textbf{Point estimation} For point estimation, across both continuous and categorical confounders \textcolor{black}{(Tables~\ref{tab:linear_regression_summary_continuous} and~\ref{tab:linear_regression_summary_dummy}, respectively)}, the PB methods (except PPI) achieve the best overall MSE performance (around 0.07--0.13) in the MCAR setting, smaller than that of CCA, WCCA (est), and Naive. MI-RF shows the smallest MSE among all methods when the $\beta$s are under the null and the confounders are dummy variables derived from a single categorical variable. Overall, however, the two MI methods have larger MSE than the PB approaches and CCA/WCCA (est).

\subsection{Missing at random}

\noindent \textbf{Hypothesis testing.} Under MAR1 (missingness depends on $Z_2$) and MAR2 (missingness depends on both $X_1$ and $Z_2$), type I error patterns differ between the continuous- and categorical-confounder settings. Specifically, CCA and WCCA (est) maintain valid type I error control across scenarios, while the Naive method again shows anti-conservative type I error. MI and MI-RF also exhibit liberal type I error (around 0.07--0.08 and 0.25--0.49, respectively) when the confounders are continuous \textcolor{black}{(Table~\ref{tab:linear_regression_summary_continuous})}; under the categorical-confounder setting \textcolor{black}{(Table~\ref{tab:linear_regression_summary_dummy})}, they exhibit conservative type I error (0.01--0.02). Among PB inference approaches, PS-PPI continues to achieve valid type I error control because it incorporates propensity score models to account for the missingness mechanism. The remaining methods (PPI, PPI++, and SynSurr) show inflated type I error, except for $\beta_2$ under MAR1 when the confounders are continuous and for $\beta_1$ under MAR1 in the categorical-confounder setting. Among methods with valid type I error control, the PB methods achieve higher power for both $\beta_1$ and $\beta_2$ across both confounder settings when compared to CCA and WCCA.

\noindent \textbf{Point estimation.} Under MAR, PS-PPI is overall competitive in terms of MSE across both confounder settings. For continuous confounders, \textcolor{black}{from Table~\ref{tab:linear_regression_summary_continuous},} PS-PPI attains MSEs of 0.12 and 0.26 for $(\beta_1,\beta_2)$ under MAR1 and 0.20 and 0.36 under MAR2. For the categorical-confounder setting \textcolor{black}{(Table~\ref{tab:linear_regression_summary_dummy})}, PS-PPI attains MSEs of 0.07 and 0.14 under MAR1 and 0.13 and 0.21 under MAR2. When the true coefficients are under the null, MI-RF yields the smallest MSE$_0$ for the categorical-confounder setting (Table~\ref{tab:linear_regression_summary_dummy}), although its type I error is uniformly conservative.

\subsection{Missing not at random}

\noindent \textbf{Hypothesis testing.} When missingness depends on $Y$, type I error patterns differ between the two confounder settings. Specifically, for continuous confounders, \textcolor{black}{as shown in Table~\ref{tab:linear_regression_summary_continuous},} all included methods exhibit inflated type I error for both $\beta_1$ and $\beta_2$ in most missingness mechanisms settings. Even when some entries for a given method appear close to 0.05 for a particular MNAR mechanism or coefficient, supplementary results in Figure~\ref{fig:linear_regression_simulation_varying_quality_type_i_error} under alternative ML imputation designs show inflated type I error. For the categorical-confounder setting \textcolor{black}{according to Table~\ref{tab:linear_regression_summary_dummy}}, CCA, WCCA (est), and the two multiple-imputation methods achieve valid or conservative type I error control when the corresponding covariate of interest is not included in the outcome-observation model. For instance, under MNAR1 (missingness depends on $Z_2$ and $Y$ but not on $X_1$ or $X_2$), CCA and WCCA (est) have type I error close to 0.05 for both $\beta_1$ and $\beta_2$, while the two MI methods are conservative, with type I error close to 0.01. When missingness depends only on $X_1$ (MNAR2) or $X_2$ (MNAR3), these methods exhibit inflated type I error for the corresponding coefficient ($\beta_1$ for MNAR2; $\beta_2$ for MNAR3). By contrast, type I error for the other coefficient ($\beta_2$ for MNAR2; $\beta_1$ for MNAR3) show results similar to MNAR1. Finally, when both $X_1$ and $X_2$ are included in the observation model (MNAR4--MNAR7), they exhibit inflated type I error for both coefficients. On the other hand, \textcolor{black}{from Table~\ref{tab:linear_regression_summary_dummy},} the PB approaches show inflated type I error in nearly all MNAR scenarios under the categorical-confounder setting.

\noindent \textbf{Point estimation.} In terms of the point estimation, under the continuous confounder settings \textcolor{black}{(Table~\ref{tab:linear_regression_summary_continuous})}, both bias and MSE for $\beta_1$ and $\beta_2$ increase significantly for all the missing data methods when comparing to those from the MCAR and MAR settings. Under the categorical-confounder setting \textcolor{black}{(Table~\ref{tab:linear_regression_summary_dummy})}, for many settings where valid or approximately valid type I error controls are achieved for $\beta_1$ or $\beta_2$, the estimation bias is significantly higher when the coefficient is not under the null, compared to when the coefficient is under the null. For example, from \textcolor{black}{Table~\ref{tab:linear_regression_summary_dummy}}, under MNAR1, where CCA, WCCA (est) achieve valid type I error controls for both the coefficients, the estimation bias for CCA is 0.38 for $\beta_1$ and 0.13 for $\beta_2$ when they are under the null; those estimation bias for WCCA (est) are 0.44 for $\beta_1$, and 0.05 for $\beta_2$, respectively. Whereas, when $\beta_1 = 0.1$ and $\beta_2 = 0.1$, the estimation bias for CCA increases to 0.59 and 0.98; for WCCA (est) the numbers become 0.47 and 0.85. Similar observations are found for the two MI methods. This indicates that the validity of these methods may depend on whether the true parameters are under the null.

\subsection{Theoretical explanation of simulation results}
\label{sec:theoretical_explanation}

We now provide theoretical results that explain the type I error \textcolor{black}{and bias} patterns observed in the simulations. The key question is: \textbf{under what conditions do CCA and PB methods yield valid inference \textcolor{black}{and consistent estimation} when the missingness mechanism is MNAR?}

The simulation results show that under MCAR and MAR, CCA yields consistent estimation in both linear and logistic regression, consistent with well-established theoretical results~\citep{little2019statistical}. The theorems below characterize the MNAR case.

\begin{theorem}
\label{thrm:linear_regression_cca}
Suppose the true model is $Y_i = \vX_i\vbeta_{\vX}^\ast + \vZ_i\vbeta_{\vZ}^\ast + \epsilon_i, \epsilon_i \sim \cN(0, \sigma^2)$; $\vbeta_{\vX}^\ast$ is the parameter of interest. Let $R_i$ be the missingness indicator for $Y_i$. Consider the complete-case analysis (CCA) estimator $\widehat{\vbeta}_{\vX, \text{CC}}$ obtained by regressing $Y_i$ on $(\vX_i, \vZ_i)$ using only samples with $R_i = 1$:
\begin{enumerate}
    \item If the outcome $Y_i$ is at least missing at random, i.e., $Y_i \perp R_i \mid (\vX_i, \vZ_i)$, then $\widehat{\vbeta}_{\vX, \text{CC}}$ is consistent and asymptotically normal with mean $\vbeta_{\vX}^\ast$.
    
    \item If the outcome $Y_i$ is missing not at random, i.e., $Y_i \not\perp R_i \mid (\vX_i, \vZ_i)$, then $\widehat{\vbeta}_{\vX, \text{CC}}$ is consistent and asymptotically normal with mean $\vbeta_{\vX}^\ast$ under the following sufficient conditions, all of which must hold:
    \begin{enumerate}
        \item $\vbeta_{\vX}^\ast = \vzero$;
        \item $R_i$ does not depend on $\vX_i$ once $Y_i$ and $\vZ_i$ are given, i.e., $\vX_i \perp R_i \mid (Y_i, \vZ_i)$;
        \item $\vZ_i$ is a multi-dimensional one-hot encoded vector, i.e., only one component in $\vZ_i$ can be 1 and all the others equal 0.
    \end{enumerate}
\end{enumerate}
\end{theorem}

Here, we provide further explanation of Theorem~\ref{thrm:linear_regression_cca} \textcolor{black}{when the outcome is MNAR}. \textcolor{black}{First, condition 2(a) and 2(b) require the covariate of interest $\vX_i$ be independent of the missingness indicator $R_i$ conditional on the outcome and confounders, and that $\vbeta_{\vX}^\ast = \vzero$. This indicates that when the conditional independence assumption between $\vX_i$ and $R_i$ holds, CCA can achieve valid type I error control and consistent estimation under the null, but it may be biased when the true effect is nonzero. This is confirmed empirically in our simulations for MNAR1--MNAR3, where $X_{i1}$ and/or $X_{i2}$ are excluded from the outcome observation model.} In particular, Tables~\ref{tab:linear_regression_summary_continuous} and~\ref{tab:linear_regression_summary_dummy} show that when the true parameters deviate from the null, CCA and WCCA produce substantially more biased estimates under \textcolor{black}{these outcome observation models}, even when their type I error remains near 0.05. \textcolor{black}{Second, Condition 2(c) requires the confounder to be a multidimensional one-hot encoded vector. This condition is satisfied either when a single categorical variable is included as a confounder and encoded using a one-hot representation, or when multiple categorical variables are treated as confounders but only their interaction terms are included in the model. If the main effects of more than one categorical variable are included, then $\vZ_i$ is generally not a one-hot encoded vector.} When confounders are continuous (Equation~(\ref{eq:linear_regression_cts})), the residual function $r(\vX, \vZ) = \E[Y \mid \vX, \vZ, R = 1] - \vG\vbeta^\ast$ is generally nonlinear in $\vZ$ and cannot be accommodated within a linear model. Consequently, $\widehat{\vbeta}_{\vX}$ is biased even when conditions 2(a) and 2(b) hold. This explains the inflated type I error of CCA and WCCA under MNAR in the continuous-confounder setting.

For logistic regression, \textcolor{black}{the sufficient conditions required to achieve valid type I error control are less restrictive}. Building on the core idea of~\citet{kundu2024framework}, we establish the following result.

\begin{theorem}
\label{thrm:logistic_regression_cca}
Suppose the true model is $\text{logit}(\E(Y_i \mid \vX_i, \vZ_i)) = \beta^\ast_0 + \vX_i\vbeta_{\vX}^\ast + \vZ_i\vbeta_{\vZ}^\ast$, \textcolor{black}{$Y_i$ is a binary outcome,} and $\vbeta_{\vX}^\ast$ is the parameter of interest. Let $R_i$ be the missingness indicator for $Y_i$. Consider the complete-case analysis (CCA) estimator $\widehat{\vbeta}_{\vX, \text{CC}}$ obtained by regressing $Y_i$ on $(\vX_i, \vZ_i)$ using only samples with $R_i = 1$:
\begin{enumerate}
\item If the outcome $Y_i$ is at least missing at random, i.e., $Y_i \perp R_i \mid (\vX_i, \vZ_i)$, then $\widehat{\vbeta}_{\vX, \text{CC}}$ is consistent and asymptotically normal with mean $\vbeta_{\vX}^\ast$.

\item If the outcome $Y_i$ is missing not at random, i.e., $Y_i \not\perp R_i \mid (\vX_i, \vZ_i)$, then $\widehat{\vbeta}_{\vX, \text{CC}}$ is consistent and asymptotically normal with mean $\vbeta_{\vX}^\ast$ when $\vX_i \perp R_i \mid (Y_i, \vZ_i)$ \textcolor{black}{and $\vZ_i \perp R_i \mid Y_i$.}
\end{enumerate}
\end{theorem}

In contrast to Theorem~\ref{thrm:linear_regression_cca}, Theorem~\ref{thrm:logistic_regression_cca} does not require $\vbeta_{\vX}^\ast = \vzero$ or the confounder to be a single categorical variable. The key difference arises because the logistic model absorbs the MNAR-induced bias as an additive term $\log r(\vX, \vZ)$ on the logit scale; when $r(\vX, \vZ)$ is a constant of $\vX$ and $\vZ$, this term is absorbed by the intercept without affecting $\vbeta_{\vX}^\ast$ \textcolor{black}{when $\vX_i$ and $\vZ_i$ are independent}~\citep{kundu2024framework}. This explains why \textcolor{black}{CC approaches} maintain valid type I errors for $\beta_1$ under MNAR1 and MNAR3 in the logistic regression simulations (Section~\ref{sec:binary-outcome-results}), \textcolor{black}{where both $X_1$ and $Z_1$, the parent node of $X_1$, do not enter the outcome observation model.} \textcolor{black}{On the other hand, although $X_2$ does not enter the outcome observation model in MNAR1 and MNAR2, the CC approaches exhibit inflated type I error since $Z_2$, an ancestor node of $X_2$, enters the outcome observation model.}

Finally, we present Theorem~\ref{theorem:linear_regression_pb} to characterize when estimator-level PB methods remain valid under MNAR in linear regression.

%Regarding the PB approaches, we have observed from Table~\ref{tab:linear_regression_summary_continuous} and~\ref{tab:linear_regression_summary_dummy} that they achieve inflated type I error controls even when CCA and WCCA (est) remain valid. However, from Figure~\ref{fig:linear_regression_simulation_varying_quality_type_i_error} and~\ref{fig:linear_regression_simulation_binary_z_varying_quality_type_i_error} it can be seen that PB approaches can stay consistent in terms of valid type I error controls with CCA and WCCA (est) when the design of the synthetic surrogates change. In Theorem~\ref{theorem:linear_regression_pb}, we provide a theoretical explanation for estimator level PB methods. The proof can be found in Section~\ref{sec:proof_theorem_linear_regression_pb} of the Supplementary Materials. Specifically, in addition to the sufficient conditions of CC approaches, one additional condition regarding the machine learning imputation is required in order to guarantee the consistency of the PB methods under either MAR or MNAR.

\begin{theorem}
\label{theorem:linear_regression_pb}
In the data setting of Theorem~\ref{thrm:linear_regression_cca}, the estimator-level PB inference estimator $\widehat{\vbeta}_{\text{PB}}$ in the form of Equation~(\ref{eq:estimator_level}) are consistent and asymptotically normal with mean $\vbeta^\ast$ when the sufficient conditions of $\widehat{\vbeta}_{\text{CC}}$, and $\widehat{Y}_i$ can be decomposed into $\widehat{Y}_i = \vG_i\tilde{\vbeta} + \tau_{i,\text{ML}}$, where $\tau_{i,\text{ML}}$ is independent of $R_i$.
\end{theorem}

Theorem~\ref{theorem:linear_regression_pb} shows that PB methods require an additional condition beyond what CCA needs: the ML imputation error $\tau_{i,\text{ML}}$ must be independent of the missingness indicator $R_i$. This condition ensures that the auxiliary term $\widehat{\vgamma}_1 - \widehat{\vgamma}_2$ in Equation~(\ref{eq:estimator_level}) converges to zero. When $\tau_{i,\text{ML}}$ depends on $R_i$, the labeled- and unlabeled-subset regressions on $\widehat{Y}_i$ converge to different limits, inducing bias. This explains why PB methods can exhibit inflated type I error under MNAR even when CCA remains valid: the synthetic surrogate $\widehat{Y}_i = Y_i + 2\sin(X_{i1}^2 + X_{i2}^3 + Z_{i1}^2 + Z_{i2}^2)$ used in our main simulations has an imputation error that depends on covariates entering the observation model, thereby violating the required independence condition. Supplementary Figures~\ref{fig:linear_regression_simulation_varying_quality_type_i_error} and~\ref{fig:linear_regression_simulation_binary_z_varying_quality_type_i_error} confirm that alternative surrogate designs satisfying the independence condition restore valid type I error for the PB methods.

In practice, the independence condition $\tau_{i,\text{ML}} \perp R_i$ is difficult to assess because the missingness mechanism is typically unknown. \textcolor{black}{However, relative to CCA, PB methods rely on a stronger set of conditions: CCA requires conditions only on the analysis and outcome observation models, whereas PB methods also require a condition on the ML imputation model. As such, when valid type I error control is the primary concern, CCA may be the more reliable choice; PB methods are more attractive under MCAR or MAR, where they can provide greater estimation efficiency and greater statistical power.}

\section{Lab Biomarker GWAS in \textit{All of Us}}
\label{sec:gwas}

To further evaluate the performance of existing PB methods in large-scale EHR-linked biobank analyses, we apply these methods to conduct GWAS of 6 laboratory biomarkers and compare the results with publicly available GWAS summary statistics. Specifically, we perform GWAS of WBC count, HbA1c, alanine aminotransferase (ALT), albumin, CRP, and vitamin D levels using data from the NIH \textit{All of Us} Research Program (AoU)~\citep{all2019all}, a large-scale database containing medical records of patients across the US and their biospecimens. The Curated Data Repository version 8 of the data is used in this study.

\subsection{\textcolor{black}{Data preparation}}
\label{sec:data_preparation}
\noindent \textbf{\textcolor{black}{Analysis data preparation.}} To perform the GWAS, we first identify 310,465 individuals \textcolor{black}{from the AoU data} who are older than 18 years, male or female sex assigned at birth, have both EHR data and short-read whole-genome sequencing (WGS) available, and have at least one measurement of body mass index (BMI), weight, and height. \textcolor{black}{As discussed below, BMI, weight, and height are used to predict missing biomarkers and are therefore required to be fully observed.} Among these individuals, we retain those of European genetic ancestry with available socioeconomic status (SES) data and at least one Systematized Nomenclature of Medicine (SNOMED)-coded diagnosis in their records, yielding 150,429 individuals. Trait-specific analytic cohorts are constructed from these individuals. Specifically, for each trait, we first assign a predictive biomarker to enhance the performance of the ML models: neutrophils count for WBC, glucose level for HbA1c, aspartate aminotransferase (AST) for ALT, CRP for albumin, albumin for CRP, and calcium for vitamin D. As a result, for each trait, we restrict the analytic cohort to individuals with at least one recorded measurement of the corresponding predictive biomarkers. The resulting sample sizes for analytic cohorts in each trait are summarized in Table~\ref{tab:gwas_summary}. For each outcome, we use the first recorded measurement in a patient’s medical record as the outcome value. We apply an inverse normal transformation to each outcome. Detailed information on the selected outcomes and other biomarkers used in this work is summarized in Table~\ref{tab:aou_trait_units_fields}. We adjust the GWAS models for age, sex, and the first five genetic principal components (PCs).

\noindent \textbf{\textcolor{black}{ML model-fitting data preparation.}} To generate ML-based imputations to be used in the PB methods, we \textcolor{black}{first identify} 382,475 individuals \textcolor{black}{in the AoU data} who meet the same eligibility criteria as the analytic cohorts, except that the WGS data are not required. \textcolor{black}{Note that we will exclude individuals already included in the analytic cohort before fitting ML models for each trait to ensure that the analytic cohort and the ML model-fitting cohort are independent.} Similarly, only individuals with available SES data and at least one SNOMED-coded diagnosis in their records, resulting in 257,832 individuals. Note that we do not filter the ML model-fitting cohort by genetic ancestry.  We extract their diagnosis records coded by SNOMED and use them as predictors. Specifically, each individual's diagnosis history is encoded as a high-dimensional binary vector, with each dimension representing a SNOMED-coded disease and a value of 1 indicating that the person has been diagnosed with this disease in the medical history (0 otherwise). We retain only conditions with prevalence greater than 5\%, resulting in 160 SNOMED-coded diseases. Additional predictors include age, sex, race, ethnicity, first recorded height, weight, and BMI, and neighborhood SES variables including median income and the proportion of residents without health insurance. For each trait, we also include its corresponding predictive biomarker as an additional predictor. The sample size of each ML model-fitting set are summarized in Table~\ref{tab:gwas_summary}. For each trait, we randomly split the ML model-fitting dataset into training and validation sets in a 4:1 ratio. Random forest, gradient boosting regression tree, histogram-based gradient boosting regression tree, and XGBoost~\citep{chen2016xgboost} are fitted, and the best model measured by $R^2$ on the validation set is selected to generate imputations for outcomes in the analytic cohort. Figure~\ref{fig:gwas-pipeline} summarizes the preprocessing pipelines for both the analytic and ML training data.

\noindent \textbf{\textcolor{black}{Genotype data preparation.}} In terms of the genotyping array data, we only keep genetic variants with at least one non-reference call, call rate $> 95\%$, minor allele frequency (MAF) $> 0.5\%$, and Hardy–Weinberg equilibrium $p \geq 1 \times 10^{-8}$. These criteria yield 158,388 variants across all traits. Within each trait-specific analytic cohort, we further exclude variants with MAF less than 1\%. 

\subsection{\textcolor{black}{Evaluation settings}}

We include CCA, WCCA, PS-PPI, and SynSurr, which cover both CC and PB approaches, with and without inverse probability weighting to account for the MAR mechanism, in our evaluation. To obtain inverse probability weights for WCCA and PS-PPI, we fit logistic regression propensity score models including the same predictors used in the ML models, except for the SNOMED-coded diagnosis variables. We evaluate the included methods from two perspectives. First, we compare SynSurr with CCA, and compare PS-PPI with WCCA. Specifically, we evaluate: 1) the number of genome-wide significant (GWS) single-nucleotide polymorphisms (SNPs), 2) the difference in $\beta$ estimates for each genetic variant, and 3) the ratio of standard errors. Second, we compare the GWAS results from each method with external reference GWAS summary statistics. For HbA1c and CRP, we use~\citet{loya2025scalable} as the reference; for ALT, WBC, and albumin, we use~\citet{verma2024diversity}; and for vitamin D, we use~\citet{manousaki2020genome}. We match the GWAS results between reference GWAS summary statistics and results from the included methods using \texttt{pyranges}~\citep{stovner2020pyranges} based on the criteria that the two genetic variants are within 250,000 base pairs. We then report the number of rediscovered GWS SNPs for each method.

\subsection{\textcolor{black}{GWAS results and comparison}} 

\noindent \textbf{\textcolor{black}{Descriptive analysis results.}} Table~\ref{tab:summary_statistics_traits} summarizes the characteristics of the analytic samples for each trait. Note that the International Classification of Diseases (ICD) codes (both versions 9 and 10) of the common phenotypes presented can be found in Table~\ref{tab:phenotype_icd_prefix}. It can be seen that, first, \textcolor{black}{the proportion of missingness differs across biomarker traits. Specifically, WBC has the lowest proportion of missingness, with 46,226 of 47,970 (96.3\%) participants having observed values. In contrast, CRP and vitamin D have among the highest proportions of missingness, with only 25,455 of 97,210 (26.2\%) and 27,172 of 109,951 (24.7\%) participants having observed values, respectively. These results indicate that, under clinically informative observation processes~\citep{du2024new,yang2026joint}, missingness patterns vary substantially across biomarkers.} Second, within each trait, population characteristics differ between the complete-case subset and the full analytic sample. More importantly, the complete-case subsets tend to represent individuals with poorer health status compared to the broader analytic samples used by PB methods, and this discrepancy becomes more pronounced as the proportion of missing outcomes increases. For example, for WBC, where the incomplete sample size (1,744) is substantially smaller than the complete-case subset (46,226), disease prevalences are similar between the two samples: in the complete-case subset, prevalences of depression, hypertension, diabetes, anxiety, and obesity are 38.1\%, 52.5\%, 20.7\%, 42.2\%, and 38.3\%, respectively, compared with 37.7\%, 52.5\%, 20.6\%, 41.9\%, and 38.1\% in the PB analytic sample. In contrast, for HbA1c, where the incomplete-case subset exceeds the complete-case subset, the complete-case subset exhibits remarkably higher prevalence of obesity-related conditions, including hypertension (63.8\%), diabetes (32.0\%), and obesity (46.7\%), compared with the full analytic sample (51.4\%, 19.8\%, and 35.5\%, respectively). Similar patterns are observed for mental health diagnoses. These findings suggest that restricting analysis to complete cases may introduce selection bias by disproportionately focusing on individuals with poorer health status. On the other hand, PB methods incorporate individuals with missing target traits, thereby enabling inference over a more representative population.

\noindent \textbf{\textcolor{black}{GWAS results.}} Table~\ref{tab:gwas_summary} and Figures~\ref{fig:beta-diff-synsurr-cca}–\ref{fig:se-ratio-psppi-wcca} present the GWAS results for all traits. It can be observed that, first, for ALT, PB methods identify more GWS variants (11 for SynSurr and 6 for PS-PPI) than CCA (6 GWS variants) and WCCA (2 GWS variants), likely due to the fact of relatively better quality of synthetic surrogates ($R^2 = 0.579$). For the remaining traits, the number of significant SNPs is similar across methods, \textcolor{black}{which may be due to the underlying MNAR mechanism in the data.} Second, as shown in Figures~\ref{fig:beta-diff-synsurr-cca} and~\ref{fig:beta-diff-psppi-wcca}, differences in $\widehat{\beta}$ between SynSurr and CCA, as well as PS-PPI and WCCA, are centered near zero for most traits. Exceptions arise for ALT and vitamin D when comparing PS-PPI to WCCA. One possible explanation is that the missingness mechanism for these traits is MNAR. As shown in Table~\ref{tab:linear_regression_summary_continuous}, under MNAR settings, the biases of the WCCA and PS-PPI estimates differ significantly in magnitude. Third, as shown in Figures~\ref{fig:se-ratio-synsurr-cca} and~\ref{fig:se-ratio-psppi-wcca}, SynSurr and PS-PPI consistently improve estimation efficiency relative to CCA and WCCA, respectively, although the magnitude of improvement varies across traits. In particular, SynSurr yields the largest efficiency gains for HbA1c, whereas improvements are more marginal for ALT, WBC, and albumin, where the incomplete-case subsets are significantly smaller than the complete-case subsets. PS-PPI similarly provides greater efficiency gains for HbA1c and ALT. Similar findings can be found within GWS variants. Fourth, when compared with external reference GWAS summary statistics, most GWS variants identified by the included methods are rediscovered signals from existing reports. Notable exceptions include vitamin D, where PS-PPI identifies 32 GWS variants absent from the reference study, and HbA1c, for which all evaluated methods identify additional GWS variants relative to the reference results.
 
\section{Discussion}
\label{sec:discussion}

This paper presents a statistical review of PB methods and a comparative evaluation of these methods with existing classical approaches using both simulated and AoU data. From the simulation, PB methods \textcolor{black}{offer increased statistical power and better estimation efficiency relative} to CCA and WCCA when the missing data mechanism assumptions are satisfied. Otherwise, they require more restrictive conditions compared to CC approaches to achieve valid type I error controls in certain regression scenarios. In GWAS applications using AoU data, PB methods improve estimation efficiency relative to CCA and WCCA while replicating previously reported GWS variants. More importantly, by incorporating individuals with missing traits, PB methods extend inference beyond the complete-case subset to a more representative study population. 

\textcolor{black}{More specifically, the evaluation provides the following insights.} First, in simulation studies across an encompassing set of outcome observation process models, the Naive method, which directly imputes missing cells, consistently fails to control type I error even under MCAR, for both continuous and binary outcomes. \textcolor{black}{This finding suggests that direct imputation without debiasing as is required in PB methods should not be used for inference.} In addition, when the missing data mechanism assumptions hold, both CC and PB approaches achieve valid type I error controls, and PB methods improve statistical power compared to CCA and WCCA, indicating that incorporating machine learning predictions can enhance \textcolor{black}{statistical power in hypothesis testing and improve estimation efficiency.}

% Specifically, we prove that under MNAR in linear regression, CCA can maintain valid type I error for testing a covariate's effect when that covariate does not enter the observation model and confounders are categorical (Theorem~\ref{thrm:linear_regression_cca}); PB methods require an additional decomposability condition on the ML predictions, namely, that the imputation error is independent of the missingness indicator, to achieve the same validity (Theorem~\ref{theorem:linear_regression_pb}). For logistic regression, CCA remains valid under MNAR whenever the covariate of interest is conditionally independent of missingness given the outcome and confounders (Theorem~\ref{thrm:logistic_regression_cca}).

When assumptions of the PB methods deviate from the true missing data mechanisms, valid type I error control is not guaranteed, and may require stronger conditions than those needed for CC approaches. \textcolor{black}{In particular, in the linear regression setting, Theorem~\ref{thrm:linear_regression_cca} shows that CC approaches can maintain valid type I error control when testing the effect of the covariate of interest in a linear regression model with a categorical confounder (Equation~(\ref{eq:linear_regression_dummy})). In contrast, the PB methods achieve valid type I error under the same setting when the imputation error is independent of the missingness (Theorem~\ref{theorem:linear_regression_pb}).} Nonetheless, both CC and PB approaches may yield biased estimates when the true parameter deviates from the null. In logistic regression with multiple confounders (Equation~(\ref{eq:logistic_regression})), CCA remains valid under MNAR whenever the covariate of interest is conditionally independent of missingness given the outcome and confounders (Theorem~\ref{thrm:logistic_regression_cca}). Whereas the PB methods are generally not guaranteed to remain valid when the outcome observation mechanism is misspecified.

% can be used when fitting linear regression with a categorical confounder (Equation~(\ref{eq:linear_regression_dummy})), type I error for testing a covariate's effect can be maintained by CC approaches when that covariate does not enter the observation model. In contrast, PB methods require an additional decomposability condition on the ML predictions, namely, that the imputation error is independent of the missingness indicator, to achieve valid type I error controls, as stated in Theorem~\ref{theorem:linear_regression_pb}. Nonetheless, both CC and PB approaches may yield biased estimates when the true parameter deviates from the null. In logistic regression with multiple confounders (Equation~(\ref{eq:logistic_regression})), CCA remains valid under MNAR whenever the covariate of interest is conditionally independent of missingness given the outcome and confounders (Theorem~\ref{thrm:logistic_regression_cca}). Whereas the PB methods are generally not guaranteed to remain valid when the outcome observation mechanism is misspecified.

\begin{figure}[H]
\centering
\includegraphics[width=0.9\linewidth]{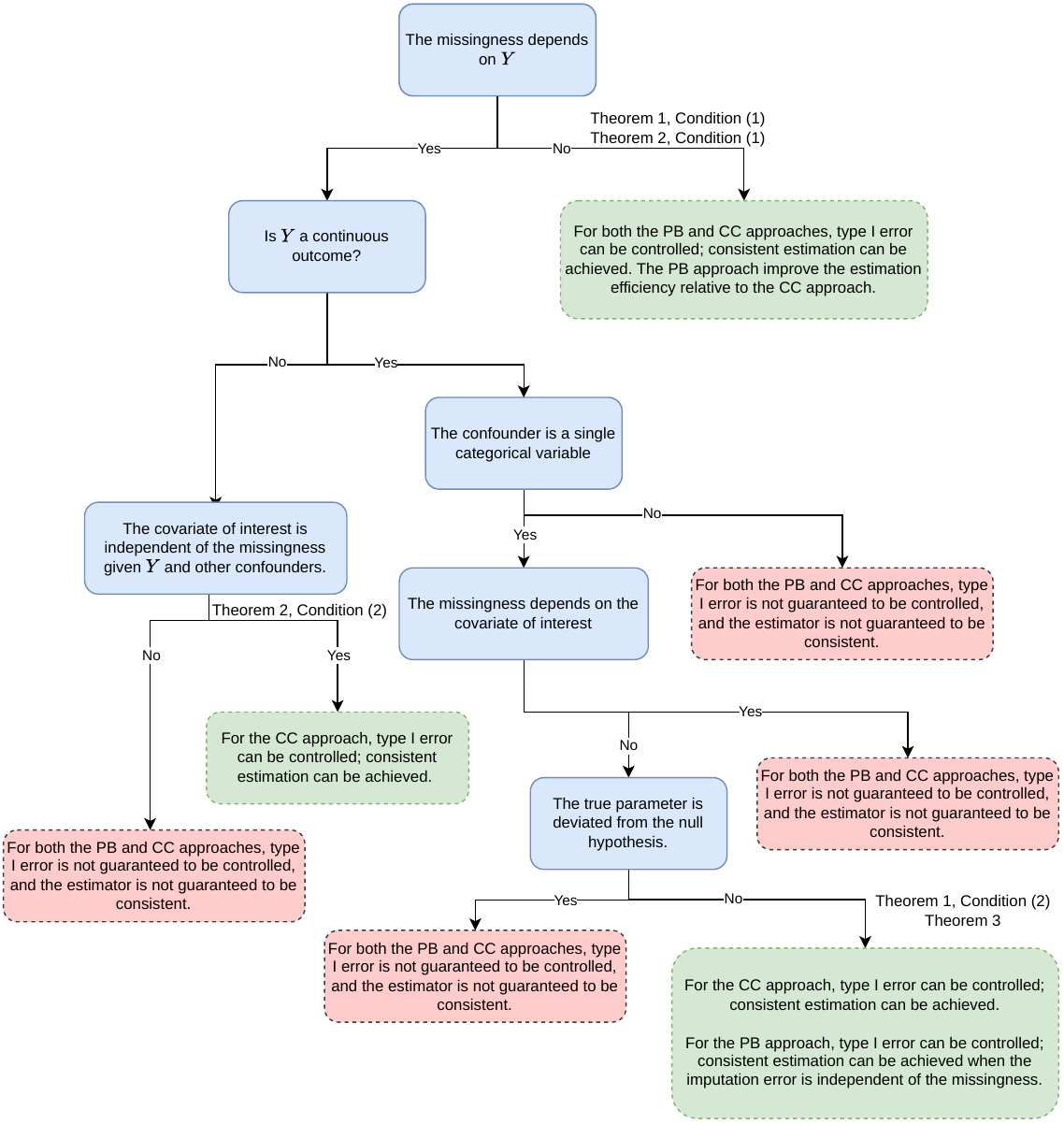}
\caption{Decision tree illustrating whether valid Type I error control and consistent estimation can be achieved under different scenarios. Abbreviations: prediction-based (PB) approaches and complete-case (CC) approaches. Red leaves indicate that neither method achieves valid Type I error control or consistent estimation, while green leaves indicate that both properties are attainable by at least one method. Relevant theorems are annotated next to the green leaves to provide formal theoretical justification.}
\label{fig:decision_tree}
\end{figure}

\textcolor{black}{To synthesize our theoretical findings,} we summarize the theoretical results in a decision tree (Figure~\ref{fig:decision_tree}) \textcolor{black}{that demonstrates} the scenarios where either CC or PB approaches can achieve valid inference. According to the tree, valid type I error control can be achieved in the following different scenarios: 1) missingness does not depend on $Y$; 2) missingness depends on $Y$, the outcome is binary, and the covariate of interest is conditionally independent of missingness given $Y$ and other confounders; 3) the missingness depends on $Y$, the outcome is continuous, the confounder is a categorical variable, and missingness does not depend on the covariate of interest. 

Finally, the AoU data analysis demonstrates that applying PB methods in EHR-linked biobank data can improve the generalizability of inferential results by extending the study population beyond the complete-case subset, which may comprise participants with poorer health status relative to the full target population. In GWAS results, PB methods enhance estimation efficiency relative to CCA and WCCA and are able to replicate the findings from previous studies. \textcolor{black}{However, the PB approaches did not consistently identify more significant SNPs compared to CCA and WCCA, likely due to the imputation quality as well as the underlying MNAR mechanisms in the EHR-linked biobank data.}

Despite the above contributions, this work has several limitations. First, although many PB methods are applicable to general Z-estimation problems, our comparative evaluation is restricted to cross-sectional regression settings. In particular, both the simulation and real-data analyses ignore the longitudinal nature of EHR-linked biobank data. This is partly due to the fact that there is no well-established PB inference theory in longitudinal data analysis or data generated from a clinically informative observation process. \textcolor{black}{Furthermore, the comparative studies in this paper focus on outcome-only missingness, whereas in practice missingness may occur across multiple variables and vary across individuals. Future research should therefore develop PB inference theory for longitudinal settings and evaluate PB methods under general missingness patterns such as~\citet{chen2025unified, xu2025blockwise, zhao2025imputation} to better understand the regimes in which the PB methods perform reliably.} Second, while this work identifies \textcolor{black}{missing mechanisms} in which consistent estimation or valid type I error control cannot be achieved, it does not propose remedies. Future research should develop methodologies that can recover consistency or validity, even if only under extra conditions. \textcolor{black}{For example,~\citet{yang2026joint} proposed a joint modeling framework that can derive valid statistical inference under the clinically informative observation process.} Finally, in the real-data application, we used a predictive biomarker for each target trait to enhance predictive performance. However, such a strategy restricts the analytic sample for each trait. Because predictive accuracy directly influences the performance of PB methods, future research should investigate strategies to enhance prediction without compromising generalizability to the target population.

%\backmatter
% \section*{Author contributions}

% \textcolor{black}{Conceptualization: Mukherjee, B.; Methodology: Chen, X.; Software: Chen, X.; Formal analysis: Chen, X.; Resources: Mukherjee, B., Wu, Z.; Writing -- original draft: Chen, X.; Writing -- review \& editing: Mukherjee, B., Wu, Z., Yang, C.; Supervision: Mukherjee, B., Wu, Z.; Funding acquisition: Mukherjee, B., Wu, Z.}

\section*{Acknowledgments}

We gratefully acknowledge \textit{All of Us} participants for their contributions, without whom this research would not have been possible. We also thank the National Institutes of Health’s \textit{All of Us} Research Program for making available the participant data examined in this study. The \textit{All of Us} data can be accessed at: \href{https://www.researchallofus.org/}{https://www.researchallofus.org/}. ZW is partly supported by an NIH grant (R33AA028315) and a Gates Foundation grant (305215). BM's effort is partially supported by UH3 CA267907 from the NIH/NCI.

\bibliographystyle{apalike}
\bibliography{sample}

\newpage

\begin{table}[H]
\caption{A summary of existing prediction-based (PB) methods. The table summarizes the accommodated missingness patterns and assumed missing mechanisms of each method.}
\label{tab:summary_literature}
\resizebox{1.0\columnwidth}{!}{
\begin{tabular}{llcccccc}
\toprule
Category & References & Method & Year & Outcome & Covariates & General & Missing Mechanism \\ 
\midrule

\multirow{6}{*}{\parbox{3cm}{\centering Augmenting\\ estimating\\ equations}}
& \citet{angelopoulos2023prediction} & PPI & 2023 & \cmark & \xmark & \xmark & MCAR \\
& \citet{angelopoulos2023ppi++} & PPI++ & 2023 & \cmark & \xmark & \xmark & MCAR \\
& \citet{gan2024prediction} & PDC & 2024 & \cmark & \xmark & \xmark & MCAR \\
& \citet{miao2025assumption} & PSPA & 2024 & \cmark & \cmark & \xmark & MCAR \\
& \citet{zhao2025imputation} & IPI & 2025 & \cmark & \cmark & \cmark & MCAR \\
& \citet{xu2025blockwise} & IBM & 2025 & \cmark & \cmark & \cmark & MCAR \\

\midrule

\multirow{3}{*}{\parbox{3cm}{\centering Augmenting\\ estimators}}
& \citet{gronsbell2024another,chen2000unified} & - & 2024 & \cmark & \xmark & \xmark & MCAR \\
& \citet{kluger2025prediction} & PTD & 2025 & \cmark & \cmark & \xmark & MAR \\
& \citet{chen2025unified} & PS-PPI & 2025 & \cmark & \cmark & \cmark & MAR \\

\midrule

\multirow{1}{*}{\parbox{3cm}{\centering Joint likelihood}}
& \citet{mccaw2023leveraging} & SynSurr & 2024 & \cmark & \xmark & \xmark & MCAR \\

\bottomrule
\end{tabular}}
\end{table}

\begin{table}[H]
\centering
\caption{Simulation results under linear regression settings with continuous confounders. \reddown; lower is better; \greenup; higher is better; $\rightarrow x$ indicates that values closer to $x$ are preferable. The best-performing results are \textbf{bolded}, and the second-best-performing results are \underline{underlined}. Abbreviations: mean squared error (MSE) and type I error (Type I). All bias and MSE values are scaled by a factor of 100 for ease of presentation.}
\label{tab:linear_regression_summary_continuous}
\resizebox{\textwidth}{!}{
\begin{tabular}{@{}%
c|l|ccc|ccc|ccc|ccc||%
c|l|ccc|ccc|ccc|ccc@{}}
\toprule
& &
\multicolumn{6}{c|}{$\beta_{1}$} &
\multicolumn{6}{c||}{$\beta_{2}$} &
& &
\multicolumn{6}{c|}{$\beta_{1}$} &
\multicolumn{6}{c}{$\beta_{2}$} \\ \midrule
\multirow{3}{*}{\shortstack{\textbf{Missing}\\\textbf{Mechanism}}} & \multirow{3}{*}{\textbf{Method}} &
\multicolumn{3}{c}{$\beta_{1}=0,\,N=10{,}000$} &
\multicolumn{3}{c|}{$\beta_{1}=0.1,\,N=10{,}000$} &
\multicolumn{3}{c}{$\beta_{2}=0,\,N=10{,}000$} &
\multicolumn{3}{c||}{$\beta_{2}=0.1,\,N=10{,}000$} &
\multirow{3}{*}{\shortstack{\textbf{Missing}\\\textbf{Mechanism}}} & \multirow{3}{*}{\textbf{Method}} &
\multicolumn{3}{c}{$\beta_{1}=0,\,N=10{,}000$} &
\multicolumn{3}{c|}{$\beta_{1}=0.1,\,N=1{,}000$} &
\multicolumn{3}{c}{$\beta_{2}=0,\,N=10{,}000$} &
\multicolumn{3}{c}{$\beta_{2}=0.1,\,N=1{,}000$} \\
& &
\textbf{$\lvert$Bias$_0\rvert$} & \textbf{MSE$_0$} & \textbf{Type I} &
\textbf{$\lvert$Bias$\rvert$} & \textbf{MSE} & \textbf{Power} &
\textbf{$\lvert$Bias$_0\rvert$} & \textbf{MSE$_0$} & \textbf{Type I} &
\textbf{$\lvert$Bias$\rvert$} & \textbf{MSE} & \textbf{Power} &
& &
\textbf{$\lvert$Bias$_0\rvert$} & \textbf{MSE$_0$} & \textbf{Type I} &
\textbf{$\lvert$Bias$\rvert$} & \textbf{MSE} & \textbf{Power} &
\textbf{$\lvert$Bias$_0\rvert$} & \textbf{MSE$_0$} & \textbf{Type I} &
\textbf{$\lvert$Bias$\rvert$} & \textbf{MSE} & \textbf{Power} \\
& &
$\rightarrow 0$ & \reddown & $\rightarrow 0.05$ &
$\rightarrow 0$ & \reddown & \greenup &
$\rightarrow 0$ & \reddown & $\rightarrow 0.05$ &
$\rightarrow 0$ & \reddown & \greenup &
& &
$\rightarrow 0$ & \reddown & $\rightarrow 0.05$ &
$\rightarrow 0$ & \reddown & \greenup &
$\rightarrow 0$ & \reddown & $\rightarrow 0.05$ &
$\rightarrow 0$ & \reddown & \greenup \\
\midrule
% ===================== Pair: MCAR (left)  ||  MNAR3 (right) =====================
\multirow{10}{*}{\shortstack{MCAR \\ $(\varnothing)$}} & Full & 0.11 & 0.03 & 0.052 & 0.11 & 0.03 & 1.000 & 0.04 & 0.05 & 0.052 & 0.04 & 0.05 & 0.992 & \multirow{10}{*}{\shortstack{MNAR3 \\ $(X_2, Z_2, Y)$}} & Full & 0.11 & 0.03 & 0.052 & 0.11 & 0.03 & 1.000 & 0.04 & 0.05 & 0.052 & 0.04 & 0.05 & 0.992 \\
& WCCA (est) & 0.37 & 0.14 & 0.048 & 0.37 & 0.14 & 0.732 & \underline{0.01} & 0.22 & 0.032 & \underline{0.01} & 0.22 & 0.592 & & WCCA (est) & \underline{3.44} & 0.35 & \textbf{0.090} & 4.85 & 0.46 & 0.142 & 8.10 & 1.05 & 0.242 & 6.99 & 0.88 & 0.406 \\
& CCA & 0.36 & 0.14 & 0.042 & 0.36 & 0.14 & 0.742 & \textbf{0.00} & 0.22 & 0.032 & \textbf{0.00} & 0.22 & 0.592 & & CCA & 6.11 & 0.53 & 0.368 & 6.83 & 0.62 & 0.010 & \underline{4.62} & \underline{0.50} & \underline{0.140} & \underline{4.00} & \underline{0.45} & 0.558 \\
& Naive & 15.23 & 2.36 & 1.000 & 15.23 & 2.36 & \textbf{0.998} & 13.19 & 1.80 & 1.000 & 13.19 & 1.80 & 0.000 & & Naive & 17.00 & 2.93 & 1.000 & 17.02 & 2.94 & \textbf{1.000} & 13.81 & 1.97 & 1.000 & 13.88 & 1.99 & 0.000 \\
& PPI & 0.17 & 0.13 & 0.044 & 0.17 & 0.13 & 0.802 & 0.08 & 0.21 & \underline{0.048} & 0.08 & 0.21 & 0.578 & & PPI & 11.99 & 1.58 & 0.910 & 11.16 & 1.38 & 0.802 & 32.15 & 10.56 & 1.000 & 30.41 & 9.47 & 0.446 \\
& PPI++ & \textbf{0.07} & \underline{0.07} & \textbf{0.052} & \textbf{0.07} & \underline{0.07} & \underline{0.972} & 0.05 & 0.10 & 0.046 & 0.05 & 0.10 & \underline{0.864} & & PPI++ & 5.41 & 0.37 & 0.514 & 4.74 & 0.30 & 0.960 & 22.12 & 5.02 & 1.000 & 20.97 & 4.53 & \underline{0.784} \\
& SynSurr & \underline{0.08} & 0.07 & 0.054 & \underline{0.08} & 0.07 & 0.968 & 0.05 & \underline{0.10} & \underline{0.048} & 0.05 & \underline{0.10} & \textbf{0.866} & & SynSurr & 3.58 & \textbf{0.20} & 0.280 & \underline{3.19} & \textbf{0.17} & \underline{0.974} & 10.88 & 1.31 & 0.852 & 10.11 & 1.15 & \textbf{0.838} \\
& PS-PPI & 0.10 & \textbf{0.07} & \textbf{0.052} & 0.10 & \textbf{0.07} & 0.970 & 0.04 & \textbf{0.10} & 0.044 & 0.04 & \textbf{0.10} & \underline{0.864} & & PS-PPI & \textbf{2.75} & \underline{0.21} & \underline{0.150} & \textbf{3.02} & \underline{0.22} & 0.260 & 6.43 & 0.63 & 0.344 & 6.25 & 0.61 & 0.708 \\
& MI & 0.13 & 0.18 & 0.064 & 0.20 & 0.18 & 0.526 & 0.10 & 0.27 & \textbf{0.050} & 0.06 & 0.26 & 0.394 & & MI & 9.83 & 1.13 & 0.518 & 10.89 & 1.36 & 0.002 & \textbf{3.59} & \textbf{0.45} & \textbf{0.116} & \textbf{2.87} & \textbf{0.39} & 0.310 \\
& MI-RF & 0.32 & 0.14 & 0.200 & 0.94 & 0.15 & 0.632 & 1.89 & 0.27 & 0.194 & 1.73 & 0.26 & 0.550 & & MI-RF & 4.07 & 0.37 & 0.418 & 6.26 & 0.59 & 0.046 & 12.26 & 1.88 & 0.766 & 10.67 & 1.53 & 0.226 \\
\midrule
% ===================== Pair: MAR1 (left)  ||  MNAR4 (right) =====================
\multirow{10}{*}{\shortstack{MAR1 \\ $(Z_2)$}} & Full & 0.11 & 0.03 & 0.052 & 0.11 & 0.03 & 1.000 & 0.04 & 0.05 & 0.052 & 0.04 & 0.05 & 0.992 & \multirow{10}{*}{\shortstack{MNAR4 \\ $(X_1, X_2, Z_2, Y)$}} & Full & 0.11 & 0.03 & 0.052 & 0.11 & 0.03 & 1.000 & 0.04 & 0.05 & 0.052 & 0.04 & 0.05 & 0.992 \\
& WCCA (est) & 0.53 & 0.29 & \underline{0.046} & 0.53 & 0.29 & 0.430 & 0.44 & 0.66 & \underline{0.052} & 0.44 & 0.66 & 0.256 & & WCCA (est) & 10.70 & 1.61 & \underline{0.390} & 9.50 & 1.34 & 0.358 & 8.60 & 1.27 & \textbf{0.224} & 7.62 & 1.10 & 0.300 \\
& CCA & \textbf{0.17} & \underline{0.14} & \textbf{0.048} & \textbf{0.17} & \underline{0.14} & 0.746 & \textbf{0.30} & 0.22 & 0.036 & \textbf{0.30} & 0.22 & 0.638 & & CCA & 17.42 & 3.22 & 0.984 & 16.93 & 3.06 & 0.652 & 8.34 & 0.97 & 0.350 & 8.11 & 0.94 & 0.558 \\
& Naive & 16.43 & 2.74 & 1.000 & 16.43 & 2.74 & \textbf{1.000} & 14.00 & 2.02 & 1.000 & 14.00 & 2.02 & 0.000 & & Naive & 15.78 & 2.53 & 1.000 & 15.84 & 2.55 & \textbf{1.000} & 12.87 & 1.72 & 1.000 & 12.93 & 1.73 & 0.000 \\
& PPI & 7.03 & 0.63 & 0.484 & 7.03 & 0.63 & 0.832 & 0.82 & 0.21 & 0.042 & 0.82 & 0.21 & 0.544 & & PPI & 31.67 & 10.19 & 1.000 & 30.66 & 9.56 & 0.000 & 40.67 & 16.77 & 1.000 & 39.17 & 15.57 & 0.000 \\
& PPI++ & 4.12 & 0.24 & 0.338 & 4.12 & 0.24 & \underline{0.988} & 0.38 & \textbf{0.11} & 0.034 & 0.38 & \textbf{0.11} & \underline{0.840} & & PPI++ & 25.88 & 6.78 & 1.000 & 25.17 & 6.43 & 0.000 & 27.50 & 7.69 & 1.000 & 26.73 & 7.27 & 0.000 \\
& SynSurr & 3.71 & 0.21 & 0.276 & 3.71 & 0.21 & 0.986 & 0.77 & \underline{0.11} & 0.046 & 0.77 & \underline{0.11} & \textbf{0.922} & & SynSurr & 16.56 & 2.83 & 1.000 & 16.18 & 2.70 & \underline{0.942} & 17.35 & 3.13 & 0.998 & 16.88 & 2.97 & \textbf{0.862} \\
& PS-PPI & \underline{0.34} & \textbf{0.12} & 0.066 & \underline{0.34} & \textbf{0.12} & 0.758 & \underline{0.33} & 0.26 & \textbf{0.050} & \underline{0.33} & 0.26 & 0.540 & & PS-PPI & \textbf{3.29} & \textbf{0.41} & \textbf{0.198} & \textbf{2.95} & \textbf{0.36} & 0.566 & \textbf{7.91} & \textbf{0.95} & 0.408 & \underline{7.57} & \underline{0.89} & \underline{0.568} \\
& MI & 0.41 & 0.18 & 0.084 & 0.45 & 0.19 & 0.312 & 0.59 & 0.30 & 0.070 & 0.42 & 0.28 & 0.422 & & MI & 11.90 & 1.64 & 0.630 & 10.57 & 1.33 & 0.128 & \underline{8.04} & \underline{0.97} & \underline{0.292} & \textbf{7.53} & \textbf{0.88} & 0.238 \\
& MI-RF & 1.10 & 0.25 & 0.252 & 1.69 & 0.26 & 0.236 & 8.04 & 1.16 & 0.492 & 7.44 & 1.06 & 0.224 & & MI-RF & \underline{8.35} & \underline{0.99} & 0.626 & \underline{6.02} & \underline{0.64} & 0.152 & 12.28 & 1.94 & 0.746 & 10.87 & 1.60 & 0.192 \\
\midrule
% ===================== Pair: MAR2 (left)  ||  MNAR5 (right) =====================
\multirow{10}{*}{\shortstack{MAR2 \\ $(X_1, Z_2)$}} & Full & 0.11 & 0.03 & 0.052 & 0.11 & 0.03 & 1.000 & 0.04 & 0.05 & 0.052 & 0.04 & 0.05 & 0.992 & \multirow{10}{*}{\shortstack{MNAR5 \\ $(X_1, X_2, Z_2, Y, X_1Y)$}} & Full & 0.11 & 0.03 & 0.052 & 0.11 & 0.03 & 1.000 & 0.04 & 0.05 & 0.052 & 0.04 & 0.05 & 0.992 \\
& WCCA (est) & 0.10 & 0.59 & \textbf{0.052} & 0.10 & 0.59 & 0.250 & 0.66 & 0.97 & \textbf{0.054} & 0.66 & 0.97 & 0.230 & & WCCA (est) & 17.44 & 3.93 & 0.512 & 18.26 & 4.18 & 0.006 & \textbf{6.72} & \underline{1.11} & \textbf{0.150} & \textbf{6.04} & \underline{0.99} & 0.324 \\
& CCA & \underline{0.09} & \textbf{0.19} & 0.036 & \underline{0.09} & \textbf{0.19} & 0.656 & \textbf{0.36} & \textbf{0.24} & 0.060 & \textbf{0.36} & \textbf{0.24} & 0.556 & & CCA & \textbf{0.77} & \textbf{0.24} & \textbf{0.048} & \textbf{0.14} & \textbf{0.24} & 0.514 & 9.37 & 1.14 & 0.454 & 8.94 & 1.06 & \underline{0.568} \\
& Naive & 15.26 & 2.37 & 1.000 & 15.26 & 2.37 & \textbf{1.000} & 12.91 & 1.73 & 1.000 & 12.91 & 1.73 & 0.000 & & Naive & 14.72 & 2.21 & 1.000 & 14.80 & 2.23 & \textbf{0.998} & 12.66 & 1.67 & 1.000 & 12.73 & 1.68 & 0.000 \\
& PPI & 18.09 & 3.43 & 0.990 & 18.09 & 3.43 & 0.818 & 7.59 & 0.77 & 0.398 & 7.59 & 0.77 & 0.754 & & PPI & 51.29 & 26.50 & 1.000 & 49.80 & 25.00 & 0.000 & 39.80 & 16.08 & 1.000 & 38.13 & 14.77 & 0.000 \\
& PPI++ & 10.24 & 1.13 & 0.952 & 10.24 & 1.13 & 0.950 & 4.45 & \underline{0.30} & 0.224 & 4.45 & \underline{0.30} & \textbf{0.928} & & PPI++ & 31.07 & 9.77 & 1.000 & 30.24 & 9.26 & 0.000 & 27.60 & 7.75 & 1.000 & 26.67 & 7.24 & 0.000 \\
& SynSurr & 9.54 & 0.99 & 0.902 & 9.54 & 0.99 & \underline{0.956} & 5.18 & 0.37 & 0.368 & 5.18 & 0.37 & \underline{0.922} & & SynSurr & 15.56 & 2.52 & 1.000 & 15.04 & 2.37 & \underline{0.894} & 17.80 & 3.29 & 1.000 & 17.28 & 3.11 & \textbf{0.826} \\
& PS-PPI & \textbf{0.02} & \underline{0.20} & \underline{0.042} & \textbf{0.02} & \underline{0.20} & 0.630 & \underline{0.49} & 0.36 & \underline{0.058} & \underline{0.49} & 0.36 & 0.448 & & PS-PPI & \underline{3.70} & 1.46 & 0.370 & \underline{4.21} & 1.34 & 0.042 & 11.19 & 1.93 & 0.548 & 10.81 & 1.78 & 0.326 \\
& MI & 0.37 & 0.25 & 0.068 & 0.43 & 0.27 & 0.344 & 0.62 & 0.33 & 0.078 & 0.52 & 0.30 & 0.334 & & MI & 4.59 & \underline{0.48} & \underline{0.130} & 5.53 & \underline{0.59} & 0.060 & \underline{8.08} & \textbf{0.97} & \underline{0.288} & \underline{7.86} & \textbf{0.91} & 0.278 \\
& MI-RF & 2.53 & 0.42 & 0.308 & 3.30 & 0.46 & 0.134 & 7.78 & 1.13 & 0.440 & 7.14 & 1.07 & 0.270 & & MI-RF & 15.20 & 2.87 & 0.776 & 17.54 & 3.63 & 0.024 & 12.55 & 2.13 & 0.688 & 11.06 & 1.77 & 0.128 \\
\midrule
% ===================== Pair: MNAR1 (left)  ||  MNAR6 (right) =====================
\multirow{10}{*}{\shortstack{MNAR1 \\ $(Z_2, Y)$}} & Full & 0.11 & 0.03 & 0.052 & 0.11 & 0.03 & 1.000 & 0.04 & 0.05 & 0.052 & 0.04 & 0.05 & 0.992 & \multirow{10}{*}{\shortstack{MNAR6 \\ $(X_1, X_2, Z_2, Y, X_2Y)$}} & Full & 0.11 & 0.03 & 0.052 & 0.11 & 0.03 & 1.000 & 0.04 & 0.05 & 0.052 & 0.04 & 0.05 & 0.992 \\
& WCCA (est) & 3.87 & 0.33 & 0.166 & 5.40 & 0.47 & 0.058 & 6.84 & 0.78 & 0.238 & 8.19 & 0.99 & 0.004 & & WCCA (est) & 12.61 & 2.21 & \underline{0.448} & 11.85 & 1.98 & 0.304 & 16.62 & 3.55 & 0.492 & 17.23 & 3.72 & 0.000 \\
& CCA & 5.37 & 0.44 & 0.302 & 6.50 & 0.57 & 0.028 & 17.23 & 3.19 & 0.964 & 18.24 & 3.55 & 0.000 & & CCA & 21.88 & 4.96 & 1.000 & 21.43 & 4.77 & \underline{0.632} & 8.50 & 1.05 & 0.314 & 8.88 & 1.12 & 0.000 \\
& Naive & 15.78 & 2.53 & 1.000 & 15.82 & 2.54 & \textbf{1.000} & 13.57 & 1.91 & 1.000 & 13.63 & 1.92 & 0.000 & & Naive & 16.07 & 2.62 & 1.000 & 16.13 & 2.64 & \textbf{1.000} & 13.43 & 1.87 & 1.000 & 13.50 & 1.88 & 0.000 \\
& PPI & 4.45 & 0.34 & 0.226 & 3.67 & 0.28 & 0.814 & 3.15 & 0.29 & \underline{0.092} & \underline{1.85} & \underline{0.22} & \underline{0.620} & & PPI & 31.96 & 10.36 & 1.000 & 30.82 & 9.65 & 0.000 & 65.31 & 42.93 & 1.000 & 63.18 & 40.20 & 0.000 \\
& PPI++ & \underline{1.00} & \underline{0.08} & \underline{0.066} & \textbf{0.16} & \textbf{0.08} & \underline{0.942} & 4.01 & \underline{0.27} & 0.214 & 5.09 & 0.37 & 0.080 & & PPI++ & 27.91 & 7.87 & 1.000 & 27.11 & 7.43 & 0.000 & 35.61 & 12.85 & 1.000 & 34.70 & 12.21 & 0.000 \\
& SynSurr & \textbf{0.29} & \textbf{0.07} & \textbf{0.064} & \underline{0.90} & \underline{0.08} & 0.920 & 6.43 & 0.52 & 0.500 & 7.21 & 0.62 & 0.000 & & SynSurr & 18.67 & 3.57 & 1.000 & 18.33 & 3.45 & 0.000 & 18.19 & 3.47 & 0.996 & 17.62 & 3.26 & \textbf{0.754} \\
& PS-PPI & 3.45 & 0.22 & 0.236 & 3.62 & 0.23 & 0.254 & \underline{0.89} & \textbf{0.19} & \textbf{0.086} & \textbf{0.70} & \textbf{0.19} & \textbf{0.666} & & PS-PPI & \textbf{5.77} & \textbf{0.98} & \textbf{0.344} & \textbf{5.49} & \textbf{0.86} & 0.384 & \textbf{1.30} & \textbf{0.70} & \textbf{0.184} & \textbf{1.79} & \textbf{0.64} & \underline{0.138} \\
& MI & 8.64 & 0.91 & 0.488 & 10.36 & 1.25 & 0.008 & 16.02 & 2.83 & 0.788 & 17.70 & 3.39 & 0.006 & & MI & 16.04 & 2.79 & 0.872 & 15.12 & 2.49 & 0.122 & \underline{7.37} & \underline{0.96} & \underline{0.248} & \underline{8.12} & \underline{1.03} & 0.004 \\
& MI-RF & 4.30 & 0.34 & 0.446 & 6.53 & 0.59 & 0.046 & \textbf{0.70} & 0.29 & 0.218 & 2.68 & 0.35 & 0.212 & & MI-RF & \underline{10.05} & \underline{1.41} & 0.680 & \underline{8.03} & \underline{1.05} & 0.218 & 13.24 & 2.36 & 0.726 & 14.69 & 2.77 & 0.008 \\
\midrule
% ===================== Pair: MNAR2 (left)  ||  MNAR7 (right) =====================
\multirow{10}{*}{\shortstack{MNAR2 \\ $(X_1, Z_2, Y)$}} & Full & 0.11 & 0.03 & 0.052 & 0.11 & 0.03 & 1.000 & 0.04 & 0.05 & 0.052 & 0.04 & 0.05 & 0.992 & \multirow{10}{*}{\shortstack{MNAR7 \\ $(X_1, X_2, Z_2, Y, Z_2Y)$}} & Full & 0.11 & 0.03 & 0.052 & 0.11 & 0.03 & 1.000 & 0.04 & 0.05 & 0.052 & 0.04 & 0.05 & 0.992 \\
& WCCA (est) & 10.06 & 1.37 & 0.442 & 8.81 & 1.12 & 0.390 & 5.73 & 0.78 & 0.168 & 6.88 & 0.91 & 0.010 & & WCCA (est) & 14.63 & 2.85 & \underline{0.502} & 13.93 & 2.60 & 0.296 & \textbf{4.13} & 1.08 & \textbf{0.096} & \textbf{3.86} & 1.03 & 0.242 \\
& CCA & 14.25 & 2.21 & 0.924 & 13.37 & 1.97 & 0.650 & 13.68 & 2.10 & 0.826 & 14.28 & 2.27 & 0.000 & & CCA & 25.15 & 6.50 & 1.000 & 24.81 & 6.33 & \underline{0.684} & 7.98 & \textbf{0.94} & 0.322 & 7.84 & 0.92 & \underline{0.500} \\
& Naive & 14.26 & 2.08 & 1.000 & 14.31 & 2.09 & \textbf{1.000} & 12.29 & 1.57 & 1.000 & 12.36 & 1.59 & 0.000 & & Naive & 15.97 & 2.59 & 1.000 & 16.02 & 2.61 & \textbf{1.000} & 13.13 & 1.79 & 1.000 & 13.21 & 1.81 & 0.000 \\
& PPI & 28.88 & 8.51 & 1.000 & 27.75 & 7.86 & 0.000 & 8.70 & 0.95 & 0.482 & 7.37 & 0.73 & 0.622 & & PPI & 30.22 & 9.28 & 1.000 & 28.92 & 8.52 & 0.000 & 51.35 & 26.59 & 1.000 & 49.18 & 24.41 & 0.000 \\
& PPI++ & 23.15 & 5.45 & 1.000 & 22.21 & 5.03 & 0.000 & \textbf{0.09} & \textbf{0.11} & \textbf{0.032} & \textbf{0.98} & \textbf{0.12} & \textbf{0.778} & & PPI++ & 28.48 & 8.20 & 1.000 & 27.55 & 7.68 & 0.000 & 36.48 & 13.45 & 1.000 & 35.41 & 12.68 & 0.000 \\
& SynSurr & 12.92 & 1.76 & 0.996 & 12.31 & 1.60 & \underline{0.942} & 1.82 & \underline{0.14} & \underline{0.108} & 2.38 & \underline{0.16} & 0.558 & & SynSurr & 19.42 & 3.86 & 1.000 & 19.06 & 3.72 & 0.000 & 21.15 & 4.61 & 1.000 & 20.60 & 4.38 & \textbf{0.788} \\
& PS-PPI & \textbf{3.07} & \textbf{0.29} & \textbf{0.188} & \textbf{2.84} & \textbf{0.27} & 0.744 & 2.44 & 0.30 & 0.134 & 2.13 & 0.28 & \underline{0.624} & & PS-PPI & \textbf{4.95} & \textbf{0.98} & \textbf{0.340} & \textbf{4.81} & \textbf{0.88} & 0.324 & \underline{6.09} & \underline{0.96} & \underline{0.238} & \underline{5.87} & \underline{0.89} & 0.322 \\
& MI & 8.97 & 1.01 & \underline{0.436} & 7.76 & 0.81 & 0.268 & 11.43 & 1.56 & 0.602 & 12.45 & 1.84 & 0.012 & & MI & 18.82 & 3.80 & 0.954 & 17.42 & 3.32 & 0.084 & 8.25 & 1.00 & 0.318 & 6.95 & \textbf{0.80} & 0.208 \\
& MI-RF & \underline{7.53} & \underline{0.80} & 0.648 & \underline{5.34} & \underline{0.52} & 0.224 & \underline{0.22} & 0.35 & 0.224 & \underline{1.73} & 0.38 & 0.160 & & MI-RF & \underline{10.25} & \underline{1.51} & 0.678 & \underline{8.63} & \underline{1.21} & 0.156 & 12.84 & 2.30 & 0.656 & 11.47 & 1.94 & 0.116 \\
\bottomrule
\end{tabular}}
\end{table}

\begin{table}[H]
\centering
\caption{Simulation results under the linear regression settings with a categorical confounder. \reddown; lower is better; \greenup; higher is better; $\rightarrow x$ indicates that values closer to $x$ are preferable. The best-performing results are \textbf{bolded}, and the second-best-performing results are \underline{underlined}. Abbreviations: mean squared error (MSE) and type I error (Type I). All bias and MSE values are scaled by a factor of 100 for ease of presentation.}
\label{tab:linear_regression_summary_dummy}
\resizebox{\textwidth}{!}{
\begin{tabular}{@{}%
c|l|ccc|ccc|ccc|ccc||%
c|l|ccc|ccc|ccc|ccc@{}}
\toprule
& &
\multicolumn{6}{c|}{$\beta_{1}$} &
\multicolumn{6}{c||}{$\beta_{2}$} &
& &
\multicolumn{6}{c|}{$\beta_{1}$} &
\multicolumn{6}{c}{$\beta_{2}$} \\ \midrule
\multirow{3}{*}{\shortstack{\textbf{Missing}\\\textbf{Mechanism}}} & \multirow{3}{*}{\textbf{Method}} &
\multicolumn{3}{c}{$\beta_{1}=0,\,N=10{,}000$} &
\multicolumn{3}{c|}{$\beta_{1}=0.1,\,N=10{,}000$} &
\multicolumn{3}{c}{$\beta_{2}=0,\,N=10{,}000$} &
\multicolumn{3}{c||}{$\beta_{2}=0.1,\,N=10{,}000$} &
\multirow{3}{*}{\shortstack{\textbf{Missing}\\\textbf{Mechanism}}} & \multirow{3}{*}{\textbf{Method}} &
\multicolumn{3}{c}{$\beta_{1}=0,\,N=10{,}000$} &
\multicolumn{3}{c|}{$\beta_{1}=0.1,\,N=1{,}000$} &
\multicolumn{3}{c}{$\beta_{2}=0,\,N=10{,}000$} &
\multicolumn{3}{c}{$\beta_{2}=0.1,\,N=1{,}000$} \\
& &
\textbf{$\lvert$Bias$_0\rvert$} & \textbf{MSE$_0$} & \textbf{Type I} &
\textbf{$\lvert$Bias$\rvert$} & \textbf{MSE} & \textbf{Power} &
\textbf{$\lvert$Bias$_0\rvert$} & \textbf{MSE$_0$} & \textbf{Type I} &
\textbf{$\lvert$Bias$\rvert$} & \textbf{MSE} & \textbf{Power} &
& &
\textbf{$\lvert$Bias$_0\rvert$} & \textbf{MSE$_0$} & \textbf{Type I} &
\textbf{$\lvert$Bias$\rvert$} & \textbf{MSE} & \textbf{Power} &
\textbf{$\lvert$Bias$_0\rvert$} & \textbf{MSE$_0$} & \textbf{Type I} &
\textbf{$\lvert$Bias$\rvert$} & \textbf{MSE} & \textbf{Power} \\
& &
$\rightarrow 0$ & \reddown & $\rightarrow 0.05$ &
$\rightarrow 0$ & \reddown & \greenup &
$\rightarrow 0$ & \reddown & $\rightarrow 0.05$ &
$\rightarrow 0$ & \reddown & \greenup &
& &
$\rightarrow 0$ & \reddown & $\rightarrow 0.05$ &
$\rightarrow 0$ & \reddown & \greenup &
$\rightarrow 0$ & \reddown & $\rightarrow 0.05$ &
$\rightarrow 0$ & \reddown & \greenup \\
\midrule
% ===================== Pair: MCAR (left)  ||  MNAR3 (right) =====================
\multirow{10}{*}{\shortstack{MCAR \\ $(\varnothing)$}} & Full & 0.11 & 0.03 & 0.052 & 0.11 & 0.03 & 1.000 & 0.04 & 0.05 & 0.050 & 0.04 & 0.05 & 0.988 & \multirow{10}{*}{\shortstack{MNAR3 \\ $(X_2, Z_2, Y)$}} & Full & 0.11 & 0.03 & 0.052 & 0.11 & 0.03 & 1.000 & 0.04 & 0.05 & 0.050 & 0.04 & 0.05 & 0.988 \\
& WCCA (est) & 0.37 & 0.14 & 0.044 & 0.37 & 0.14 & 0.740 & \underline{0.04} & 0.26 & \textbf{0.050} & \underline{0.04} & 0.26 & 0.472 & & WCCA (est) & 0.28 & 0.26 & \underline{0.046} & 0.59 & 0.26 & 0.552 & 11.99 & 1.88 & 0.390 & 12.74 & 2.06 & 0.502 \\
& CCA & 0.36 & 0.14 & 0.036 & 0.36 & 0.14 & 0.740 & \textbf{0.04} & 0.26 & 0.052 & \textbf{0.04} & 0.26 & 0.466 & & CCA & 0.22 & 0.15 & \textbf{0.052} & 0.69 & \underline{0.16} & 0.774 & 13.56 & 2.22 & 0.590 & 14.33 & 2.43 & 0.570 \\
& Naive & 15.23 & 2.36 & 1.000 & 15.23 & 2.36 & \textbf{0.998} & 12.54 & 1.65 & 1.000 & 12.54 & 1.65 & 0.000 & & Naive & 16.82 & 2.87 & 1.000 & 16.83 & 2.87 & \textbf{1.000} & 12.01 & 1.51 & 0.994 & 12.06 & 1.53 & 0.000 \\
& PPI & 0.18 & 0.13 & 0.040 & 0.18 & 0.13 & 0.804 & 0.23 & 0.25 & \textbf{0.050} & 0.23 & 0.25 & 0.550 & & PPI & 10.27 & 1.19 & 0.810 & 9.62 & 1.06 & 0.836 & 34.69 & 12.28 & 1.000 & 33.43 & 11.42 & 0.560 \\
& PPI++ & \underline{0.08} & 0.07 & \textbf{0.050} & \underline{0.08} & 0.07 & \underline{0.972} & 0.14 & 0.13 & \textbf{0.050} & 0.14 & \underline{0.13} & \underline{0.808} & & PPI++ & 6.28 & 0.47 & 0.648 & 6.23 & 0.47 & 0.968 & 26.65 & 7.25 & 1.000 & 26.19 & 7.00 & 0.000 \\
& SynSurr & \textbf{0.08} & 0.07 & \underline{0.054} & \textbf{0.08} & \underline{0.07} & 0.968 & 0.14 & 0.13 & \textbf{0.050} & 0.14 & 0.13 & \underline{0.808} & & SynSurr & 5.52 & 0.38 & 0.556 & 5.71 & 0.40 & \underline{0.978} & 19.22 & 3.84 & 0.996 & 19.03 & 3.76 & \textbf{0.854} \\
& PS-PPI & 0.09 & \underline{0.07} & \underline{0.054} & 0.09 & \textbf{0.07} & 0.970 & 0.14 & \underline{0.12} & 0.044 & 0.14 & \textbf{0.12} & \textbf{0.818} & & PS-PPI & \textbf{0.06} & \underline{0.12} & 0.064 & \underline{0.46} & \textbf{0.12} & 0.874 & \textbf{4.43} & \textbf{0.39} & \underline{0.146} & \underline{4.82} & \underline{0.42} & \underline{0.772} \\
& MI & 0.16 & 0.55 & 0.016 & 0.21 & 0.23 & 0.668 & 0.99 & 1.23 & 0.014 & 0.49 & 0.40 & 0.446 & & MI & 0.27 & 0.26 & 0.038 & \textbf{0.18} & 0.23 & 0.534 & 12.20 & 3.40 & \textbf{0.130} & 12.76 & 2.28 & 0.358 \\
& MI-RF & 0.16 & \textbf{0.03} & 0.016 & 5.57 & 0.36 & 0.578 & 0.05 & \textbf{0.06} & 0.008 & 5.45 & 0.38 & 0.382 & & MI-RF & \underline{0.11} & \textbf{0.05} & 0.022 & 4.84 & 0.31 & 0.484 & \underline{5.54} & \underline{0.42} & 0.160 & \textbf{1.41} & \textbf{0.22} & 0.284 \\
\midrule
% ===================== Pair: MAR1 (left)  ||  MNAR4 (right) =====================
\multirow{10}{*}{\shortstack{MAR1 \\ $(Z_2)$}} & Full & 0.11 & 0.03 & 0.052 & 0.11 & 0.03 & 1.000 & 0.04 & 0.05 & 0.050 & 0.04 & 0.05 & 0.988 & \multirow{10}{*}{\shortstack{MNAR4 \\ $(X_1, X_2, Z_2, Y)$}} & Full & 0.11 & 0.03 & 0.052 & 0.11 & 0.03 & 1.000 & 0.04 & 0.05 & 0.050 & 0.04 & 0.05 & 0.988 \\
& WCCA (est) & 0.55 & 0.15 & 0.030 & 0.55 & 0.15 & 0.718 & \textbf{0.03} & 0.31 & 0.046 & \textbf{0.03} & 0.31 & 0.444 & & WCCA (est) & 10.87 & 1.73 & 0.312 & 11.61 & 1.86 & 0.404 & 11.47 & 2.07 & 0.282 & 12.16 & 2.20 & 0.366 \\
& CCA & 0.48 & 0.13 & 0.032 & \underline{0.48} & 0.13 & 0.768 & 0.11 & 0.30 & \underline{0.048} & \underline{0.11} & 0.30 & 0.476 & & CCA & 14.48 & 2.30 & 0.890 & 15.12 & 2.49 & 0.712 & 14.69 & 2.54 & 0.682 & 15.36 & 2.74 & 0.574 \\
& Naive & 15.58 & 2.47 & 1.000 & 15.58 & 2.47 & \textbf{1.000} & 12.32 & 1.59 & 0.994 & 12.32 & 1.59 & 0.000 & & Naive & 15.74 & 2.52 & 1.000 & 15.77 & 2.53 & \textbf{1.000} & 10.99 & 1.28 & 0.982 & 11.03 & 1.29 & 0.000 \\
& PPI & 2.26 & 0.18 & 0.080 & 2.26 & 0.18 & 0.862 & 3.46 & 0.36 & 0.104 & 3.46 & 0.36 & 0.618 & & PPI & 26.44 & 7.14 & 1.000 & 25.68 & 6.74 & 0.000 & 45.71 & 21.17 & 1.000 & 44.66 & 20.22 & 0.000 \\
& PPI++ & 0.97 & 0.07 & 0.066 & 0.97 & 0.07 & \underline{0.984} & 1.89 & 0.18 & 0.102 & 1.89 & \underline{0.18} & \textbf{0.850} & & PPI++ & 21.83 & 4.85 & 1.000 & 21.62 & 4.76 & 0.000 & 33.77 & 11.56 & 1.000 & 33.43 & 11.33 & 0.000 \\
& SynSurr & 0.98 & 0.07 & \underline{0.064} & 0.98 & \underline{0.07} & \underline{0.984} & 1.94 & 0.18 & 0.090 & 1.94 & 0.18 & \textbf{0.850} & & SynSurr & 15.19 & 2.39 & 1.000 & 15.25 & 2.41 & \underline{0.974} & 26.25 & 7.05 & 1.000 & 26.19 & 7.02 & \textbf{0.796} \\
& PS-PPI & \textbf{0.22} & \underline{0.07} & \textbf{0.046} & \textbf{0.22} & \textbf{0.07} & 0.950 & 0.16 & \underline{0.14} & \textbf{0.050} & 0.16 & \textbf{0.14} & 0.768 & & PS-PPI & \textbf{4.22} & \underline{0.38} & \textbf{0.154} & \underline{4.46} & \underline{0.40} & 0.732 & \textbf{4.77} & \underline{0.55} & \textbf{0.148} & \underline{5.00} & \underline{0.57} & \underline{0.594} \\
& MI & 0.30 & 0.46 & 0.010 & 0.60 & 0.27 & 0.594 & 0.45 & 1.36 & 0.016 & 0.53 & 0.49 & 0.416 & & MI & 12.07 & 1.92 & 0.452 & 12.65 & 1.97 & 0.200 & 11.43 & 1.99 & 0.234 & 12.05 & 2.11 & 0.306 \\
& MI-RF & \underline{0.26} & \textbf{0.03} & 0.020 & 5.47 & 0.36 & 0.588 & \underline{0.06} & \textbf{0.07} & 0.012 & 5.35 & 0.39 & 0.302 & & MI-RF & \underline{5.28} & \textbf{0.35} & \underline{0.302} & \textbf{0.53} & \textbf{0.12} & 0.192 & \underline{5.74} & \textbf{0.45} & \underline{0.182} & \textbf{1.24} & \textbf{0.20} & 0.198 \\
\midrule
% ===================== Pair: MAR2 (left)  ||  MNAR5 (right) =====================
\multirow{10}{*}{\shortstack{MAR2 \\ $(X_1, Z_2)$}} & Full & 0.11 & 0.03 & 0.052 & 0.11 & 0.03 & 1.000 & 0.04 & 0.05 & 0.050 & 0.04 & 0.05 & 0.988 & \multirow{10}{*}{\shortstack{MNAR5 \\ $(X_1, X_2, Z_2, Y, X_1Y)$}} & Full & 0.11 & 0.03 & 0.052 & 0.11 & 0.03 & 1.000 & 0.04 & 0.05 & 0.050 & 0.04 & 0.05 & 0.988 \\
& WCCA (est) & 0.16 & 0.38 & \underline{0.046} & \underline{0.16} & 0.38 & 0.368 & 0.27 & 0.50 & \underline{0.056} & 0.27 & 0.50 & 0.312 & & WCCA (est) & 64.24 & 41.56 & 1.000 & 63.28 & 40.32 & 0.386 & 7.61 & 1.11 & 0.178 & 7.39 & 1.05 & 0.418 \\
& CCA & 0.22 & 0.20 & 0.058 & 0.22 & \underline{0.20} & 0.606 & 0.39 & 0.32 & 0.068 & 0.39 & \underline{0.32} & 0.402 & & CCA & 55.52 & 30.98 & 1.000 & 55.03 & 30.45 & 0.710 & 6.30 & 0.79 & 0.162 & 6.47 & 0.81 & 0.506 \\
& Naive & 14.17 & 2.05 & 1.000 & 14.17 & 2.05 & \textbf{1.000} & 11.15 & 1.32 & 0.990 & 11.15 & 1.32 & 0.000 & & Naive & \underline{16.40} & \underline{2.73} & 1.000 & \underline{16.40} & \underline{2.73} & \textbf{1.000} & 11.52 & 1.40 & 0.992 & 11.53 & 1.40 & 0.000 \\
& PPI & 15.18 & 2.46 & 0.956 & 15.18 & 2.46 & 0.814 & 10.67 & 1.40 & 0.570 & 10.67 & 1.40 & 0.650 & & PPI & \textbf{8.28} & \textbf{0.83} & \textbf{0.626} & \textbf{8.03} & \textbf{0.79} & \underline{0.836} & 41.04 & 17.12 & 1.000 & 40.35 & 16.56 & 0.000 \\
& PPI++ & 7.91 & 0.71 & 0.824 & 7.91 & 0.71 & 0.958 & 5.83 & 0.48 & 0.416 & 5.83 & 0.48 & \textbf{0.858} & & PPI++ & 27.60 & 7.71 & 1.000 & 27.28 & 7.53 & 0.000 & 26.84 & 7.37 & 1.000 & 26.48 & 7.17 & 0.000 \\
& SynSurr & 8.16 & 0.75 & 0.804 & 8.16 & 0.75 & \underline{0.970} & 5.97 & 0.50 & 0.384 & 5.97 & 0.50 & \textbf{0.858} & & SynSurr & 26.42 & 7.06 & 1.000 & 26.16 & 6.92 & 0.000 & 22.00 & 5.01 & 1.000 & 21.70 & 4.87 & \textbf{0.798} \\
& PS-PPI & \textbf{0.05} & \underline{0.13} & \textbf{0.050} & \textbf{0.05} & \textbf{0.13} & 0.768 & 0.24 & \underline{0.21} & \textbf{0.048} & \underline{0.24} & \textbf{0.21} & 0.656 & & PS-PPI & 22.28 & 5.23 & 0.996 & 22.94 & 5.51 & 0.000 & \underline{2.84} & \underline{0.41} & \underline{0.124} & \textbf{3.20} & \underline{0.42} & \underline{0.592} \\
& MI & 0.62 & 0.69 & 0.024 & 0.68 & 0.51 & 0.432 & \textbf{0.05} & 1.22 & 0.016 & \textbf{0.12} & 0.71 & 0.350 & & MI & 60.06 & 36.32 & 1.000 & 60.38 & 36.69 & 0.030 & 6.83 & 1.03 & 0.134 & 7.89 & 1.16 & 0.200 \\
& MI-RF & \underline{0.07} & \textbf{0.04} & 0.004 & 6.27 & 0.45 & 0.358 & \underline{0.05} & \textbf{0.07} & 0.010 & 5.08 & 0.37 & 0.340 & & MI-RF & 43.09 & 18.89 & \underline{0.894} & 40.35 & 16.66 & 0.064 & \textbf{2.08} & \textbf{0.13} & \textbf{0.016} & \underline{3.34} & \textbf{0.24} & 0.260 \\
\midrule
% ===================== Pair: MNAR1 (left)  ||  MNAR6 (right) =====================
\multirow{10}{*}{\shortstack{MNAR1 \\ $(Z_2, Y)$}} & Full & 0.11 & 0.03 & 0.052 & 0.11 & 0.03 & 1.000 & 0.04 & 0.05 & 0.050 & 0.04 & 0.05 & 0.988 & \multirow{10}{*}{\shortstack{MNAR6 \\ $(X_1, X_2, Z_2, Y, X_2Y)$}} & Full & 0.11 & 0.03 & 0.052 & 0.11 & 0.03 & 1.000 & 0.04 & 0.05 & 0.050 & 0.04 & 0.05 & 0.988 \\
& WCCA (est) & 0.44 & 0.15 & 0.030 & 0.47 & 0.15 & 0.766 & \underline{0.05} & 0.34 & \textbf{0.046} & 0.85 & 0.35 & 0.490 & & WCCA (est) & 8.70 & 1.19 & 0.278 & 8.78 & 1.17 & 0.396 & 68.19 & 46.98 & 1.000 & 67.81 & 46.45 & 0.430 \\
& CCA & 0.38 & 0.14 & \textbf{0.032} & 0.59 & 0.14 & 0.820 & 0.13 & 0.33 & \underline{0.044} & 0.98 & 0.35 & 0.530 & & CCA & 8.02 & 0.84 & 0.422 & 8.30 & 0.89 & 0.710 & 61.70 & 38.37 & 1.000 & 61.67 & 38.34 & 0.532 \\
& Naive & 15.52 & 2.45 & 1.000 & 15.54 & 2.45 & \textbf{1.000} & 12.27 & 1.58 & 0.994 & 12.31 & 1.59 & 0.000 & & Naive & 15.38 & 2.41 & 1.000 & 15.43 & 2.42 & \textbf{1.000} & \textbf{11.03} & \textbf{1.29} & 0.980 & \textbf{11.02} & \textbf{1.29} & 0.000 \\
& PPI & 2.16 & 0.17 & 0.074 & 1.39 & 0.14 & 0.850 & 3.12 & 0.32 & 0.088 & 2.07 & 0.27 & 0.616 & & PPI & 24.13 & 5.98 & 1.000 & 23.53 & 5.69 & 0.814 & \underline{20.20} & \underline{4.36} & \underline{0.968} & \underline{19.77} & \underline{4.18} & \underline{0.544} \\
& PPI++ & 1.21 & 0.09 & 0.070 & 1.09 & 0.08 & 0.968 & 2.02 & 0.18 & 0.080 & 1.68 & \underline{0.17} & \underline{0.826} & & PPI++ & 17.43 & 3.12 & 1.000 & 17.21 & 3.05 & 0.954 & 37.45 & 14.18 & 1.000 & 37.17 & 13.97 & 0.000 \\
& SynSurr & 1.02 & \underline{0.08} & 0.070 & 1.17 & \textbf{0.08} & \underline{0.986} & 1.80 & 0.17 & 0.072 & 1.73 & 0.17 & \textbf{0.858} & & SynSurr & 12.03 & 1.53 & 0.984 & 11.91 & 1.50 & \underline{0.968} & 35.15 & 12.50 & 1.000 & 34.95 & 12.36 & 0.000 \\
& PS-PPI & \underline{0.15} & 0.08 & \textbf{0.068} & \underline{0.32} & \underline{0.08} & 0.942 & 0.12 & \underline{0.16} & 0.060 & \underline{0.58} & \textbf{0.16} & 0.740 & & PS-PPI & \underline{3.55} & \underline{0.34} & \underline{0.166} & \underline{4.02} & \underline{0.37} & 0.690 & 22.42 & 5.33 & 0.984 & 22.99 & 5.58 & \textbf{0.604} \\
& MI & 0.19 & 0.51 & 0.010 & \textbf{0.29} & 0.26 & 0.650 & 0.29 & 1.34 & 0.010 & \textbf{0.45} & 0.50 & 0.470 & & MI & 7.99 & 0.92 & 0.236 & 8.61 & 1.02 & 0.254 & 65.26 & 43.25 & 0.988 & 66.15 & 44.22 & 0.084 \\
& MI-RF & \textbf{0.12} & \textbf{0.03} & 0.010 & 5.08 & 0.32 & 0.658 & \textbf{0.02} & \textbf{0.08} & 0.014 & 4.79 & 0.33 & 0.392 & & MI-RF & \textbf{3.27} & \textbf{0.16} & \textbf{0.080} & \textbf{2.02} & \textbf{0.13} & 0.350 & 46.87 & 22.41 & \textbf{0.912} & 45.47 & 21.14 & 0.044 \\
\midrule
% ===================== Pair: MNAR2 (left)  ||  MNAR7 (right) =====================
\multirow{10}{*}{\shortstack{MNAR2 \\ $(X_1, Z_2, Y)$}} & Full & 0.11 & 0.03 & 0.052 & 0.11 & 0.03 & 1.000 & 0.04 & 0.05 & 0.050 & 0.04 & 0.05 & 0.988 & \multirow{10}{*}{\shortstack{MNAR7 \\ $(X_1, X_2, Z_2, Y, Z_2Y)$}} & Full & 0.11 & 0.03 & 0.052 & 0.11 & 0.03 & 1.000 & 0.04 & 0.05 & 0.050 & 0.04 & 0.05 & 0.988 \\
& WCCA (est) & 11.51 & 1.73 & 0.420 & 12.31 & 1.90 & 0.468 & 0.14 & 0.53 & \underline{0.054} & 0.90 & 0.53 & 0.304 & & WCCA (est) & 15.80 & 2.98 & \underline{0.596} & 16.59 & 3.21 & 0.426 & 14.84 & 2.88 & 0.448 & 15.55 & 3.07 & 0.378 \\
& CCA & 13.94 & 2.17 & 0.840 & 14.70 & 2.38 & 0.734 & 0.23 & 0.34 & 0.064 & 1.09 & 0.36 & 0.470 & & CCA & 23.36 & 5.65 & 1.000 & 24.11 & 6.00 & 0.766 & 20.11 & 4.41 & 0.900 & 20.83 & 4.70 & \textbf{0.584} \\
& Naive & 14.47 & 2.14 & 1.000 & 14.50 & 2.14 & \textbf{1.000} & 11.40 & 1.37 & 0.990 & 11.44 & 1.38 & 0.000 & & Naive & 15.59 & 2.47 & 1.000 & 15.63 & 2.48 & \textbf{1.000} & 10.80 & 1.24 & 0.976 & 10.86 & 1.25 & 0.000 \\
& PPI & 22.56 & 5.26 & 1.000 & 21.75 & 4.90 & 0.726 & 10.41 & 1.33 & 0.532 & 9.44 & 1.14 & 0.520 & & PPI & 26.31 & 7.07 & 1.000 & 25.32 & 6.56 & 0.000 & 45.22 & 20.72 & 1.000 & 43.86 & 19.52 & 0.000 \\
& PPI++ & 19.37 & 3.84 & 1.000 & 19.15 & 3.76 & 0.000 & 6.65 & 0.59 & 0.462 & 6.37 & 0.55 & \underline{0.808} & & PPI++ & 25.26 & 6.46 & 1.000 & 24.89 & 6.27 & 0.000 & 36.35 & 13.37 & 1.000 & 35.76 & 12.94 & 0.000 \\
& SynSurr & 12.70 & 1.70 & 0.986 & 12.76 & 1.72 & \underline{0.958} & 5.97 & 0.50 & 0.378 & 5.93 & 0.49 & \textbf{0.836} & & SynSurr & 17.43 & 3.11 & 1.000 & 17.49 & 3.14 & \underline{0.978} & 27.18 & 7.54 & 1.000 & 27.04 & 7.46 & 0.000 \\
& PS-PPI & \textbf{4.28} & \underline{0.33} & \textbf{0.190} & \underline{4.65} & \underline{0.36} & 0.836 & 0.22 & \underline{0.22} & \textbf{0.050} & \underline{0.57} & \textbf{0.23} & 0.664 & & PS-PPI & \textbf{5.47} & \textbf{0.49} & \textbf{0.256} & \underline{5.72} & \underline{0.51} & 0.770 & \textbf{5.77} & \textbf{0.65} & \textbf{0.226} & \underline{6.05} & \underline{0.67} & \underline{0.576} \\
& MI & 11.33 & 1.87 & 0.334 & 12.90 & 2.05 & 0.222 & \textbf{0.02} & 0.41 & 0.046 & \textbf{0.10} & 0.42 & 0.294 & & MI & 18.16 & 3.66 & 0.866 & 17.96 & 3.58 & 0.148 & 14.18 & 2.55 & 0.516 & 14.10 & 2.57 & 0.194 \\
& MI-RF & \underline{5.03} & \textbf{0.32} & \underline{0.284} & \textbf{0.52} & \textbf{0.13} & 0.262 & \underline{0.04} & \textbf{0.08} & 0.006 & 4.98 & \underline{0.35} & 0.286 & & MI-RF & \underline{8.61} & \underline{0.84} & 0.604 & \textbf{5.40} & \textbf{0.48} & 0.114 & \underline{7.20} & \underline{0.65} & \underline{0.242} & \textbf{3.19} & \textbf{0.33} & 0.204 \\
\bottomrule
\end{tabular}}
\end{table}

\begin{landscape}

\begin{table}[]
\centering
\caption{Summary statistics for complete cases (CC) and total analytical sample across six biomarker traits in the \textit{All of Us} program data. \textcolor{black}{Sample sizes of the complete-case subsets and the total analytic cohorts are reported. The proportion of the complete-case subset within the total analytic cohort is shown in parentheses next to the complete-case sample size.} For categorical variables, values represent proportions with counts in parentheses. For continuous variables, mean (SD) are reported. Socioeconomic status (SES) summary statistics are obtained based on three digit zip code linkage by the All of Us program. Abbreviations of the biomarkers are as follows: hemoglobin A1c (HbA1c), alanine aminotransferase (ALT), white blood cell count (WBC), and C-reactive protein (CRP).}
\label{tab:summary_statistics_traits}
\resizebox{1.0\columnwidth}{!}{
\begin{tabular}{lcccccccccccc}
\toprule
 & \multicolumn{2}{c}{\textbf{WBC}} & \multicolumn{2}{c}{\textbf{HbA1c}} & \multicolumn{2}{c}{\textbf{ALT}} & \multicolumn{2}{c}{\textbf{Albumin}} & \multicolumn{2}{c}{\textbf{CRP}} & \multicolumn{2}{c}{\textbf{Vitamin D}} \\
\cmidrule(lr){2-3} \cmidrule(lr){4-5} \cmidrule(lr){6-7} \cmidrule(lr){8-9} \cmidrule(lr){10-11} \cmidrule(lr){12-13}
 & \textbf{CC} & \textbf{Total} & \textbf{CC} & \textbf{Total} & \textbf{CC} & \textbf{Total} & \textbf{CC} & \textbf{Total} & \textbf{CC} & \textbf{Total} & \textbf{CC} & \textbf{Total} \\
\midrule
\textbf{Sample size} & 46,226 (96.3\%) & 47,970 & 50,135 (45.3\%) & 110,781 & 66,501 (87.3\%) & 76,146 & 25,455 (87.7\%) & 29,027 & 25,455 (26.2\%) & 97,210 & 27,172 (24.7\%) & 109,951 \\
\addlinespace[0.5em]
\textbf{Age group} ($\%$ (count)) & & & & & & & & & & & & \\
\quad 18--24 years old & 0.7 (303) & 0.7 (314) & 0.3 (152) & 0.7 (728) & 0.6 (367) & 0.6 (428) & 0.7 (180) & 0.8 (222) & 0.7 (180) & 0.6 (616) & 0.3 (92) & 0.6 (712) \\
\quad 25--34 years old & 6.8 (3,125) & 6.7 (3,237) & 3.8 (1,928) & 6.8 (7,502) & 6.1 (4,034) & 6.2 (4,704) & 5.2 (1,312) & 5.4 (1,560) & 5.2 (1,312) & 6.4 (6,219) & 3.6 (986) & 6.7 (7,407) \\
\quad 35--44 years old & 12.2 (5,622) & 12.1 (5,817) & 8.9 (4,450) & 11.8 (13,104) & 11.1 (7,392) & 11.4 (8,694) & 10.1 (2,581) & 10.3 (2,986) & 10.1 (2,581) & 11.5 (11,144) & 7.9 (2,138) & 11.8 (12,937) \\
\quad 45--54 years old & 12.0 (5,555) & 12.0 (5,762) & 11.0 (5,531) & 12.1 (13,399) & 11.9 (7,887) & 12.1 (9,185) & 12.2 (3,112) & 12.3 (3,569) & 12.2 (3,112) & 12.0 (11,673) & 9.9 (2,685) & 12.1 (13,255) \\
\quad 55--64 years old & 17.7 (8,185) & 17.7 (8,512) & 19.0 (9,528) & 18.4 (20,370) & 18.5 (12,303) & 18.6 (14,175) & 19.2 (4,885) & 19.1 (5,554) & 19.2 (4,885) & 18.4 (17,895) & 17.6 (4,779) & 18.4 (20,244) \\
\quad 65+ years old & 50.7 (23,436) & 50.7 (24,328) & 56.9 (28,546) & 50.3 (55,678) & 51.9 (34,518) & 51.2 (38,960) & 52.6 (13,385) & 52.1 (15,136) & 52.6 (13,385) & 51.1 (49,663) & 60.7 (16,492) & 50.4 (55,396) \\
\addlinespace[0.5em]
\textbf{Sex at birth} ($\%$ (count)) & & & & & & & & & & & & \\
\quad Female & 61.4 (28,360) & 61.3 (29,404) & 53.5 (26,829) & 59.6 (66,048) & 56.1 (37,325) & 57.2 (43,557) & 62.4 (15,876) & 63.0 (18,282) & 62.4 (15,876) & 59.6 (57,908) & 65.7 (17,839) & 59.7 (65,609) \\
\quad Male & 38.6 (17,866) & 38.7 (18,566) & 46.5 (23,306) & 40.4 (44,733) & 43.9 (29,176) & 42.8 (32,589) & 37.6 (9,579) & 37.0 (10,745) & 37.6 (9,579) & 40.4 (39,302) & 34.3 (9,333) & 40.3 (44,342) \\
\addlinespace[0.5em]
\textbf{SES} & & & & & & & & & & & & \\
\quad Median income (\$) & 64264.1 (14541.0) & 64084.5 (14637.9) & 69767.4 (18739.8) & 67661.8 (17275.7) & 70063.4 (19427.9) & 69196.3 (18662.4) & 64646.7 (15537.1) & 64820.4 (15301.6) & 64646.7 (15537.1) & 67276.2 (17696.9) & 70494.3 (20241.9) & 67608.9 (17252.7) \\
\quad No health insurance (\%) & 7.3 (3.1) & 7.4 (3.2) & 7.4 (3.4) & 7.7 (3.6) & 8.4 (3.9) & 8.3 (3.8) & 7.9 (3.3) & 8.0 (3.4) & 7.9 (3.3) & 7.8 (3.7) & 6.1 (2.9) & 7.7 (3.6) \\
\textbf{Body measurements} & & & & & & & & & & & & \\
\quad BMI (kg/m$^2$) & 29.7 (7.3) & 29.7 (7.2) & 30.6 (7.4) & 29.5 (7.2) & 29.3 (7.2) & 29.4 (7.2) & 30.5 (7.8) & 30.4 (7.8) & 30.5 (7.8) & 29.6 (7.3) & 29.6 (7.2) & 29.5 (7.2) \\
\quad Body weight (kg) & 84.8 (22.7) & 84.7 (22.7) & 87.9 (23.1) & 84.3 (22.4) & 84.4 (22.4) & 84.7 (22.6) & 86.2 (23.6) & 86.1 (23.6) & 86.2 (23.6) & 84.5 (22.5) & 83.4 (22.9) & 84.3 (22.4) \\
\quad Body height (cm) & 169.7 (9.9) & 169.7 (9.9) & 170.6 (10.2) & 169.9 (9.9) & 170.3 (10.0) & 170.3 (10.0) & 169.4 (10.1) & 169.4 (10.0) & 169.4 (10.1) & 169.8 (9.9) & 168.8 (9.7) & 169.9 (9.9) \\
\addlinespace[0.5em]
\textbf{Depression} ($\%$ (count)) & 38.1 (17,600) & 37.7 (18,085) & 38.7 (19,416) & 34.0 (37,693) & 35.6 (23,675) & 35.0 (26,678) & 46.3 (11,775) & 44.1 (12,802) & 46.3 (11,775) & 36.2 (35,230) & 43.1 (11,712) & 34.2 (37,591) \\
\addlinespace[0.5em]
\textbf{Hypertension} ($\%$ (count)) & 52.5 (24,282) & 52.5 (25,189) & 63.8 (32,007) & 51.4 (56,892) & 55.1 (36,647) & 54.5 (41,512) & 61.8 (15,725) & 60.8 (17,645) & 61.8 (15,725) & 53.2 (51,721) & 61.3 (16,666) & 51.6 (56,722) \\
\addlinespace[0.5em]
\textbf{Diabetes} ($\%$ (count)) & 20.7 (9,554) & 20.6 (9,881) & 32.0 (16,024) & 19.8 (21,966) & 22.4 (14,874) & 21.6 (16,473) & 28.0 (7,123) & 27.2 (7,882) & 28.0 (7,123) & 21.1 (20,515) & 24.9 (6,778) & 19.9 (21,905) \\
\addlinespace[0.5em]
\textbf{Anxiety} ($\%$ (count)) & 42.2 (19,506) & 41.9 (20,092) & 42.0 (21,033) & 38.8 (43,021) & 39.8 (26,458) & 39.5 (30,042) & 51.1 (13,014) & 49.0 (14,214) & 51.1 (13,014) & 41.0 (39,862) & 47.7 (12,965) & 39.0 (42,832) \\
\addlinespace[0.5em]
\textbf{Obesity} ($\%$ (count)) & 38.3 (17,724) & 38.1 (18,278) & 46.7 (23,436) & 35.5 (39,275) & 37.0 (24,613) & 36.9 (28,080) & 46.4 (11,802) & 44.6 (12,937) & 46.4 (11,802) & 37.4 (36,359) & 45.1 (12,255) & 35.5 (39,069) \\
\addlinespace[0.5em]
\bottomrule
\end{tabular}}
\end{table}

\end{landscape}

\begin{table}[H]
\centering
\caption{Summary of genome-wide association studies (GWAS) results across different traits including Hemoglobin A1c (HbA1c), alanine aminotransferase (ALT), white blood cell (WBC) count, albumin, C-reactive protein (CRP), and Vitamin D. For each trait, we report the total analytical sample size, which includes complete and incomplete sample sizes, and the machine learning (ML) fitting sample size. We report the predictive performance ($R^2$) of ML models in each trait. We report the source of the reference GWAS results, their total single-nucleotide polymorphisms (SNPs) and significant SNPs. Finally, we report the number of significant SNPs from each included method and the number of rediscovered SNPs, where rediscovery denotes overlap with the reference GWAS. The significance threshold are $p < 5 \times 10^{-8}$. Abbreviations of the biomarkers not mentioned above are as follows: aspartate aminotransferase (AST).}
\label{tab:gwas_summary}
\resizebox{1.0\textwidth}{!}{
\begin{tabular}{lcccccc}
\toprule
 & \textbf{HbA1c} & \textbf{ALT} & \textbf{WBC} & \textbf{Albumin} & \textbf{CRP} & \textbf{Vitamin D} \\
\midrule
\textbf{Sample sizes} & & & & & & \\
\quad Complete & 50{,}135 & 66{,}501 & 46{,}226 & 25{,}455 & 25{,}455 & 27{,}172 \\
\quad Incomplete & 60{,}646 & 9{,}645 & 1{,}744 & 3{,}572 & 71{,}755 & 82{,}779 \\
\quad Total analytical & 110{,}781 & 76{,}146 & 47{,}970 & 29{,}027 & 97{,}210 & 109{,}951 \\
\quad ML fitting & 39{,}556 & 57{,}616 & 34{,}206 & 17{,}523 & 17{,}523 & 13{,}685 \\
\addlinespace[0.5em]
\textbf{Predictive biomarker} & Glucose & AST & Neutrophils & CRP & Albumin & Calcium \\
\textbf{ML performance ($R^2$)} & 0.5051 & 0.5790 & 0.5532 & 0.2109 & 0.2362 & 0.2376 \\
\addlinespace[0.5em]
\textbf{Reference GWAS} & \citet{loya2025scalable} & \citet{verma2024diversity} & \citet{verma2024diversity} & \citet{verma2024diversity} & \citet{loya2025scalable} & \citet{manousaki2020genome} \\
\quad Total SNPs & 13{,}308{,}322 & 19{,}725{,}446 & 19{,}737{,}397 & 19{,}744{,}728 & 13{,}308{,}322 & 16{,}668{,}957 \\
\quad Significant SNPs & 102{,}583 & 30{,}365 & 128{,}903 & 20{,}301 & 57{,}330 & 15{,}445 \\
\addlinespace[0.5em]
\textbf{CCA} & & & & & & \\
\quad Significant SNPs & 59 & 6 & 10 & 0 & 1 & 13 \\
\quad Rediscovered SNPs & 53 & 6 & 10 & 0 & 1 & 13 \\
\addlinespace[0.5em]
\textbf{WCCA} & & & & & & \\
\quad Significant SNPs & 61 & 2 & 14 & 0 & 2 & 9 \\
\quad Rediscovered SNPs & 53 & 2 & 14 & 0 & 2 & 9 \\
\addlinespace[0.5em]
\textbf{PS-PPI} & & & & & & \\
\quad Significant SNPs & 54 & 6 & 10 & 0 & 4 & 45 \\
\quad Rediscovered SNPs & 45 & 6 & 10 & 0 & 4 & 12 \\
\addlinespace[0.5em]
\textbf{SynSurr} & & & & & & \\
\quad Significant SNPs & 65 & 11 & 10 & 0 & 3 & 13 \\
\quad Rediscovered SNPs & 55 & 11 & 10 & 0 & 3 & 13 \\
\bottomrule
\end{tabular}}
\end{table}

\appendix

\renewcommand{\thesection}{S\arabic{section}}
\renewcommand\thefigure{S\arabic{figure}}
\renewcommand\thetable{S\arabic{table}}
\numberwithin{equation}{section}
\makeatletter
\renewcommand\theequation{\thesection.\arabic{equation}}
% "activate" the preparatory code, but for section-level headers only
\newcommand{\section@cntformat}{Supplement \thesection:\ }
\makeatother

\newpage
\section*{Supplementary Materials}

\setcounter{figure}{0}
\setcounter{table}{0}

\begin{figure}[H]
\centering
\includegraphics[width=0.9\linewidth]{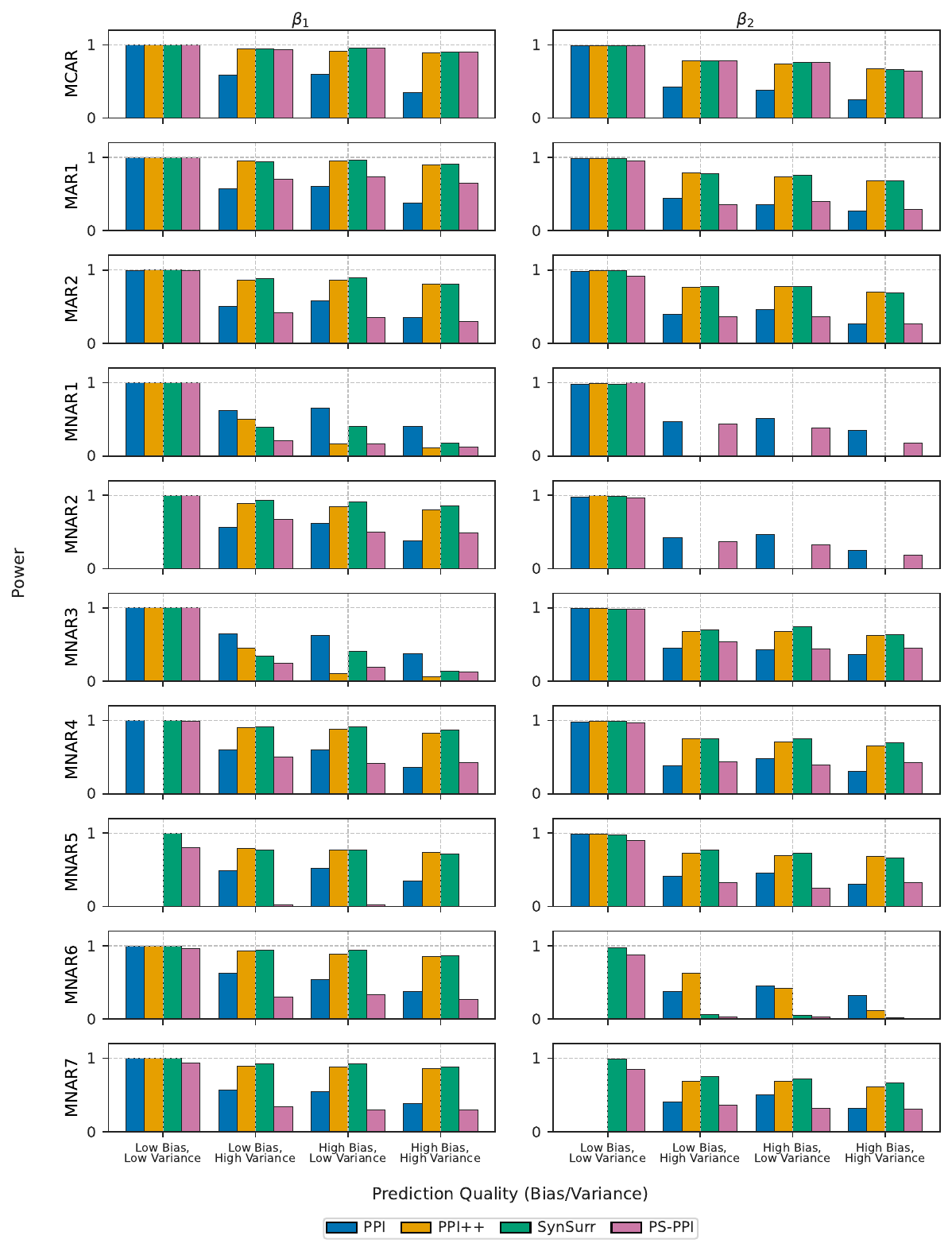}
\caption{Mean power of the included prediction-based inference methods for $\beta_1$ and $\beta_2$ under different quality of imputations in the continuous-confounder linear regression setting. Each row corresponds to a outcome observation model. Within each subplot, the y-axis represents the power of the selected methods, and the x-axis represents method groups under different levels of imputation quality.}
\label{fig:linear_regression_simulation_varying_quality_power}
\end{figure}

\begin{figure}[H]
\centering
\includegraphics[width=0.9\linewidth]{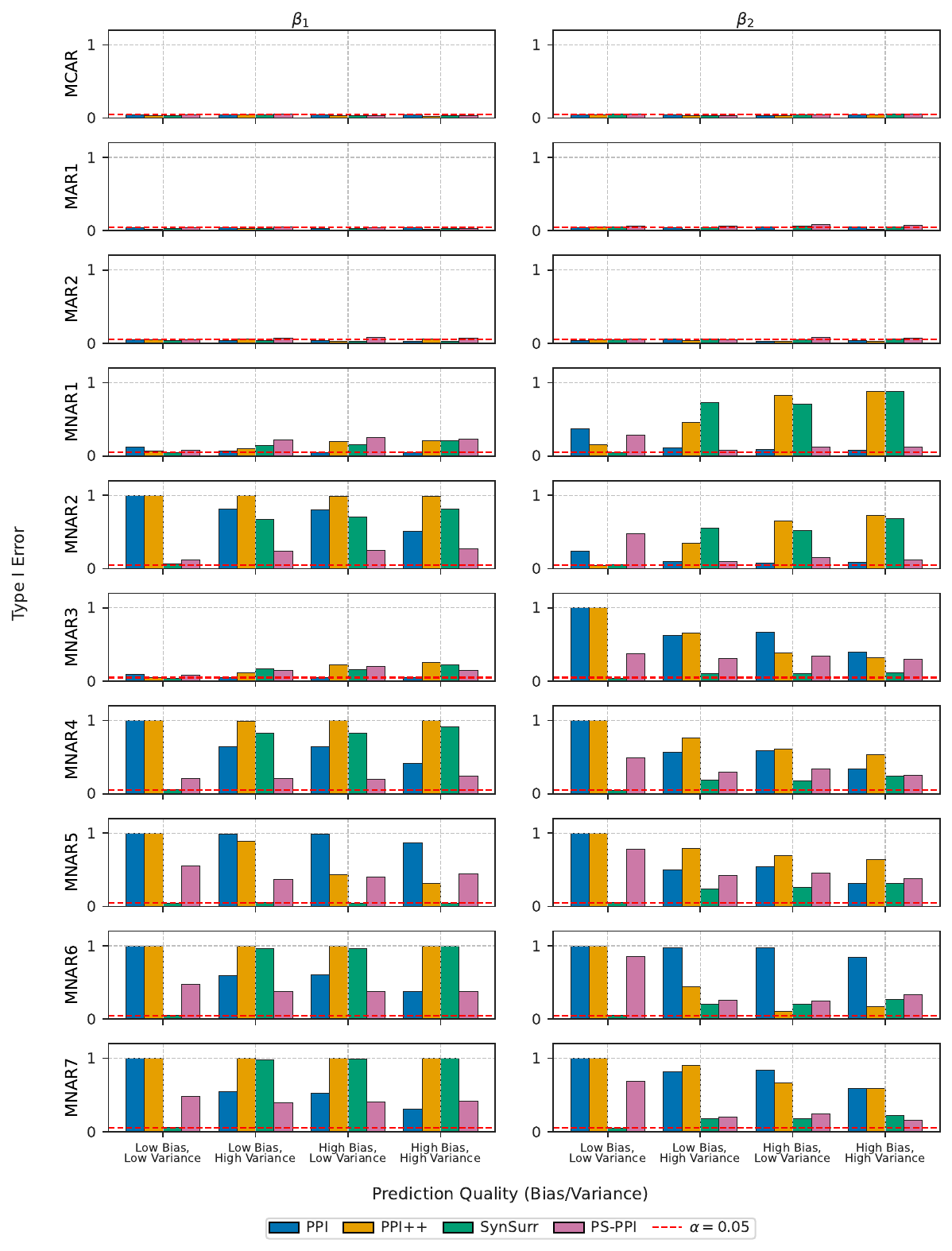}
\caption{Type I error of the included prediction-based inference methods for $\beta_1$ and $\beta_2$ under different quality of imputations in the continuous-confounder linear regression setting. Each row corresponds to a outcome observation model. Within each subplot, the y-axis represents the type I error of the selected methods, and the x-axis represents method groups under different levels of imputation quality.}
\label{fig:linear_regression_simulation_varying_quality_type_i_error}
\end{figure}

\begin{figure}[H]
\centering
\includegraphics[width=0.9\linewidth]{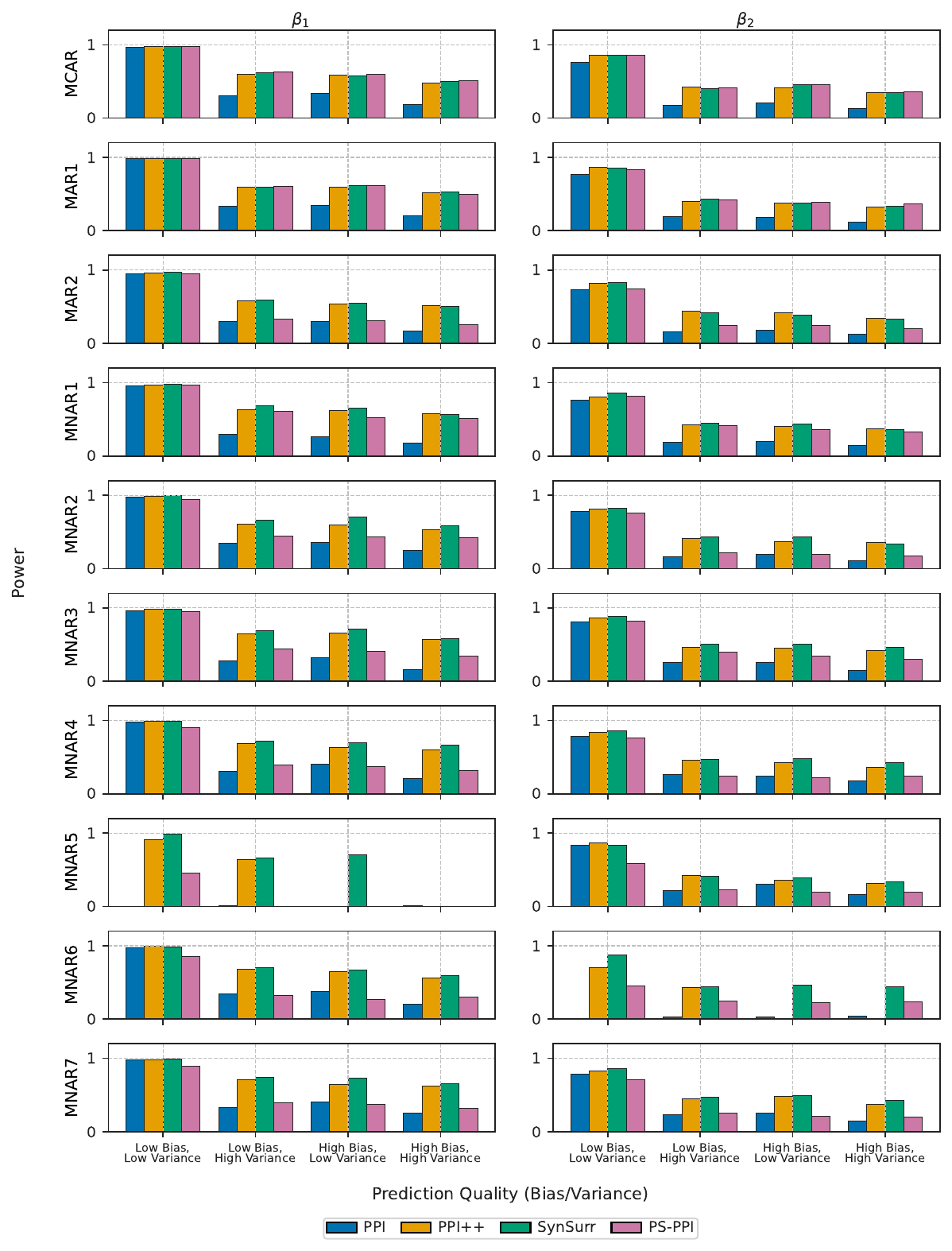}
\caption{Power of the included prediction-based inference methods for $\beta_1$ and $\beta_2$ under different quality of imputations in the categorical-confounder linear regression setting. Each row corresponds to a outcome observation model. Within each subplot, the y-axis represents the power of the selected methods, and the x-axis represents method groups under different levels of imputation quality.}
\label{fig:linear_regression_simulation_binary_z_varying_quality_power}
\end{figure}

\begin{figure}[H]
\centering
\includegraphics[width=0.9\linewidth]{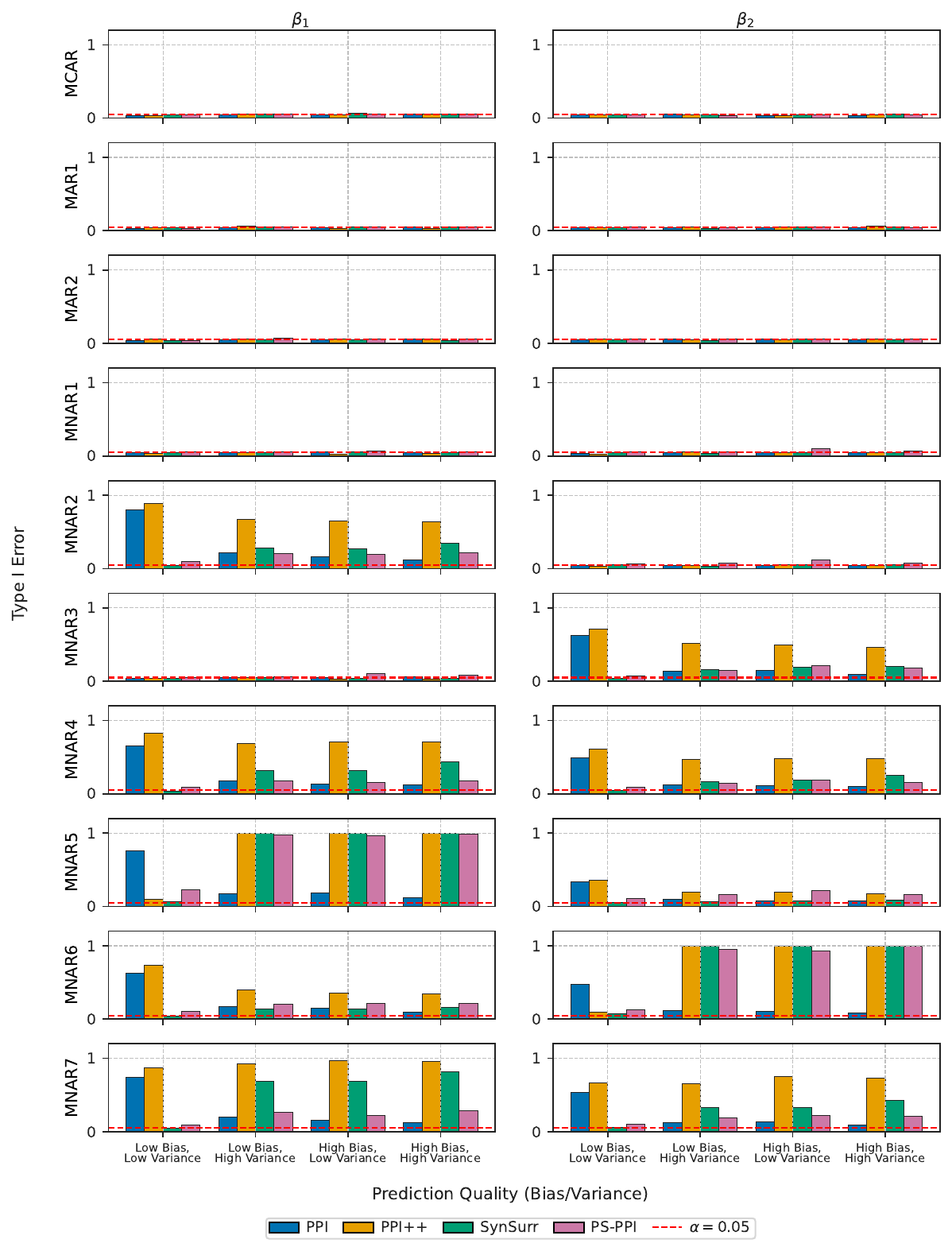}
\caption{Type I error of the included prediction-based inference methods for $\beta_1$ and $\beta_2$ under different quality of imputations in the categorical-confounder linear regression setting. Each row corresponds to a outcome observation model. Within each subplot, the y-axis represents the type I error of the selected methods, and the x-axis represents method groups under different levels of imputation quality.}
\label{fig:linear_regression_simulation_binary_z_varying_quality_type_i_error}
\end{figure}

\begin{figure}[H]
\centering
\includegraphics[width=0.9\linewidth]{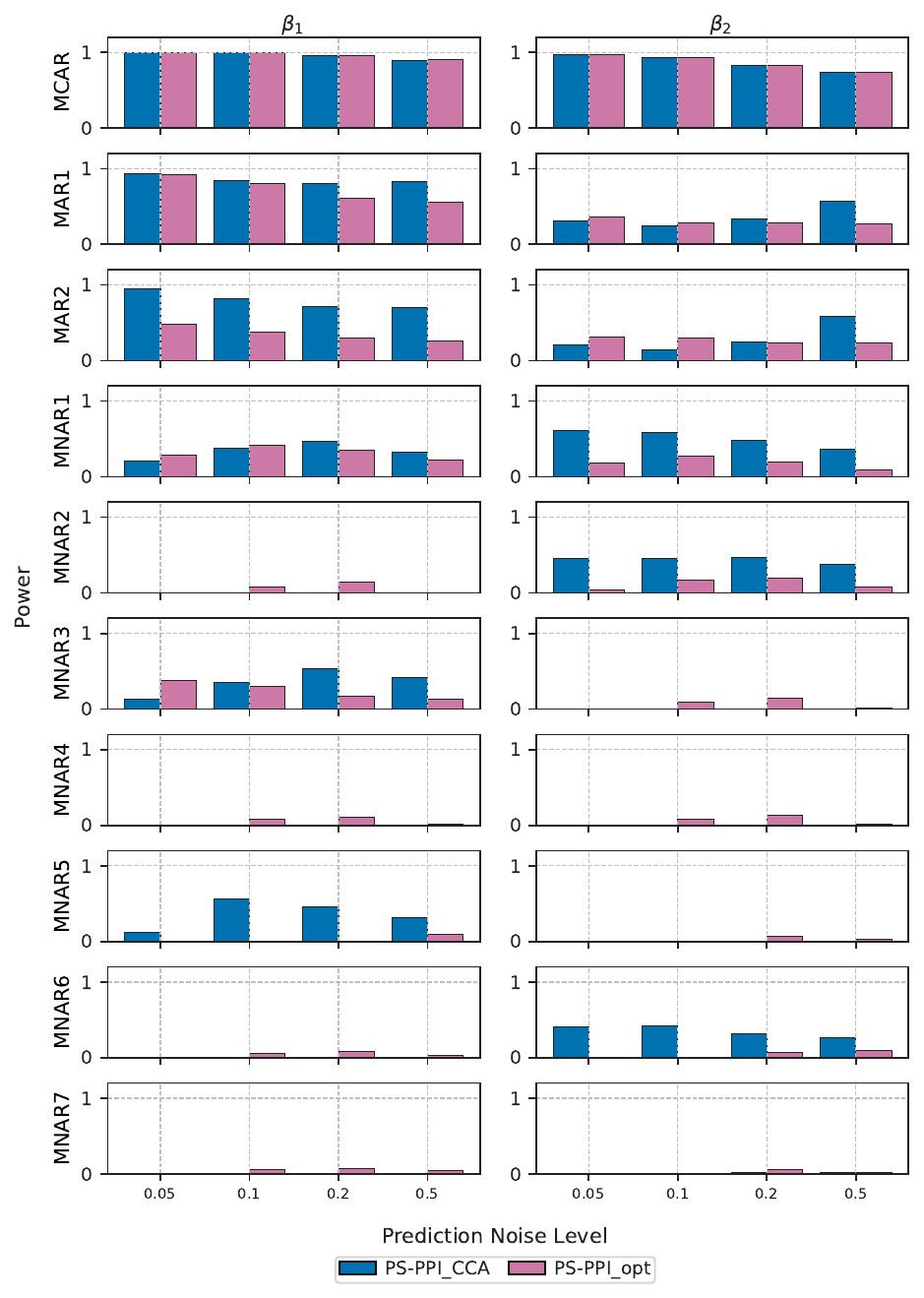}
\caption{Power of the included prediction-based inference methods for $\beta_1$ and $\beta_2$ under different quality of imputations in the logistic regression setting. Each row corresponds to a outcome observation model. Within each subplot, the y-axis represents the power of the selected methods, and the x-axis represents method groups under different levels of imputation quality.}
\label{fig:logistic_regression_varying_quality_power}
\end{figure}

\begin{figure}[H]
\centering
\includegraphics[width=0.9\linewidth]{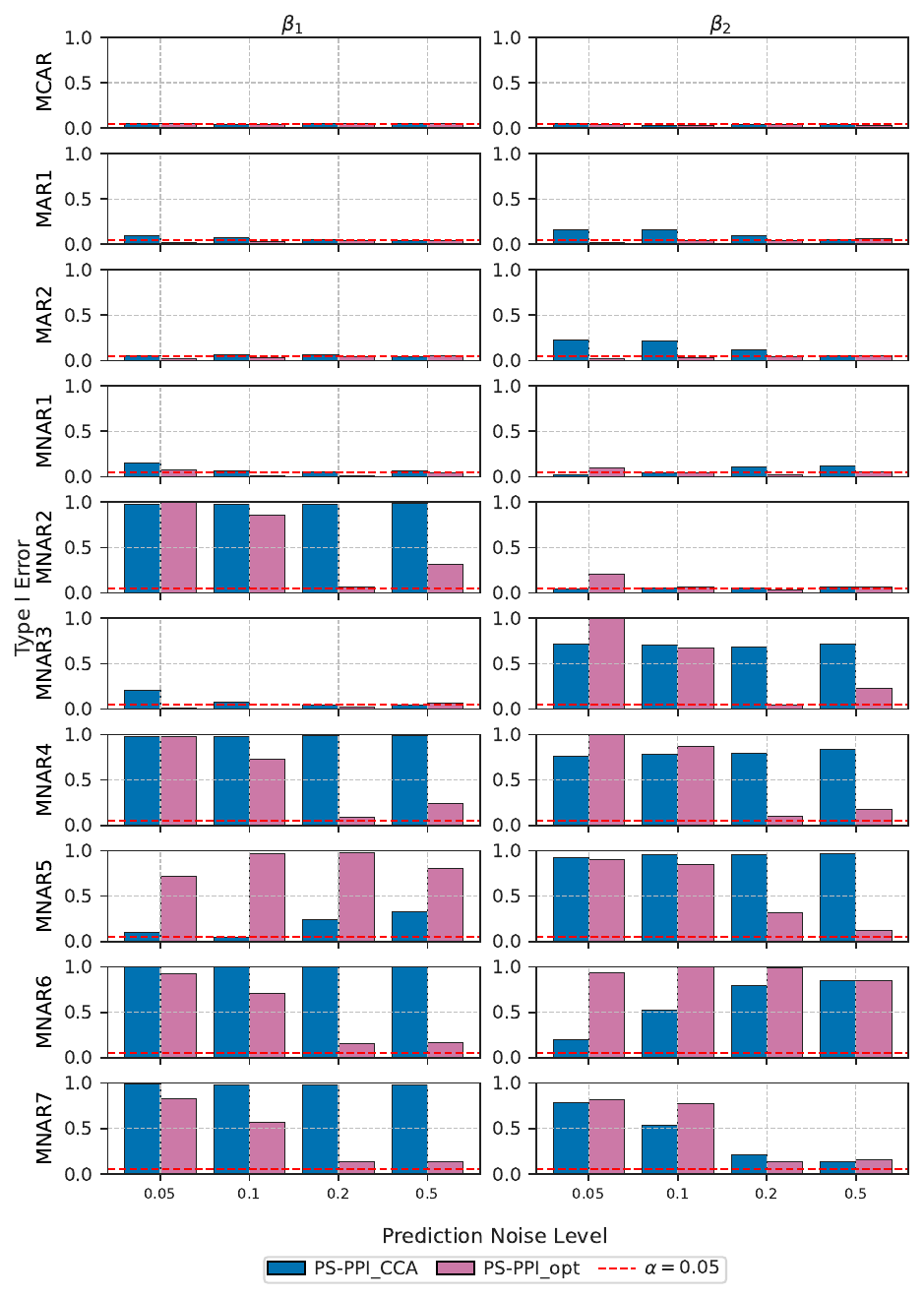}
\caption{Type I error of the included prediction-based inference methods for $\beta_1$ and $\beta_2$ under different quality of imputations in the logistic regression setting. Each row corresponds to a outcome observation model. Within each subplot, the y-axis represents the type I error of the selected methods, and the x-axis represents method groups under different levels of imputation quality.}
\label{fig:logistic_regression_varying_quality_type_i_error}
\end{figure}

% \begin{figure}[H]
%     \centering
%     \includegraphics[width=0.85\linewidth]{figs/linear_regression_tune_omega.pdf}
%     \caption{Type I error and power of the included methods for $\beta_1$ and $\beta_2$ under varying main effects in the selection models. Within each subplot, the X-axis represents the magnitude of the main effect of $X_1$ and $X_2$ in the selection model. The Y-axis represents the type I error and power in the first and second row, respectively.}
%     \label{fig:linear_regression_tune_omega}
% \end{figure}

% \begin{table}[!t]
% \centering
% \caption{Proportion of fully observed data under different $\omega$ in the linear regression simulation. The results are averaged across 500 replicates of simulation.}
% \label{tab:linear_regression_proportion_of_observed}
% \resizebox{1.0\columnwidth}{!}{
% \begin{tabular}{ccccccccccc}
% \toprule
%  & MCAR & MAR1 & MAR2 & MNAR1 & MNAR2 & MNAR3 & MNAR4 & MNAR5 & MNAR6 & MNAR7 \\
%  & $(\varnothing)$ & $(Z_{2})$ & $(X_{1},Z_{2})$ & $(Z_{2},Y)$ & $(X_{1},Z_{2},Y)$ & $(X_{2},Z_{2},Y)$ & $(X_{1},X_{2},Z_{2},Y)$ & $(X_{1},X_{2},Z_{2},Y,X_{1}Y)$ & $(X_{1},X_{2},Z_{2},Y,X_{2}Y)$ & $(X_{1},X_{2},Z_{2},Y,Z_{2}Y)$ \\
% \midrule
% $\omega = -1.5$ & 18.3\% & 20.9\% & 19.9\% & 15.3\% & 21.7\% & 23.2\% & 23.8\% & 23.2\% & 22.7\% & 23.9\% \\
% $\omega = 0.0$  & 49.9\% & 52.4\% & 49.9\% & 44.2\% & 49.2\% & 46.2\% & 48.5\% & 48.1\% & 48.6\% & 48.8\% \\
% $\omega = 1.5$  & 81.6\% & 82.2\% & 80.0\% & 77.4\% & 77.2\% & 70.1\% & 73.8\% & 73.9\% & 76.0\% & 74.1\% \\
% \bottomrule
% \end{tabular}}
% \end{table}

\begin{table}[H]
\centering
\caption{Simulation results under the logistic regression settings. \reddown; lower is better; \greenup; higher is better; $\rightarrow x$ indicates that values closer to $x$ are preferable. The best-performing results are \textbf{bolded}, and the second-best-performing results are \underline{underlined}. Abbreviations: mean squared error (MSE) and type I error (Type I). All bias and MSE values are scaled by a factor of 100 for ease of presentation.}
\label{tab:logistic_regression_summary}
\resizebox{\textwidth}{!}{
\begin{tabular}{@{}%
c|l|ccc|ccc|ccc|ccc||%
c|l|ccc|ccc|ccc|ccc@{}}
\toprule
& &
\multicolumn{6}{c|}{$\beta_{1}$} &
\multicolumn{6}{c||}{$\beta_{2}$} &
& &
\multicolumn{6}{c|}{$\beta_{1}$} &
\multicolumn{6}{c}{$\beta_{2}$} \\ \midrule
\multirow{3}{*}{\shortstack{\textbf{Missing}\\\textbf{Mechanism}}} & \multirow{3}{*}{\textbf{Method}} &
\multicolumn{3}{c}{$\beta_{1}=0,\,N=50{,}000$} &
\multicolumn{3}{c|}{$\beta_{1}=0.1,\,N=50{,}000$} &
\multicolumn{3}{c}{$\beta_{2}=0,\,N=50{,}000$} &
\multicolumn{3}{c||}{$\beta_{2}=0.1,\,N=50{,}000$} &
\multirow{3}{*}{\shortstack{\textbf{Missing}\\\textbf{Mechanism}}} & \multirow{3}{*}{\textbf{Method}} &
\multicolumn{3}{c}{$\beta_{1}=0,\,N=10{,}000$} &
\multicolumn{3}{c|}{$\beta_{1}=0.1,\,N=1{,}000$} &
\multicolumn{3}{c}{$\beta_{2}=0,\,N=10{,}000$} &
\multicolumn{3}{c}{$\beta_{2}=0.1,\,N=1{,}000$} \\
& &
\textbf{$\lvert$Bias$_0\rvert$} & \textbf{MSE$_0$} & \textbf{Type I} &
\textbf{$\lvert$Bias$\rvert$} & \textbf{MSE} & \textbf{Power} &
\textbf{$\lvert$Bias$_0\rvert$} & \textbf{MSE$_0$} & \textbf{Type I} &
\textbf{$\lvert$Bias$\rvert$} & \textbf{MSE} & \textbf{Power} &
& &
\textbf{$\lvert$Bias$_0\rvert$} & \textbf{MSE$_0$} & \textbf{Type I} &
\textbf{$\lvert$Bias$\rvert$} & \textbf{MSE} & \textbf{Power} &
\textbf{$\lvert$Bias$_0\rvert$} & \textbf{MSE$_0$} & \textbf{Type I} &
\textbf{$\lvert$Bias$\rvert$} & \textbf{MSE} & \textbf{Power} \\
& &
$\rightarrow 0$ & \reddown & $\rightarrow 0.05$ &
$\rightarrow 0$ & \reddown & \greenup &
$\rightarrow 0$ & \reddown & $\rightarrow 0.05$ &
$\rightarrow 0$ & \reddown & \greenup &
& &
$\rightarrow 0$ & \reddown & $\rightarrow 0.05$ &
$\rightarrow 0$ & \reddown & \greenup &
$\rightarrow 0$ & \reddown & $\rightarrow 0.05$ &
$\rightarrow 0$ & \reddown & \greenup \\
\midrule
% ===================== Pair: MCAR (left)  ||  MNAR3 (right) =====================
\multirow{8}{*}{\shortstack{MCAR \\ $(\varnothing)$}} & Full & 0.06 & 0.02 & 0.042 & 0.02 & 0.02 & 1.000 & 0.07 & 0.03 & 0.038 & 0.06 & 0.03 & 1.000 & \multirow{8}{*}{\shortstack{MNAR3 \\ $(X_2, Z_2, Y)$}} & Full & 0.06 & 0.02 & 0.042 & 0.02 & 0.02 & 1.000 & 0.07 & 0.03 & 0.038 & 0.06 & 0.03 & 1.000 \\
& WCCA (est) & \underline{0.01} & 0.10 & 0.062 & 0.12 & 0.10 & 0.870 & \underline{0.03} & 0.17 & \textbf{0.050} & 0.05 & 0.17 & 0.692 & & WCCA (est) & 0.30 & 1.25 & \textbf{0.050} & \underline{0.22} & 1.35 & 0.142 & 16.29 & 4.65 & \underline{0.220} & 15.94 & 4.52 & 0.006 \\
& CCA & \textbf{0.01} & 0.10 & 0.064 & 0.13 & 0.10 & 0.866 & \textbf{0.03} & 0.17 & \underline{0.046} & \underline{0.05} & 0.17 & 0.700 & & CCA & 0.30 & \underline{0.31} & 0.040 & 0.39 & \underline{0.32} & 0.422 & 21.60 & 5.40 & 0.726 & 21.44 & 5.34 & 0.000 \\
& Naive & 0.03 & \textbf{0.04} & 0.136 & \textbf{0.02} & \textbf{0.03} & \textbf{1.000} & 0.05 & \textbf{0.05} & 0.132 & \textbf{0.02} & \textbf{0.05} & \textbf{0.996} & & Naive & \textbf{0.20} & \textbf{0.04} & 0.162 & 0.23 & \textbf{0.04} & \textbf{0.998} & \underline{12.02} & \underline{1.50} & 1.000 & \underline{12.01} & \underline{1.50} & \textbf{0.990} \\
& PS-PPI (CCA) & 0.05 & 0.10 & \underline{0.046} & \underline{0.10} & \underline{0.10} & \underline{0.904} & 0.16 & \underline{0.16} & 0.032 & 0.18 & \underline{0.17} & \underline{0.732} & & PS-PPI (CCA) & \underline{0.24} & 0.34 & \underline{0.048} & \textbf{0.19} & 0.35 & \underline{0.432} & 21.86 & 5.54 & 0.732 & 21.66 & 5.48 & 0.004 \\
& PS-PPI & 0.05 & \underline{0.10} & \textbf{0.050} & 0.10 & 0.10 & \underline{0.904} & 0.15 & 0.16 & 0.030 & 0.18 & 0.17 & 0.728 & & PS-PPI & 0.31 & 1.28 & 0.060 & 0.52 & 1.29 & 0.132 & 17.43 & 5.13 & 0.224 & 16.96 & 5.06 & 0.010 \\
& MI & 8.44 & 0.81 & 0.570 & 7.56 & 0.67 & 0.194 & 6.07 & 0.54 & 0.252 & 4.98 & 0.42 & 0.348 & & MI & 22.57 & 5.41 & 0.874 & 20.87 & 4.72 & 0.138 & \textbf{2.54} & \textbf{0.73} & \textbf{0.078} & \textbf{4.46} & \textbf{0.88} & \underline{0.082} \\
& MI-RF & 2.00 & 0.20 & 0.334 & 3.22 & 0.26 & 0.700 & 2.74 & 0.37 & 0.322 & 4.86 & 0.54 & 0.466 & & MI-RF & 3.34 & 1.45 & 0.286 & 4.35 & 1.62 & 0.130 & 20.93 & 7.38 & 0.542 & 19.86 & 6.87 & 0.038 \\
\midrule
% ===================== Pair: MAR1 (left)  ||  MNAR4 (right) =====================
\multirow{8}{*}{\shortstack{MAR1 \\ $(Z_2)$}} & Full & 0.06 & 0.02 & 0.042 & 0.02 & 0.02 & 1.000 & 0.07 & 0.03 & 0.038 & 0.06 & 0.03 & 1.000 & \multirow{8}{*}{\shortstack{MNAR4 \\ $(X_1, X_2, Z_2, Y)$}} & Full & 0.06 & 0.02 & 0.042 & 0.02 & 0.02 & 1.000 & 0.07 & 0.03 & 0.038 & 0.06 & 0.03 & 1.000 \\
& WCCA (est) & 0.20 & 0.23 & \underline{0.048} & 0.30 & 0.23 & 0.522 & 0.84 & 0.58 & 0.040 & 0.88 & 0.57 & 0.220 & & WCCA (est) & 14.81 & 5.29 & \underline{0.214} & 13.95 & 5.21 & 0.010 & 15.38 & 5.76 & \underline{0.152} & 14.39 & 5.45 & 0.014 \\
& CCA & \underline{0.11} & 0.13 & \textbf{0.052} & 0.24 & 0.13 & 0.752 & 0.37 & \underline{0.20} & \textbf{0.046} & 0.37 & \underline{0.20} & 0.578 & & CCA & 30.80 & 9.91 & 0.996 & 30.34 & 9.66 & 0.000 & 25.20 & 7.17 & 0.830 & 24.68 & 6.92 & 0.000 \\
& Naive & \textbf{0.00} & \textbf{0.04} & 0.144 & \textbf{0.05} & \textbf{0.04} & \textbf{1.000} & \textbf{0.04} & \textbf{0.06} & 0.144 & \textbf{0.00} & \textbf{0.06} & \textbf{0.984} & & Naive & \underline{8.23} & \textbf{0.72} & 0.994 & \underline{8.18} & \textbf{0.71} & \textbf{0.998} & \underline{10.47} & \textbf{1.15} & 1.000 & \underline{10.40} & \textbf{1.14} & \textbf{0.994} \\
& PS-PPI (CCA) & 0.16 & \underline{0.12} & 0.042 & \underline{0.20} & \underline{0.13} & \underline{0.828} & \underline{0.09} & 0.23 & \underline{0.058} & \underline{0.08} & 0.23 & \underline{0.582} & & PS-PPI (CCA) & 30.74 & 9.92 & 0.992 & 30.08 & 9.52 & 0.000 & 25.63 & 7.35 & 0.832 & 25.03 & 7.10 & 0.002 \\
& PS-PPI & 0.20 & 0.22 & 0.046 & 0.25 & 0.22 & 0.564 & 0.34 & 0.64 & 0.060 & 0.38 & 0.65 & 0.262 & & PS-PPI & 15.00 & 5.55 & 0.240 & 14.49 & 5.48 & 0.014 & 15.85 & 5.92 & 0.174 & 15.11 & 5.93 & 0.018 \\
& MI & 13.35 & 1.92 & 0.854 & 11.99 & 1.56 & 0.094 & 18.21 & 3.49 & 0.918 & 16.51 & 2.90 & 0.116 & & MI & \textbf{7.31} & \underline{1.16} & \textbf{0.120} & \textbf{4.78} & \underline{0.86} & \underline{0.036} & \textbf{6.51} & \underline{1.55} & \textbf{0.112} & \textbf{4.88} & \underline{1.28} & \underline{0.052} \\
& MI-RF & 0.64 & 0.35 & 0.314 & 1.75 & 0.38 & 0.314 & 10.15 & 1.86 & 0.554 & 11.40 & 2.18 & 0.190 & & MI-RF & 19.14 & 5.49 & 0.552 & 18.09 & 5.22 & 0.028 & 22.03 & 8.52 & 0.518 & 20.50 & 8.11 & 0.022 \\
\midrule
% ===================== Pair: MAR2 (left)  ||  MNAR5 (right) =====================
\multirow{8}{*}{\shortstack{MAR2 \\ $(X_1, Z_2)$}} & Full & 0.06 & 0.02 & 0.042 & 0.02 & 0.02 & 1.000 & 0.07 & 0.03 & 0.038 & 0.06 & 0.03 & 1.000 & \multirow{8}{*}{\shortstack{MNAR5 \\ $(X_1, X_2, Z_2, Y, X_1Y)$}} & Full & 0.06 & 0.02 & 0.042 & 0.02 & 0.02 & 1.000 & 0.07 & 0.03 & 0.038 & 0.06 & 0.03 & 1.000 \\
& WCCA (est) & \underline{0.06} & 0.57 & 0.068 & 0.20 & 0.56 & 0.238 & 0.98 & 0.97 & 0.076 & 0.83 & 0.91 & 0.130 & & WCCA (est) & 71.57 & 62.18 & 0.818 & 74.17 & 67.06 & 0.124 & 13.29 & 11.13 & \underline{0.130} & 11.53 & 10.84 & \underline{0.040} \\
& CCA & 0.07 & 0.17 & 0.062 & 0.19 & 0.17 & 0.676 & 0.65 & \underline{0.22} & 0.060 & 0.62 & \underline{0.22} & 0.478 & & CCA & \underline{14.95} & \underline{2.99} & \underline{0.348} & \underline{17.54} & 3.90 & \underline{0.358} & 37.67 & 15.22 & 0.974 & 36.52 & 14.37 & 0.000 \\
& Naive & \textbf{0.00} & \textbf{0.04} & 0.152 & \textbf{0.05} & \textbf{0.04} & \textbf{1.000} & \textbf{0.05} & \textbf{0.05} & 0.122 & \textbf{0.00} & \textbf{0.06} & \textbf{0.986} & & Naive & 19.08 & 3.68 & 1.000 & 18.89 & \textbf{3.61} & \textbf{1.000} & \textbf{10.80} & \textbf{1.22} & 1.000 & \underline{10.70} & \textbf{1.20} & \textbf{0.994} \\
& PS-PPI (CCA) & 0.15 & \underline{0.16} & \underline{0.042} & 0.10 & \underline{0.16} & \underline{0.698} & \underline{0.12} & 0.23 & \textbf{0.048} & \underline{0.12} & 0.24 & \underline{0.598} & & PS-PPI (CCA) & \textbf{14.42} & \textbf{2.90} & \textbf{0.326} & \textbf{17.24} & \underline{3.81} & 0.330 & 37.90 & 15.41 & 0.968 & 36.69 & 14.58 & 0.002 \\
& PS-PPI & 0.12 & 0.52 & \textbf{0.056} & \underline{0.05} & 0.49 & 0.264 & 0.16 & 0.90 & \underline{0.056} & 0.38 & 0.87 & 0.222 & & PS-PPI & 70.50 & 59.66 & 0.806 & 72.75 & 63.44 & 0.096 & 12.20 & 11.15 & \textbf{0.126} & \textbf{10.63} & 11.17 & 0.034 \\
& MI & 12.75 & 1.80 & 0.714 & 11.57 & 1.51 & 0.124 & 19.65 & 4.02 & 0.952 & 18.34 & 3.54 & 0.138 & & MI & 38.01 & 15.41 & 0.860 & 40.05 & 17.12 & 0.090 & 20.66 & \underline{5.10} & 0.466 & 21.14 & 5.24 & 0.014 \\
& MI-RF & 3.37 & 0.67 & 0.370 & 2.67 & 0.61 & 0.076 & 7.67 & 1.57 & 0.456 & 8.93 & 1.79 & 0.094 & & MI-RF & 47.20 & 25.31 & 0.908 & 48.59 & 26.95 & 0.058 & \underline{11.59} & 5.54 & 0.276 & 11.55 & \underline{5.10} & 0.032 \\
\midrule
% ===================== Pair: MNAR1 (left)  ||  MNAR6 (right) =====================
\multirow{8}{*}{\shortstack{MNAR1 \\ $(Z_2, Y)$}} & Full & 0.06 & 0.02 & 0.042 & 0.02 & 0.02 & 1.000 & 0.07 & 0.03 & 0.038 & 0.06 & 0.03 & 1.000 & \multirow{8}{*}{\shortstack{MNAR6 \\ $(X_1, X_2, Z_2, Y, X_2Y)$}} & Full & 0.06 & 0.02 & 0.042 & 0.02 & 0.02 & 1.000 & 0.07 & 0.03 & 0.038 & 0.06 & 0.03 & 1.000 \\
& WCCA (est) & \textbf{0.06} & 0.80 & \textbf{0.054} & \textbf{0.03} & 0.81 & 0.200 & \underline{3.07} & 2.20 & \underline{0.080} & \textbf{2.73} & 2.18 & 0.076 & & WCCA (est) & 12.38 & 8.98 & \underline{0.170} & 11.38 & 9.33 & 0.018 & 74.11 & 61.09 & 0.864 & 76.03 & 64.34 & 0.118 \\
& CCA & 0.07 & \underline{0.35} & 0.040 & \underline{0.06} & \underline{0.35} & \underline{0.392} & 6.13 & 0.97 & 0.138 & 6.23 & 0.98 & \underline{0.372} & & CCA & 39.63 & 16.25 & 1.000 & 38.86 & 15.70 & 0.000 & 33.57 & 12.43 & 0.862 & 35.39 & 13.69 & 0.250 \\
& Naive & 0.06 & \textbf{0.04} & 0.144 & 0.09 & \textbf{0.04} & \textbf{0.996} & \textbf{2.73} & \textbf{0.13} & 0.352 & \underline{2.80} & \textbf{0.14} & \textbf{0.990} & & Naive & \textbf{7.71} & \textbf{0.64} & 0.990 & \textbf{7.65} & \textbf{0.63} & \textbf{0.996} & \textbf{20.47} & \textbf{4.25} & 1.000 & \textbf{20.29} & \textbf{4.17} & \textbf{0.994} \\
& PS-PPI (CCA) & 0.22 & 0.38 & 0.060 & 0.07 & 0.38 & 0.318 & 6.12 & \underline{0.92} & 0.126 & 6.22 & \underline{0.95} & 0.364 & & PS-PPI (CCA) & 39.96 & 16.55 & 1.000 & 38.97 & 15.78 & 0.000 & \underline{33.18} & \underline{12.18} & \underline{0.858} & \underline{35.04} & \underline{13.53} & \underline{0.256} \\
& PS-PPI & \underline{0.06} & 0.79 & \underline{0.044} & 0.15 & 0.81 & 0.222 & 3.25 & 1.91 & \textbf{0.052} & 2.91 & 1.95 & 0.090 & & PS-PPI & 12.73 & 8.63 & \textbf{0.168} & 11.67 & 8.64 & \underline{0.032} & 72.58 & 59.36 & \textbf{0.852} & 74.75 & 62.84 & 0.100 \\
& MI & 20.71 & 4.67 & 0.800 & 19.19 & 4.06 & 0.128 & 19.41 & 4.25 & 0.616 & 17.49 & 3.56 & 0.130 & & MI & 12.13 & \underline{2.34} & 0.276 & 11.13 & \underline{2.27} & 0.024 & 56.26 & 33.25 & 0.888 & 57.61 & 34.99 & 0.092 \\
& MI-RF & 4.84 & 1.16 & 0.314 & 5.37 & 1.19 & 0.168 & 3.56 & 2.57 & 0.314 & 5.58 & 2.81 & 0.112 & & MI-RF & \underline{8.36} & 3.80 & 0.264 & \underline{8.59} & 3.71 & 0.026 & 69.69 & 54.49 & 0.954 & 72.60 & 58.83 & 0.068 \\
\midrule
% ===================== Pair: MNAR2 (left)  ||  MNAR7 (right) =====================
\multirow{8}{*}{\shortstack{MNAR2 \\ $(X_1, Z_2, Y)$}} & Full & 0.06 & 0.02 & 0.042 & 0.02 & 0.02 & 1.000 & 0.07 & 0.03 & 0.038 & 0.06 & 0.03 & 1.000 & \multirow{8}{*}{\shortstack{MNAR7 \\ $(X_1, X_2, Z_2, Y, Z_2Y)$}} & Full & 0.06 & 0.02 & 0.042 & 0.02 & 0.02 & 1.000 & 0.07 & 0.03 & 0.038 & 0.06 & 0.03 & 1.000 \\
& WCCA (est) & 15.99 & 4.39 & 0.338 & 15.55 & 4.34 & 0.014 & \underline{2.18} & 2.67 & 0.092 & \textbf{1.49} & 2.68 & 0.096 & & WCCA (est) & 12.19 & 14.75 & 0.170 & 10.93 & 15.19 & \underline{0.048} & 22.64 & 22.10 & 0.166 & 20.49 & 21.07 & 0.036 \\
& CCA & 26.74 & 7.55 & 0.986 & 26.44 & 7.40 & 0.000 & 2.51 & 0.62 & 0.078 & 2.53 & 0.62 & 0.312 & & CCA & 38.14 & 15.35 & 0.992 & 37.59 & 15.03 & 0.002 & 11.23 & 2.96 & \underline{0.140} & 11.22 & 2.97 & 0.012 \\
& Naive & \underline{10.92} & \underline{1.23} & 1.000 & \underline{10.88} & \underline{1.22} & \textbf{1.000} & 2.40 & \textbf{0.12} & 0.302 & \underline{2.35} & \textbf{0.11} & \textbf{0.992} & & Naive & \textbf{6.87} & \textbf{0.51} & 0.976 & \textbf{6.79} & \textbf{0.51} & \textbf{0.996} & \textbf{10.96} & \textbf{1.26} & 1.000 & \textbf{10.85} & \textbf{1.24} & \textbf{0.992} \\
& PS-PPI (CCA) & 26.73 & 7.55 & 0.986 & 26.26 & 7.31 & 0.000 & 2.63 & \underline{0.55} & \underline{0.066} & 2.59 & \underline{0.58} & \underline{0.390} & & PS-PPI (CCA) & 37.73 & 15.08 & 0.982 & 37.09 & 14.61 & 0.000 & \underline{11.22} & \underline{2.69} & \textbf{0.130} & \underline{11.02} & \underline{2.76} & 0.010 \\
& PS-PPI & 16.21 & 4.18 & \underline{0.312} & 15.73 & 4.01 & 0.006 & 3.23 & 2.45 & \textbf{0.064} & 2.77 & 2.54 & 0.078 & & PS-PPI & \underline{10.48} & 13.43 & \underline{0.134} & \underline{9.03} & 13.15 & 0.040 & 21.09 & 20.95 & 0.158 & 18.24 & 19.09 & 0.018 \\
& MI & \textbf{4.71} & \textbf{0.78} & \textbf{0.102} & \textbf{4.50} & \textbf{0.81} & \underline{0.076} & 19.10 & 4.20 & 0.534 & 18.45 & 3.94 & 0.144 & & MI & 10.67 & \underline{2.20} & \textbf{0.132} & 10.37 & \underline{2.14} & 0.012 & 38.59 & 17.08 & 0.462 & 41.58 & 19.16 & \underline{0.128} \\
& MI-RF & 17.61 & 4.45 & 0.626 & 17.10 & 4.33 & 0.018 & \textbf{2.16} & 2.61 & 0.296 & 5.27 & 2.90 & 0.112 & & MI-RF & 18.41 & 7.46 & 0.178 & 20.19 & 7.96 & 0.022 & 25.27 & 19.69 & 0.344 & 24.18 & 18.37 & 0.028 \\
\bottomrule
\end{tabular}}
\end{table}

\begin{figure}[H]
    \centering
    \includegraphics[width=1.0\linewidth]{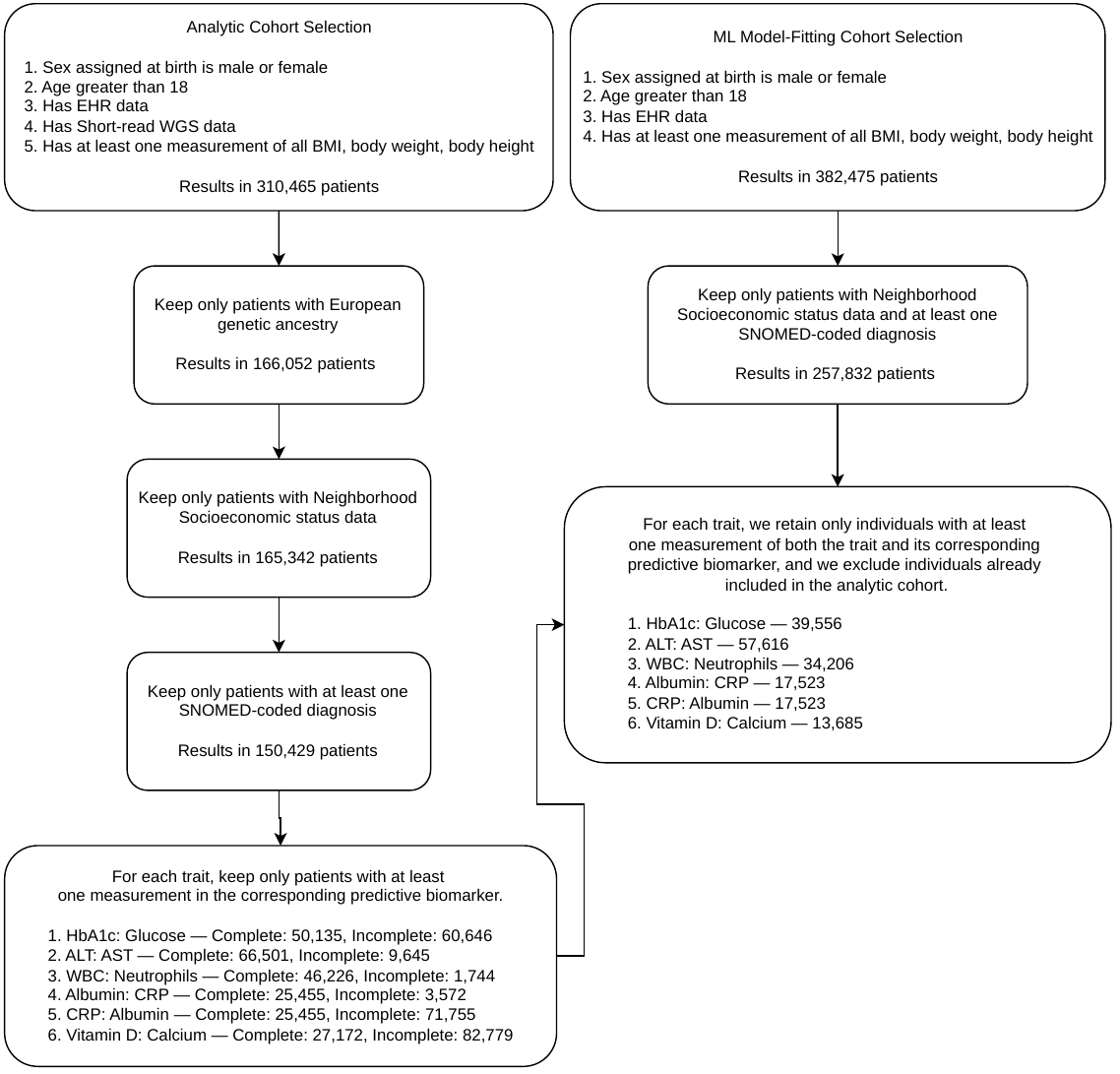}
    \caption{Flow diagram of the preprocessing of two datasets derived from the All of Us (AoU) database for lab biomarker genome-wide association studies (GWAS) in Section~\ref{sec:gwas}. The analytic cohort is used to fit the GWAS model; the machine learning (ML) model-fitting cohort is used to train ML models for imputing missing biomarker values in prediction-based inference approaches.}
    \label{fig:gwas-pipeline}
\end{figure}

\begin{table}[H]
\centering
\caption{International Classification of Diseases (ICD) code (both versions 9 and 10) prefixes of common phenotypes used in Table~\ref{tab:summary_statistics_traits}.}
\label{tab:phenotype_icd_prefix}
\begin{tabular}{lc}
\toprule
\textbf{Phenotype} & \textbf{ICD codes} \\
\midrule
Depression   & F32,\; F33,\; 311 \\
Hypertension & I10,\; 401 \\
Diabetes     & 250,\; E11 \\
Anxiety      & F41,\; 300 \\
Obesity      & E66,\; 278 \\
\bottomrule
\end{tabular}
\end{table}

\begin{table}[!tbp]
\centering
\caption{Units and All of Us (AoU) field names for the biomarkers used in this work. Abbreviations of the biomarkers are as follows: hemoglobin A1c (HbA1c), alanine aminotransferase (ALT), aspartate aminotransferase (AST), white blood cell count (WBC),  and C-reactive protein (CRP). }
\label{tab:aou_trait_units_fields}
\resizebox{\textwidth}{!}{
\begin{tabular}{lll}
\toprule
\textbf{Biomarkers} & \textbf{Unit} & \textbf{AoU field name} \\
\midrule
HbA1c        & percent                    & Hemoglobin A1c/Hemoglobin.total in Blood \\
ALT          & international unit per liter & Alanine aminotransferase [Enzymatic activity/volume] in Serum or Plasma \\
WBC          & thousand per microliter    & Leukocytes [\#/volume] in Blood by Automated count \\
Albumin      & gram per liter             & Albumin [Mass/volume] in Serum or Plasma \\
CRP          & milligram per deciliter    & C reactive protein [Mass/volume] in Serum or Plasma \\
Vitamin D    & nanogram per milliliter    & 25-Hydroxyvitamin D3+25-Hydroxyvitamin D2 [Mass/volume] in Serum or Plasma \\
Glucose      & milligram per deciliter    & Glucose [Mass/volume] in Serum or Plasma \\
AST          & international unit per liter & Aspartate aminotransferase [Enzymatic activity/volume] in Serum or Plasma \\
Neutrophils  & thousand per microliter    & Neutrophils [\#/volume] in Blood by Automated count \\
Calcium      & milligram per deciliter    & Calcium [Mass/volume] in Serum or Plasma \\
Body weight  & kilogram                   & Body weight \\
Body height  & centimeter                 & Body height \\
BMI          & kilogram per square meter  & Body mass index (BMI) [Ratio] \\
\bottomrule
\end{tabular}}
\end{table}

% \begin{table}[H]
% \centering
% \caption{Biomarker measurement patterns among 150,429 patients, of whom 150,055 had at least one outpatient visit. \textbf{Left}: $N$: the number of patients with at least one biomarker measurements, as a proportion (Prop.) among those with at least one outpatient visit. \textbf{Right}: \#\ Visits: the median [IQR] number of outpatient visits among those with at least one biomarker measurement; \#\ Meas.: the median [IQR] number of biomarker measurements per patient; Prop.: the median proportion of a patient’s outpatient visits at which the biomarker was measured.}
% \label{tab:biomarkers}
% %\resizebox{\textwidth}{!}{%
% \begin{tabular}{lrrrrr}
% \toprule
%  & \multicolumn{2}{c}{\shortstack[c]{\textbf{Number of Patients with}\\\textbf{at least one measurement}}} & \multicolumn{3}{c}{\shortstack[c]{\textbf{Summary of repeated}\\\textbf{measurements if at least one measure}}} \\
% \cmidrule(lr){2-3} \cmidrule(lr){4-6}
% \textbf{Biomarker}  & \textbf{$N$}  & \textbf{Prop.\ (\%)}  & \textbf{\#\ Visits}  & \textbf{\#\ Meas.}  & \textbf{Prop.\ (\%)} \\
% \midrule
% WBC & 111,850 & 74.5 & 33 [6, 100] & 7 [3, 18] & 4.4 \\
% Albumin & 104,179 & 69.4 & 44 [8, 133] & 6 [2, 15] & 3.3 \\
% ALT & 104,392 & 69.6 & 39 [7, 123] & 7 [3, 16] & 3.7 \\
% HbA1c & 65,770 & 43.8 & 62 [15, 172] & 3 [1, 7] & 0.8  \\
% CRP & 32,178 & 21.4 & 49 [9, 160] & 2 [1, 3] & 0.0 \\
% Vitamin D & 29,269 & 19.5 & 106 [42, 245] & 2 [1, 5] & 0.8 \\
% \bottomrule
% \end{tabular}%
% %}
% \end{table}

\begin{figure}[H]
    \centering
    \includegraphics[width=0.85\linewidth]{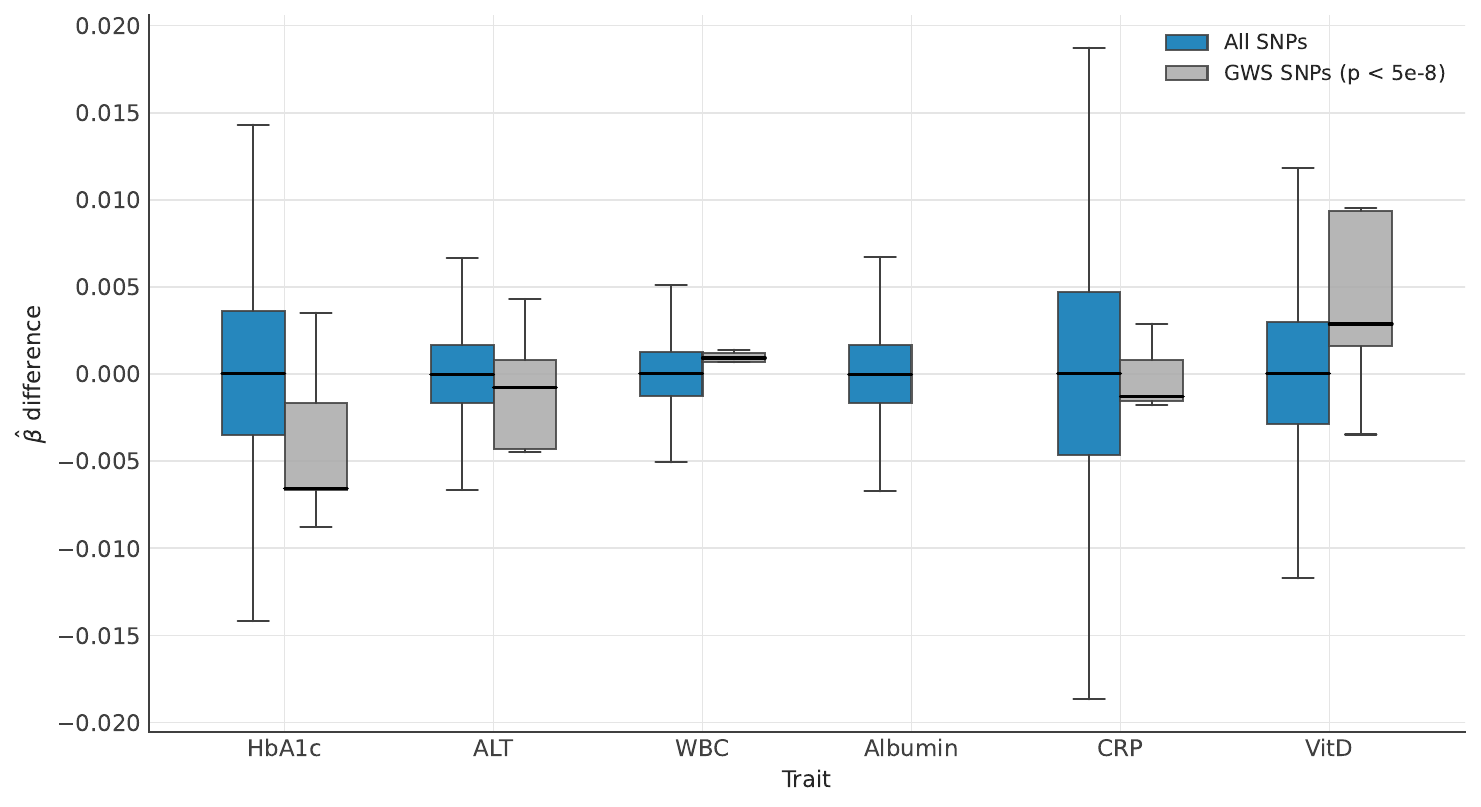}
    \caption{Difference of $\hat{\beta}$ estimates (SynSurr - CCA) across traits for all SNPs. The dashed line indicates 0.}
    \label{fig:beta-diff-synsurr-cca}
\end{figure}

\begin{figure}[H]
    \centering
    \includegraphics[width=0.85\linewidth]{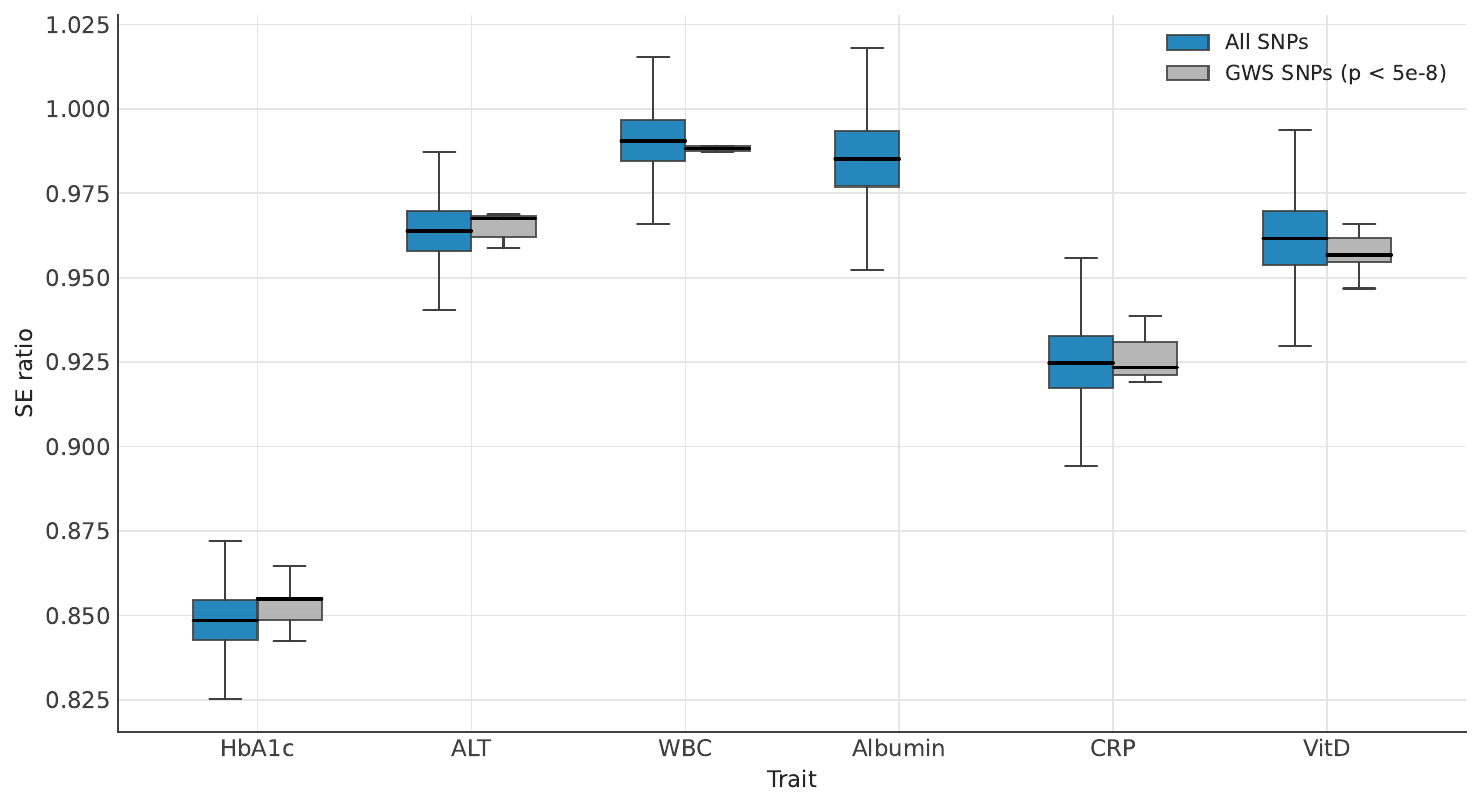}
    \caption{Ratio of standard errors (SynSurr / CCA) across traits for all SNPs. The dashed line indicates 1.}
    \label{fig:se-ratio-synsurr-cca}
\end{figure}

\begin{figure}[H]
    \centering
    \includegraphics[width=0.85\linewidth]{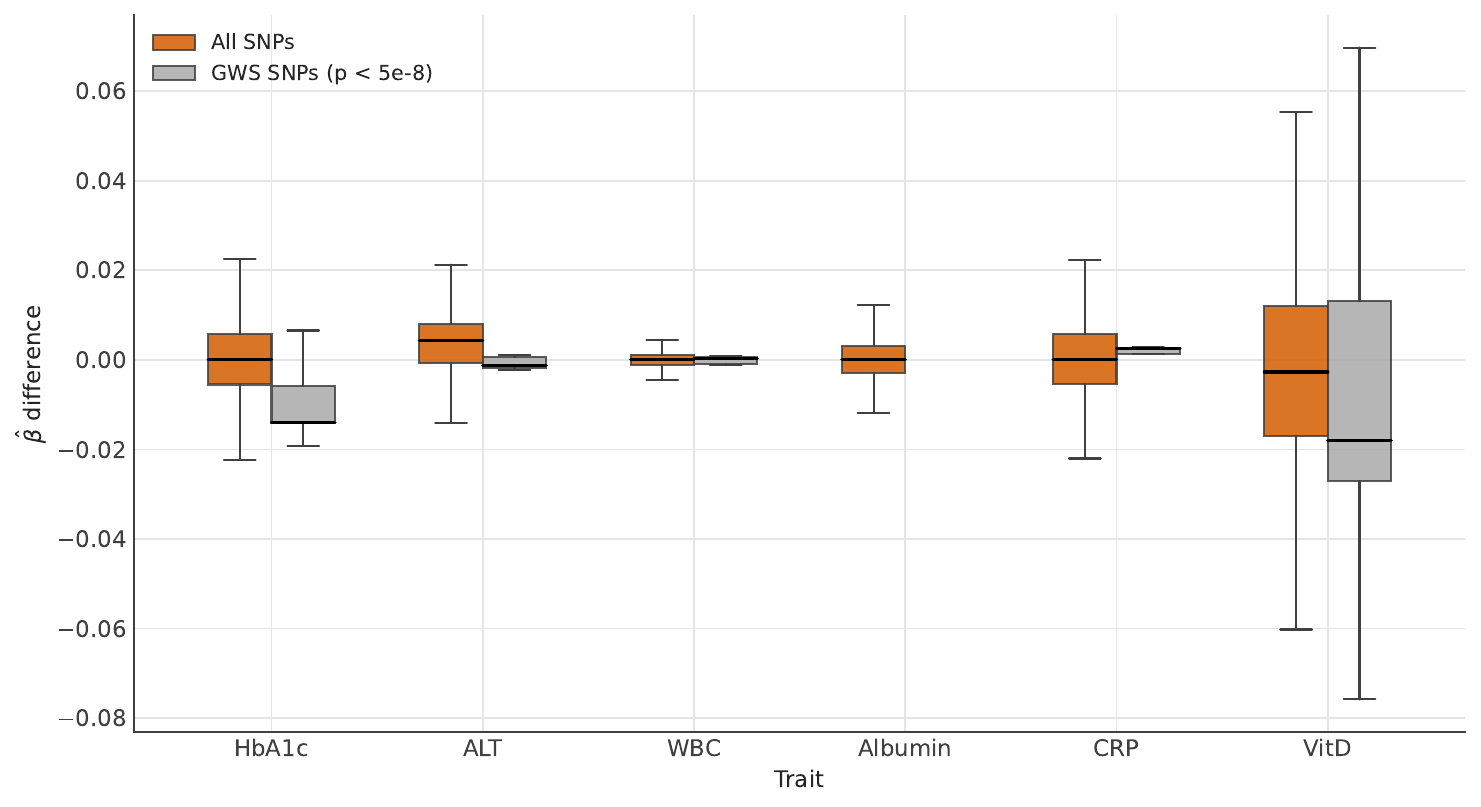}
    \caption{Difference of $\hat{\beta}$ estimates (PS-PPI - WCCA) across traits for all SNPs. The dashed line indicates 0.}
    \label{fig:beta-diff-psppi-wcca}
\end{figure}

\begin{figure}[H]
    \centering
    \includegraphics[width=0.85\linewidth]{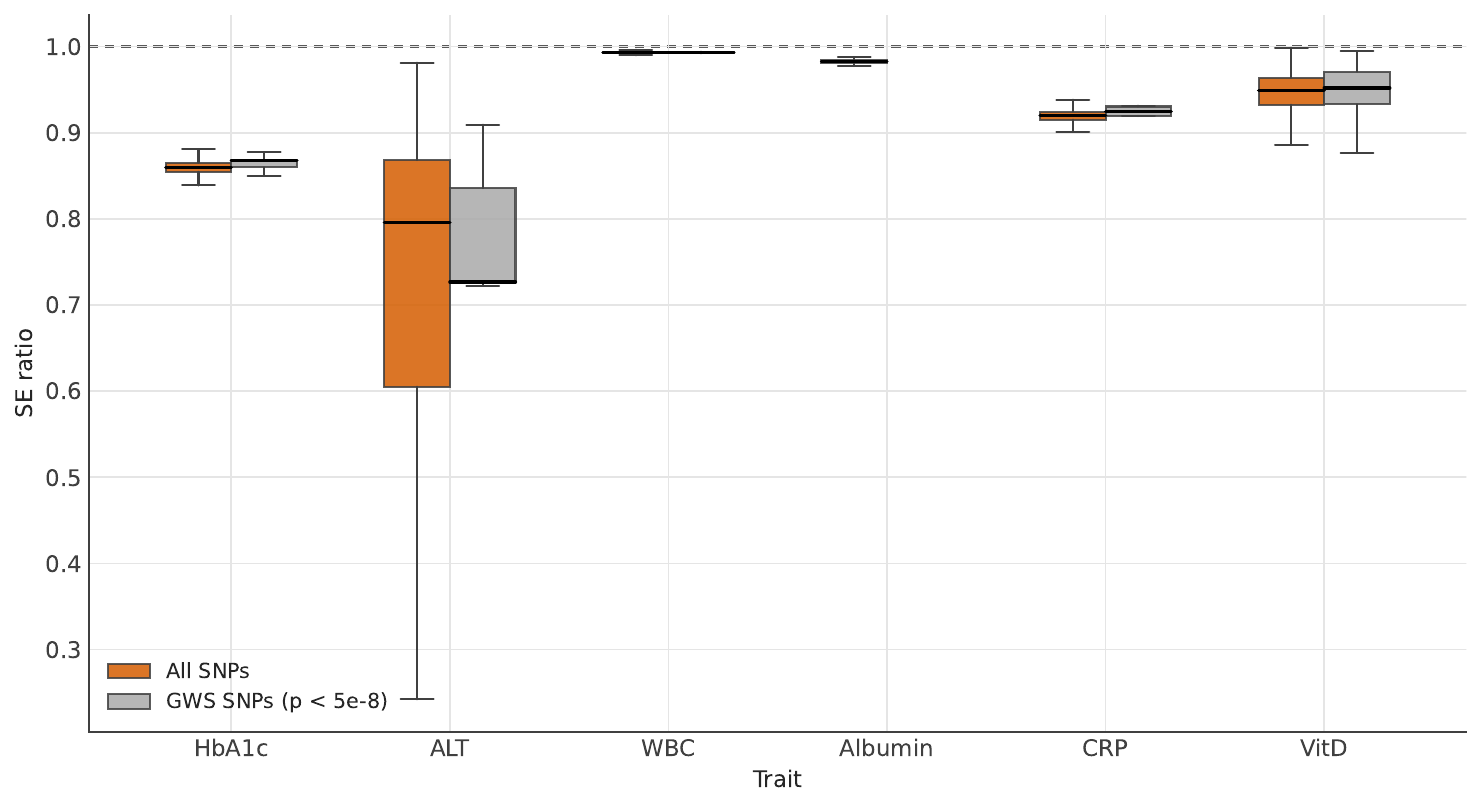}
    \caption{Ratio of standard errors (PS-PPI / WCCA) across traits for all SNPs. The dashed line indicates 1.}
    \label{fig:se-ratio-psppi-wcca}
\end{figure}

% =========== significant SNPs ============
% \begin{figure}[H]
%     \centering
%     \includegraphics[width=0.85\linewidth]{figs/sig_beta_diff_synsurr_cca.pdf}
%     \caption{Difference of $\hat{\beta}$ estimates (SynSurr - CCA) across traits for significant SNPs ($p<5\times10^{-8}$). The dashed line indicates 0.}
%     \label{fig:sig-beta-diff-synsurr-cca}
% \end{figure}

% \begin{figure}[H]
%     \centering
%     \includegraphics[width=0.85\linewidth]{figs/sig_se_ratio_synsurr_cca.pdf}
%     \caption{Ratio of standard errors (SynSurr / CCA) across traits for significant SNPs ($p<5\times10^{-8}$). The dashed line indicates 1.}
%     \label{fig:sig-se-ratio-synsurr-cca}
% \end{figure}

% \begin{figure}[H]
%     \centering
%     \includegraphics[width=0.85\linewidth]{figs/sig_beta_diff_psppi_wcca.pdf}
%     \caption{Difference of $\hat{\beta}$ estimates (PS-PPI - WCCA) across traits for significant SNPs ($p<5\times10^{-8}$). The dashed line indicates 0.}
%     \label{fig:sig-beta-diff-psppi-wcca}
% \end{figure}

% \begin{figure}[H]
%     \centering
%     \includegraphics[width=0.85\linewidth]{figs/sig_se_ratio_psppi_wcca.pdf}
%     \caption{Ratio of standard errors (PS-PPI / WCCA) across traits for significant SNPs ($p<5\times10^{-8}$). The dashed line indicates 1.}
%     \label{fig:sig-se-ratio-psppi-wcca}
% \end{figure}

\newpage

\section{Proofs}
\label{sec:proofs}

\subsection{Proof of Theorem~\ref{thrm:linear_regression_cca}}
\label{sec:proof_theorem_linear_regression_cca}
\begin{proof}
Notice that $Y_i = \beta_0 + \vX_i\vbeta^\ast_{\vX} + \vZ_i\vbeta_{\vZ}^\ast + \epsilon_i$, where $\vX_i$ and $\vZ_i$ are row vectors, and $\epsilon_i \sim \cN(0, \sigma^2)$. Denote $\vG_i = [1, \vX_i, \vZ_i]$, $\vbeta = [\beta_0, \vbeta^\top_{\vX}, \vbeta^\top_{\vZ}]^\top$, fitting a linear regression on $Y_i$ to $\vX_i$ and $\vZ_i$ using only samples with $R_i = 1$ is equivalent to solve $\widehat{\vbeta}_{\text{CC}}$ from the following equation:
$$ \frac{1}{n}\sum^n_{i=1} I(R_i = 1) \vG^\top_i (\vG_i \vbeta - Y_i) = \vzero.$$
Denote $\vbeta^\triangle$ as the population parameter of the model above, i.e., $\vbeta^\triangle$ satisfies
$$\E \left [\vG^\top_1 (\vG_1 \vbeta^\triangle - Y_1) \mid R_1 = 1 \right ] = \vzero.$$
By rewriting the above equation we have:
\begin{align*}
\vzero &= \E \left [\vG^\top_1 \vG_1 \vbeta^\triangle - \vG^\top_1 Y_1 \mid R_1 = 1 \right ] \\
&= \E_{(\vX, \vZ) \mid R = 1} \left \{ \vG^\top_1 \vG_1 \vbeta^\triangle - \vG^\top_1 \E \left [ Y_1 \mid \vX_1, \vZ_1, R_1 = 1 \right ] \right \} \\
&= \E_{(\vX, \vZ) \mid R = 1} \left \{ \vG^\top_1 \vG_1 (\vbeta^\triangle - \vbeta^\ast ) - \vG^\top_1 (\E \left [ Y_1 \mid \vX_1, \vZ_1, R_1 = 1 \right ] - \vG_1\vbeta^\ast) \right \} \\
&= \E \left [ \vG^\top_1 \vG_1 (\vbeta^\triangle - \vbeta^\ast ) - \vG^\top_1 r(\vX_1, \vZ_1) \mid R_1 = 1 \right ],
\end{align*}
where we denote $r(\vX_1, \vZ_1) = \E \left [ Y_1 \mid \vX_1, \vZ_1, R_1 = 1 \right ] - \vG_1\vbeta^\ast$. Hence,
$$ \vbeta^\triangle = \vbeta^\ast + \underbrace{ \E \left (\vG^\top_1 \vG_1 \mid R_1 = 1 \right)^{-1} \E \left ( \vG_1  r(\vX_1, \vZ_1) \mid R_1 = 1  \right ) }_{\vzeta^\ast},$$
where $\vzeta^\ast$ is the OLS estimate of linear regression between $r(\vX_i, \vZ_i)$ and $\vG$ using the fully observed cases.

When the missing data mechanism is at least MAR, \textcolor{black}{we know that $Y$ and $R$ are independent given $\vX$ and $\vZ$. Hence,} we have $r(\vX_1, \vZ_1) = \vG_1\vbeta^\ast - \vG_1\vbeta^\ast = 0$. As a result, $\vzeta^\ast = 0$ and $\vbeta^\triangle = \vbeta^\ast$.

When the missing data mechanism is MNAR, we know that when $r(\vX, \vZ)$ is a linear function that only depends on $\vZ$, i.e., $r(\vX, \vZ) = h(\vZ) = \vZ\tilde{\vbeta}^\ast$, we have $\vzeta^\ast_{\vX} = \vzero$, namely, $\vbeta_{\vX}^\triangle = \vbeta_{\vX}^\ast$.

Now we show the following conditions suffice for $r(\vX, \vZ)$ to be a linear function of $\vZ$ only: (a) $\beta^\ast = 0$; (b) $R_i$ does not depend on $\vX_i$ given $Y_i$ and $\vZ_i$; (c) \textcolor{black}{$\vZ_i$ is a multi-dimensional one-hot encoded vector}. First, \textcolor{black}{according to condition (a) and (b),} the conditional distribution
\begin{align*}
f(Y_i | \vX_i, \vZ_i, R_i = 1) &= \frac{f(R_i = 1 | \vX_i, \vZ_i, Y_i)f(Y_i | \vX_i, \vZ_i)}{\int f(R_i = 1 | \vX_i, \vZ_i, Y_i) f(Y_i | \vX_i, \vZ_i) d\epsilon_i} \\
&= \frac{f(R_i = 1 | \vZ_i, Y_i)f(Y_i | \vZ_i)}{\int f(R_i = 1 | \vZ_i, Y_i) f(Y_i | \vZ_i) d\epsilon_i}.
\end{align*}
This is because first, \textcolor{black}{as $R_i$ does not depend on $\vX_i$ given $Y_i$ and $\vZ_i$}, we have $f(R_i = 1 | \vX_i, Y_i,  \vZ_i) = f(R_i = 1 | Y_i,  \vZ_i)$. Second, \textcolor{black}{as} $\beta^\ast = 0$, \textcolor{black}{we have} $f(Y_i | \vX_i, \vZ_i) = f(Y_i | \vZ_i)$. \textcolor{black}{The above equation indicates that $Y_i$ and $\vX_i$ are independent given conditions (a) and (b).} Accordingly, we have
$$\E \left [ Y_i | \vX_i, \vZ, R_i = 1 \right ] = \E \left [ Y_i | \vZ_i, R_i = 1 \right ] = h(\vZ_i),$$
which is a function of only $\vZ_i$. Finally, if $\vZ_i$ is a multi-dimensional one-hot encoded random vector, then such a conditional expectation can be written as a linear function $h(\vZ_i) = \vZ_i\tilde{\vbeta}^\ast$, \textcolor{black}{where the $k$-th element of $\tilde{\vbeta}^\ast$ is the conditional expectation of $Y_i$ given that the $k$th element of $\vZ_i$ equals 1 and $R_i = 1$, that is,} $\tilde{\vbeta}^\ast_k = \E \left [ Y_i | \vZ_{i, k} = 1, R_i = 1 \right ]$.
\end{proof}

\subsection{Proof of Theorem~\ref{thrm:logistic_regression_cca}}
\label{sec:proof_theorem_logistic_regression}

\textcolor{black}{Suppose the true model is $\text{logit}(\E(Y_i \mid \vX, \vZ)) = \beta^\ast_0 + \vX_i\vbeta_{\vX}^\ast + \vZ_i\vbeta_{\vZ}^\ast$. When the outcome $Y_i$ is partially observed, the complete-case analysis targets the following model induced by conditioning on $R_i = 1$~\citep{kundu2024framework}:}
$$\text{logit}(\E(Y_i \mid \vX, \vZ, R = 1)) = \beta^\ast_0 + \vX_i\vbeta_{\vX}^\ast + \vZ_i\vbeta_{\vZ}^\ast + \log r(\vX_i, \vZ_i),$$
where $r(\vX_i, \vZ_i) = \frac{P(R_i = 1 \mid Y_i = 1, \vX_i, \vZ_i)}{P(R_i = 1 \mid Y_i = 0, \vX_i, \vZ_i)}$. When the missing data mechanism is at least MAR, we have $r(\vX_i, \vZ_i) = 1$, and therefore
$$\text{logit}(\E(Y_i \mid \vX, \vZ, R = 1)) = \beta^\ast_0 + \vX_i\vbeta_{\vX}^\ast + \vZ_i\vbeta_{\vZ}^\ast.$$

When the missingness data mechanism is MNAR, with $\vX_i \perp R_i \mid (Y_i, \vZ_i)$ and $\vZ_i \perp R_i \mid Y_i$, we have $r(\vX_i, \vZ_i) = \frac{P(R_i = 1 \mid Y_i = 1, \vX_i, \vZ_i)}{P(R_i = 1 \mid Y_i = 0, \vX_i, \vZ_i)} = \frac{P(R_i = 1 \mid Y_i = 1)}{P(R_i = 1 \mid Y_i = 0)}$, which is a constant function of $\vX_i$ and $\vZ_i$. As such, \textcolor{black}{consistent estimator} of $\vbeta^\ast_{\vX}$ \textcolor{black}{can be obtained since this term is fully absorbed into the intercept term}. \textcolor{black}{Under standard regularity conditions, the resulting estimator of $\vbeta^\ast_{\vX}$ is also asymptotically normal.}

\subsection{Proof of Theorem~\ref{theorem:linear_regression_pb}}
\label{sec:proof_theorem_linear_regression_pb}

Estimator-level PB inference estimators are in the form of Equation~(\ref{eq:estimator_level}). From Theorem~\ref{thrm:linear_regression_cca}, with the set of sufficient conditions given, the base estimator $\widehat{\vbeta}_{\text{CC}}$ is a consistent estimation of $\vtheta^\ast$. As such, we want to show that $\widehat{\vgamma}_1 - \widehat{\vgamma}_2$ is a consistent estimate of zero.

We know that $\widehat{\vgamma}_1$ and $\widehat{\vgamma}_2$ are solutions of the following two estimating equations:
$$ \frac{1}{n}\sum^n_{i=1} I(R_i = 1) \vG^\top_i (\vG_i \vgamma_1 -
\widehat{Y}_i) = \vzero,$$
$$ \frac{1}{n}\sum^n_{i=1} I(R_i = 0) \vG^\top_i (\vG_i \vgamma_2 -
\widehat{Y}_i) = \vzero.$$
where $\widehat{Y}_i = \vG_i\tilde{\vbeta} + \tau_{i,\text{ML}}$. We want to show that there exists a $\vgamma^\triangle$ such that it is the solution for both
$$\E \left [\vG^\top_1 (\vG_1 \vgamma^\triangle - \widehat{Y}_1) \mid R_1 = 1 \right ] = \vzero.$$
and
$$\E \left [\vG^\top_1 (\vG_1 \vgamma^\triangle - \widehat{Y}_1) \mid R_1 = 0 \right ] = \vzero,$$
Specifically, we know
\begin{align*}
\E \left [\vG^\top_1 (\vG_1 \vgamma^\triangle - \widehat{Y}_1) \mid R_1 = 1 \right ] &= \E \left [\vG^\top_1 (\vG_1 \vgamma^\triangle - \vG_1\tilde{\vbeta} - \tau_{i,\text{ML}}) \mid R_1 = 1 \right ] \\
&= \E \left [\vG^\top_1\vG_1 (\vgamma^\triangle - \tilde{\vbeta}) \mid R_1 = 1 \right ] - \E \left [ \tau_{i,\text{ML}} \right ] \\
&= \vzero.
\end{align*}
As such, $\vgamma^\triangle = \tilde{\vbeta} - (\E \left [ \tau_{i,\text{ML}} \right ], \vzero)^\top$. Similarly, we can obtain that $\vgamma^\triangle$ is also the solution of $E \left [\vG^\top_1 (\vG_1 \vgamma^\triangle - \widehat{Y}_1) \mid R_1 = 0 \right ] = \vzero$. Hence, $\widehat{\vgamma}_1 - \widehat{\vgamma}_2$ is a consistent estimate of zero.

\section{Supplementary Details of Simulation Settings}
\label{sec:supplementary-setting-details}

\subsection{Settings of outcome observation models}

In this subsection, we present the detailed configuration of the outcome observation models in each simulation settings.

\paragraph{Linear regression with continuous confounders.} Under the MCAR setting, we set $\Pr(R_i = 1) = 0.2$ for all $i$. When the missing data mechanisms are MAR, we adopt the following models
\begin{align*}
&\text{logit}\{\Pr(R_i = 1)\} = -1.7 + 0.5 Z_{i2}, \tag{MAR1} \\
&\text{logit}\{\Pr(R_i = 1)\} = -2 + X_{i1} + 0.5 Z_{i2}. \tag{MAR2}
\end{align*}
When the missing data mechanisms are MNAR, we adopt the following models
\begin{align*}
&\text{logit}\{\Pr(R_i = 1)\}
= -2.8 + 0.5\, Z_{i2} + 2\, I(Y_i < 1), &&\text{(MNAR1)} \\
&\text{logit}\{\Pr(R_i = 1)\}
= -3 + X_{i1} + 0.5\, Z_{i2} + 2\, I(Y_i < 1), &&\text{(MNAR2)} \\
&\text{logit}\{\Pr(R_i = 1)\}
= -3 + X_{i2} + 0.5\, Z_{i2} + 2\, I(Y_i < 1), &&\text{(MNAR3)} \\
&\text{logit}\{\Pr(R_i = 1)\}
= -3.2 + X_{i1} + X_{i2} + 0.5\, Z_{i2} + 2\, I(Y_i < 1), &&\text{(MNAR4)} \\
&\text{logit}\{\Pr(R_i = 1)\}
= -3.5 + X_{i1} + X_{i2} + 0.5\, Z_{i2} + 2\, I(Y_i < 1)
+ X_{i1} I(Y_i < 1), &&\text{(MNAR5)} \\
&\text{logit}\{\Pr(R_i = 1)\}
= -3.2 + X_{i1} + X_{i2} + 0.5\, Z_{i2} + 2\, I(Y_i < 1)
+ X_{i2} I(Y_i < 1), &&\text{(MNAR6)} \\
&\text{logit}\{\Pr(R_i = 1)\}
= -3.2 + X_{i1} + X_{i2} + 0.5\, Z_{i2} + 2\, I(Y_i < 1)
+ 0.5\, Z_{i2} I(Y_i < 1). &&\text{(MNAR7)}
\end{align*}

\paragraph{Linear regression with a categorical confounder.} Under the MCAR setting, we set $\Pr(R_i = 1) = 0.2$ for all $i$. When the missing data mechanisms are MAR, we adopt the following models
\begin{align*}
&\text{logit}\{\Pr(R_i = 1)\} = -1.7 + Z_{i4}, \tag{MAR1} \\
&\text{logit}\{\Pr(R_i = 1)\} = -2 + X_{i1} + Z_{i4}. \tag{MAR2}
\end{align*}
When the missing data mechanisms are MNAR, we adopt the following models
\begin{align*}
&\text{logit}\{\Pr(R_i = 1)\}
= -2 + Z_{i4} + I(Y_i < 0), &&\text{(MNAR1)} \\
&\text{logit}\{\Pr(R_i = 1)\}
= -2.4 + X_{i1} + Z_{i4} + I(Y_i < 0), &&\text{(MNAR2)} \\
&\text{logit}\{\Pr(R_i = 1)\}
= -2.2 + X_{i2} + Z_{i4} + I(Y_i < 0), &&\text{(MNAR3)} \\
&\text{logit}\{\Pr(R_i = 1)\}
= -2.5 + X_{i1} + X_{i2} + Z_{i4} + I(Y_i < 0), &&\text{(MNAR4)} \\
&\text{logit}\{\Pr(R_i = 1)\}
= -2.5 + X_{i1} + X_{i2} + Z_{i4} + I(Y_i < 0)
- X_{i1} I(Y_i < 0), &&\text{(MNAR5)} \\
&\text{logit}\{\Pr(R_i = 1)\}
= -2.5 + X_{i1} + X_{i2} + Z_{i4} + I(Y_i < 0)
- X_{i2} I(Y_i < 0), &&\text{(MNAR6)} \\
&\text{logit}\{\Pr(R_i = 1)\}
= -2.5 + X_{i1} + X_{i2} + Z_{i4} + I(Y_i < 0)
+ Z_{i4} I(Y_i < 0). &&\text{(MNAR7)}
\end{align*}

\paragraph{Logistic Regression.} Under the MCAR setting, we set $\Pr(R_i = 1) = 0.2$ for all $i$.
When the missing data mechanisms are MAR, we adopt the following models
\begin{align*}
&\text{logit}\{\Pr(R_i = 1)\} = -1.7 + 0.5 Z_{i2}, \tag{MAR1} \\
&\text{logit}\{\Pr(R_i = 1)\} = -2 + X_{i1} + 0.5 Z_{i2}. \tag{MAR2}
\end{align*}
When the missing data mechanisms are MNAR, we adopt the following models
\begin{align*}
&\text{logit}\{\Pr(R_i = 1)\}
= -3.3 + 0.5 Z_{i2} + 2 Y_i, &&\text{(MNAR1)} \\
&\text{logit}\{\Pr(R_i = 1)\}
= -3.5 + X_{i1} + 0.5 Z_{i2} + 2 Y_i, &&\text{(MNAR2)} \\
&\text{logit}\{\Pr(R_i = 1)\}
= -3.5 + X_{i2} + 0.5 Z_{i2} + 2 Y_i, &&\text{(MNAR3)} \\
&\text{logit}\{\Pr(R_i = 1)\}
= -3.9 + X_{i1} + X_{i2} + 0.5 Z_{i2} + 2 Y_i, &&\text{(MNAR4)} \\
&\text{logit}\{\Pr(R_i = 1)\}
= -4.2 + X_{i1} + X_{i2} + 0.5 Z_{i2} + 2 Y_i
+ X_{i1} Y_i, &&\text{(MNAR5)} \\
&\text{logit}\{\Pr(R_i = 1)\}
= -4.2 + X_{i1} + X_{i2} + 0.5 Z_{i2} + 2 Y_i
+ X_{i2} Y_i, &&\text{(MNAR6)} \\
&\text{logit}\{\Pr(R_i = 1)\}
= -4.5 + X_{i1} + X_{i2} + 0.5 Z_{i2} + 2 Y_i
+ 0.5 Z_{i2} Y_i. &&\text{(MNAR7)}
\end{align*}

\subsection{Settings of synthetic surrogates in logistic regression}
\label{sec:synthetic_surrogates_logistic_regression}

For the logistic regression settings, another set of $(Y_i, X_{i1}, X_{i2}, Z_{i1}, Z_{i2})$ is sampled and fit a correctly specified logistic regression model to obtain $\beta_{\text{train}}$. The fitted model is then used to make predictions. Specifically, $\widehat{Y}_i$ is sampled from Bernoulli($\widehat{p}_i$), where $\text{logit}(\widehat{p}_i) = \vG_i\beta_{\text{train}}$.

\section{Supplementary Simulations}
\label{sec:supplementary-results}

\subsection{Simulation on the binary outcome setting}
\label{sec:binary-outcome-results}

\subsubsection{Missing completely at random}

\paragraph{Hypothesis testing} Under the MCAR setting, CCA, WCCA (est), and the two PS-PPI variants show valid type I error close to 0.05 for both $\beta_1$ and $\beta_2$. Specifically, \textcolor{black}{for $\beta_1$,} CCA and WCCA (est) have type I errors around 0.06, while PS-PPI (CCA) and PS-PPI are 0.05. For $\beta_2$, CCA and WCCA (est) show type I errors close to the 0.05 nominal level, while PS-PPI (CCA) and PS-PPI are slightly conservative with type I errors around 0.03. The two MI methods and the Naive method show inflated type I error controls. In terms of power, the two PS-PPI variants exhibit strongest statistical power of around 0.90 and 0.73 for $\beta_1$ and $\beta_2$, higher than those from their base estimators CCA and WCCA (est).

\paragraph{Point estimation} For point estimation under MCAR, the Naive method achieves the lowest MSE ranging from 0.03 to 0.05 for both coefficients. In contrast, CCA, WCCA (est) and the two PS-PPI variants show MSE of 0.1-0.2 for both coefficients.

\subsubsection{Missing at random}

\paragraph{Hypothesis testing} Under the two MAR settings, we observe similar type I error patterns. Specifically, CCA, WCCA (est) and PS-PPI with estimated propensity score show valid type I error controls for both coefficients under MAR1 and MAR2, while the two MI methods and Naive show inflated type I errors. Note that although PS-PPI (CCA) shows valid type I error controls on Table~\ref{tab:logistic_regression_summary}, it exhibits inflated type I errors in Figure~\ref{fig:logistic_regression_varying_quality_type_i_error}. In terms of power, among methods with valid type I error controls, CCA achieves strongest power performance of 0.75 and 0.58 for $\beta_1$ and $\beta_2$ under MAR1, and 0.68 and 0.48 for the two coefficients under MAR2. PS-PPI exhibits weaker power than CCA but consistently outperforms WCCA (est). Under MAR1, PS-PPI attains power of 0.56 for $\beta_1$ and 0.26 for $\beta_2$, compared to 0.52 and 0.22 for WCCA (est), respectively; under MAR2, the corresponding power values are 0.26 and 0.22 for PS-PPI versus 0.24 and 0.13 for WCCA (est).

\paragraph{Point estimation} For point estimation, consistent with the MCAR setting, the Naive method achieves the smallest MSE across both coefficients and both missing data mechanisms settings, with MSEs of 0.04 for $\beta_1$ and 0.05--0.06 for $\beta_2$. Furthermore, the two PS-PPI variants show MSE that are overall better than CCA and WCCA (est) and the two MI methods.

\subsubsection{Missing not at random}

\paragraph{Hypothesis testing} Under the MNAR settings, the type I error patterns are different under different missing data mechanism settings. Specifically, according to Table~\ref{tab:logistic_regression_summary} and Figure~\ref{fig:logistic_regression_varying_quality_type_i_error}, when $X_1$ is not included in the observation model, CCA, WCCA (est) exhibit valid type I error control of around 0.05 on $\beta_1$. However, when $X_1$ is included in the observation model, all the selected method show inflated type I error controls. We provide theoretical explanation of this phenomenon in Theorem~\ref{thrm:logistic_regression_cca}. On the other hand, for $\beta_2$, all the selected methods show inflated type I error controls in all the missing data mechanism settings.

\paragraph{Point estimation} Similar to what have been observed from the MCAR and MAR settings, under the MNAR setting, the Naive method show the overall best performance in terms of achieving smallest MSE. However, among the performance of the remaining methods, there is no method that uniformly dominates the other methods across both coefficients and all MNAR settings.

\subsection{Simulation on the quality of imputation}
\label{sec:simulation_quality_of_imputation}

In this subsection we conduct a supplemental simulation to investigate the impact of the quality of machine learning imputation on the performance of the included methods. To allow controls over the bias and variance of the imputation, we introduce artificially constructed imputations. Specifically, instead of fitting a model with additional data, we generate predictions by perturbing the underlying true values with pre-specified bias and noise levels. Under the linear regression setting, we set
$$\widehat{Y}_i = Y_i + b_{i, \text{pred}} + \epsilon_{i, \text{pred}}.$$
Here, we define $b_{i,\text{pred}} \sim \textsf{Exponential}(\lambda_{\text{pred}})$, and $\epsilon_{i, \text{pred}} \sim \cN(0, \sigma^2_{\text{pred}})$. For the logistic regression setting, we generate predictions by flipping the true outcomes $Y_i$ with a predefined probability $p$, thereby introducing the controllable prediction error. We iterate $\lambda_{\text{pred}}$ and $\sigma_{\text{pred}}$ from $\{0.2, 2.0\}$ to investigate the performance of naive imputation and PB inference methods under different quality of imputation. Similarly, we iterate $p$ from $\{0.05, 0.1, 0.2, 0.5\}$. We replace the imputations introduced in the main simulation with these different quality of imputations, and redo the simulation. Note that we only include PB inference methods in this simulation as they are the only methods that are impacted by the quality of the imputations in addition to Naive. We exclude Naive from this supplemental simulation as it has been shown to fail to achieve valid type I error in all the settings in the main simulation. The results for the linear regression setting with continuous confounders are shown in Figure~\ref{fig:linear_regression_simulation_varying_quality_power} and~\ref{fig:linear_regression_simulation_varying_quality_type_i_error}; the results for the linear regression setting with a categorical confounder are shown in Figure~\ref{fig:linear_regression_simulation_binary_z_varying_quality_power} and~\ref{fig:linear_regression_simulation_binary_z_varying_quality_type_i_error}; the results for the logistic regression setting are shown in Figure~\ref{fig:logistic_regression_varying_quality_power} and~\ref{fig:logistic_regression_varying_quality_type_i_error}.

\end{document}

%% file: math_command.tex
\usepackage{amsmath,amsfonts,bm}

\def\eqref#1{equation~\ref{#1}}
\def\1{\bm{1}}

\def\vzero{{\bm{0}}}

\def\vtheta{{\bm{\theta}}}
\def\vomega{{\bm{\omega}}}
\def\vbeta{{\bm{\beta}}}
\def\valpha{{\bm{\alpha}}}
\def\vpsi{{\bm{\psi}}}
\def\vzeta{{\bm{\zeta}}}
\def\vomega{{\bm{\omega}}}
\def\vgamma{{\bm{\gamma}}}

\DeclareMathAlphabet{\mathsfit}{\encodingdefault}{\sfdefault}{m}{sl}
\SetMathAlphabet{\mathsfit}{bold}{\encodingdefault}{\sfdefault}{bx}{n}

\newcommand{\E}{\mathbb{E}}

\newcommand{\R}{\mathbb{R}}

\newcommand{\Var}{\mathrm{Var}}

\renewcommand{\cal}{\mathcal}

\newcommand{\cD}{{\cal D}}

\newcommand{\cN}{{\cal N}}

\renewcommand{\geq}{\geqslant}
\newtheorem{theorem}{Theorem}

\let\Pr\relax
\DeclareMathOperator*{\Pr}{\mathbb{P}}

\def\vzero{\boldsymbol{0}}

\def\vX{\boldsymbol{X}}

\def\vZ{\boldsymbol{Z}}
\def\vG{\boldsymbol{G}}

\def\vTheta{\boldsymbol{\Theta}}